\documentclass[12pt]{article}

\usepackage{verbatim,color,amssymb,bm}
\usepackage{amsmath,bbm}
\usepackage{amsthm}					
\usepackage{natbib}
\usepackage{multirow}
\usepackage{setspace}
\usepackage{adjustbox,rotating}
\usepackage[mathscr]{euscript}
\usepackage{fancyhdr}
\usepackage{enumitem}
\usepackage{graphicx,grffile}
\usepackage{geometry}
\usepackage{xcolor}[dvipsnames]
\usepackage{physics,xfrac}
\usepackage[colorlinks=true,citecolor=blue,linkcolor=blue]{hyperref}
\usepackage{caption,subcaption}
\usepackage{soul,cancel}
\usepackage[linesnumbered,ruled,vlined]{algorithm2e}
\usepackage{algpseudocode}
\graphicspath{{Figures/}}

\usepackage{multibib}
\newcites{latex}{References}

\usepackage{tabularx, booktabs}
\newcolumntype{Y}{>{\centering\arraybackslash}X}
\usepackage{lscape}

\usepackage{float}
\usepackage{lineno}

\newtheorem{remark}{Remark}

\newtheorem{prop}{Proposition}
\newtheorem{corollary}{Corollary}

\def\rR{\mathbb{R}}
\def\eE{\mathbb{E}}

\def\C{{\cal C}}

\def\O{{\cal O}}

\def\S{{\cal S}}

\def\wh{\widehat}
\def\wt{\widetilde}

\def\wbar{\overline}

\def\var{\hbox{var}}

\def\corr{\hbox{corr}}

\newcommand{\mn}{{\mathrm{N}}}
\def\Bern{\mathrm{Bernoulli}}

\def\Dir{\mathrm{Dir}}
\def\DP{\mathrm{DP}}

\def\Exp{\mathrm{Exp}}

\def\Ga{\mathrm{Ga}}

\def\GEM{\mathrm{GEM}}

\def\NIG{\mathrm{NIG}}

\def\Unif{\mathrm{Unif}}
\def\Mult{\mathrm{Mult}}

\def\P_25_ICML{{\it Proceedings of the 25th international conference on Machine learning}}

\def\bse{\begin{eqnarray*}}
\def\ese{\end{eqnarray*}}
\def\be{\begin{eqnarray}}
\def\ee{\end{eqnarray}}
\def\bq{\begin{equation}}
\def\eq{\end{equation}}

\def\wh{\widehat}

\def\trans{^{\rm T}}

\def\th{^{th}}
\def\bone{{\mathbf 1}}
\def\Data{{\hbox{Data}}}

\def\bb{\bm{b}}

\def\bc{\bm{c}}

\def\b1e{\bm{e}}
\def\b1f{\bm{f}}
\def\bF{{ F}}

\def\bg{\bm{g}}

\def\bI{\bm{I}}

\def\bu{\bm{u}}
\def\bU{\bm{U}}

\def\bx{\bm{x}}
\def\bX{\bm{X}}
\def\bXs{\bm{X^\star}}

\def\bY{\bm{Y}}

\def\bzero{\bm{0}}

\def\simind{\stackrel{\mathrm{\scriptsize{ind}}}{\sim}}
\def\simiid{\stackrel{\mathrm{\scriptsize{iid}}}{\sim}}

\newcommand{\bmu}{\bm{\mu}}

\newcommand{\bvarphi}{\bm{\varphi}}

\newcommand{\bpi}{\bm{\pi}}

\newcommand{\bxi}{\bm{\xi}}

\newcommand{\btheta}{\bm{\theta}}
\newcommand{\thb}{\bm{\theta}}

\newcommand{\bbeta}{\bm{\beta}}

\newcommand{\bzeta}{\bm{\zeta}}

\newcommand{\bomega}{\bm{\omega}}

\newcommand{\de}{\mathrm{d}}
\newcommand{\brho}{\bm{\rho}}
\newcommand{\brhon}{\brho_{n}}

\newcommand{\kn}{k{(n_{2})}}
\newcommand{\wi}{w_i} 
\newcommand{\tilkn}{\kn} 
\newcommand{\deltah}{\wh{\delta}}
\newcommand{\zetat}{\wt{\bzeta}}
\newcommand{\omegat}{\wt{\bomega}}
\newcommand{\tht}{\wt{\btheta}}
\newcommand{\ct}{\wt{\bc}}
\newcommand{\zetas}{\bm{\zeta^\star}}

\newcommand{\bXon}{\bX_{1}}
\newcommand{\bXtw}{\bX_{2}}
\newcommand{\bcon}{\bc_{1}}
\newcommand{\bctw}{\bc_{2}}
\newcommand{\bpion}{\bpi_{1}}
\newcommand{\bpitw}{\bpi_{2}}
\newcommand{\ns}{n_{s}}
\newcommand{\gh}{\widehat{g}}
\newcommand{\mm}{^{(m)}}
\def\boldq{{q}}

\newcommand{\h}{h}
\renewcommand{\H}{H}
\newcommand{\HR}{\mathrm{HR}}

\renewcommand\footnoterule{\kern-3pt \hrule \textwidth 2in \kern 2.6pt}

\usepackage[normalem]{ulem}
\newcommand{\removen}[1] { {\color{blue} \sout{#1}}}
\newcommand{\changen}[1]{{\color{blue} #1}}

\def\boxit#1{\vbox{\hrule\hbox{\vrule\kern6pt \vbox{\kern6pt \textcolor{blue}{#1}\kern6pt}\kern6pt\vrule}\hrule}}

\def\authorfootnote#1{{\let\thefootnote\relax\footnotetext{#1}}}

\newcommand{\ech}{\color{black}\rm  }    

\allowdisplaybreaks

\begin{document}
\thispagestyle{empty}
\baselineskip=28pt

\newcommand{\papertitle}{\LARGE{\bf Bayesian Nonparametric Common Atoms \\
		\vskip-1ex
		Regression  for Generating \\
		\vskip-1ex
		Synthetic Controls in Clinical Trials}}
\begin{center}
\papertitle
		 \end{center}
\baselineskip=12pt
\vskip 12pt 

\newcommand{\authors}{\begin{center}
		Noirrit Kiran Chandra$^{a}$ (noirrit.chandra@utdallas.edu)\\
		Abhra Sarkar$^{b}$ (abhra.sarkar@utexas.edu)\\
		John F. de Groot$^{c}$ (john.degroot@ucsf.edu)\\
		Ying Yuan$^{d}$ (yyuan@mdanderson.org)\\
		Peter M\"uller$^{b,e}$ (pmueller@math.utexas.edu)\\
		
		\vskip 7mm
		$^{a}$Department of Mathematical Sciences, \\
		The University of Texas at Dallas, TX, USA 
		\vskip 5pt
		$^{b}$Department of Statistics and Data Sciences, \\
		The University of Texas at Austin, TX, USA 
		\vskip 5pt
		$^{c}$Department of Neurological Surgery,\\
		University of California San Francisco, CA, USA\\
		\vskip 5pt
		$^{d}$Department of Biostatistics,\\
		The University of Texas MD Anderson Cancer Center,
		Houston, TX, USA\\
		\vskip 5pt 
		$^{e}$Department of Mathematics, \\
		The University of Texas at Austin, TX, USA 
\end{center}}
\authors

\begin{center}
{\Large{\bf Abstract}} 
\end{center}
\baselineskip=14pt
The availability of electronic health records (EHR) has opened opportunities to supplement increasingly expensive and difficult to
carry out randomized controlled trials (RCT) with evidence from readily available real world data. 
In this paper, we use EHR data to construct synthetic control arms for treatment-only single arm trials.
We propose a novel nonparametric Bayesian common atoms mixture model
that allows us to find equivalent population strata in the EHR and the treatment arm and then resample the EHR data to create equivalent patient populations under 
both the single arm trial and the resampled EHR.  
Resampling is implemented via a density-free importance sampling scheme. 
Using the synthetic control arm, inference for the treatment
effect can then be carried out using any method available for RCTs.
Alternatively the proposed nonparametric Bayesian model allows straightforward model-based inference. 
In simulation experiments, the proposed method 
exhibits higher power than alternative methods in detecting treatment
effects,  specifically for non-linear  response functions. 
We apply the method to supplement single arm treatment-only glioblastoma studies
with a synthetic control arm based on historical trials.

\vskip 10pt 
\baselineskip=12pt
\noindent\underline{\bf Key Words}: 
common atoms mixture,
glioblastoma,
importance sampling,
mixtures,
real world data,
single-arm trials.

\par\medskip\noindent
\underline{\bf Short/Running Title}: Common Atoms Mixture Model for Synthetic Controls

\par\medskip\noindent
\underline{\bf Corresponding Author}: Noirrit Kiran Chandra (noirrit.chandra@utdallas.edu)

\pagenumbering{arabic}
\setcounter{page}{0}
\newlength{\gnat}
\setlength{\gnat}{26pt}
\baselineskip=\gnat
\vskip 4cm

\newpage
\section{Introduction}
\label{sec:intro}

We introduce a novel Bayesian nonparametric 
regression model to construct synthetic control arms from external real world
data (RWD) to supplement single arm treatment-only trials.  The use of
common atoms across multiple random probability measures is a critical
feature of the proposed construction.  Models with similar features
have been used before in the literature, including
\citet{denti_common,camerlenghi2019nDP,rodriguez2008nDP,teh_HDP}.

Randomized controlled trials (RCT) are the gold standard in
evidence-based evaluation of new treatments.  RCTs are, however,
increasingly associated with bottlenecks involving volunteer
recruitment, patient 
truancy and adverse events \citep{rct_challenges} and hence are often
very time consuming, expensive and laborious.  
This is of particular concern for rare diseases, such as glioblastoma (GBM). 
With digitization of
health records and other advances in medical informatics, new data
sources are becoming available that can supplement RCTs.  For example,
relevant information on a control treatment is often
available from completed RCTs, electronic health record data,
insurance claims data or patient registries from hospitals
\citep{ehr}.  Such external data, also referred to as RWD,
can augment or substitute the control group in the target
clinical trial \citep{singlearm}. This has led researchers to
consider the creation of synthetic control arms from RWD (see
\citealt{synthetic_control_review} for a review).  However, the
heterogeneity of RWD prohibits the direct use of patient
level data as a control arm, lest differences with the actual treatment
population with respect to patient profiles
bias inference on treatment effects
\citep{synthetic_control_challenges}.  Many existing methods adjust
for the lack of randomization in treatment assignments by correcting
the bias in the response model and hence can be sensitive to the
specification of the treatment assignment as well as the response
model as we discuss below.
%
In this article, we take a fundamentally different approach
by resampling the RWD to
construct a cohort equivalent to the 
treatment arm in terms of their covariate profiles
which can then serve
as the (synthetic) control arm.

There is a fast growing literature on the problem of incorporating
RWD in clinical trials.
Traditional meta-analytic approaches aim to combine
information across studies to construct comparisons of treatments
\citep{meta_analysis3}.
Power prior
\citep{prevost2000,chen_power_prior}, commensurate prior
\citep{hobbs2011} and elastic prior \citep{jiang2021}
 constructions 
try to incorporate information from historical data
by way of informative prior models.
However, these approaches may be inadequate when the RWD population is
considerably more heterogeneous than the experimental arm; see \citet{mueller_cam23} for a review. 

Many methods to incorporate RWD in trial design and data analysis
are based on propensity scores (PSs), defined as the conditional probabilities of treatment assignment
given covariates.
In the context of incorporating external data, investigators
often use PSs
for a patient being selected into the 
current trial versus the external data; 
in case of supplementing a single arm treatment-only trial, 
the PSs are identical to treatment assignments.
\citet{rosenbaum83propscore}
showed that an unbiased estimate of the average treatment effect can
be obtained by PS adjustments.
%
Most PS-based methods can be broadly classified to be based on 
matching, stratification, weighting, or regression. 
Matching is used to achieve covariate balance across different
arms. 
However, matching
PSs do not generally imply matching covariates
\citep{propensity_limitation2}.  
Stratification splits the
data into strata with respect to PSs
and calculates an average treatment effect as a weighted
average of within-stratum estimates \citep{propensity_power_prior,
propensity_CL, propensity_multiple}.  
PS-stratification
may be sensitive to the definition of the strata and
weight-based estimators may be sensitive to the misspecification of
the PS model \citep{zhao2004}. 
Regression adjustments, that use the
PS as a regressor for the outcome, address these
issues \citep{rosenbaum83propscore} but the estimates may again be
biased if the regression model is misspecified
\citep{vansteelandt2014}.  
Bayesian nonparametric models that
avoid a particular parametric family or structure,  such as linearity, of the
regression relationship have thus also been proposed \citep{wang2019}.
Nevertheless, consolidated unidimensional PSs can be inadequate in matching multivariate covariates from multiple studies 
\citep{propensity_limitation1,propensity_limitation2}.
Additionally, these methods often do not efficiently use all available
data by dropping unmatched data. 
Finally, some other methods \citep{hasegawa2017,li_integrated_framework}, 
although not specifically designed to create synthetic controls,
also integrate multiple studies using the covariate distributions.

In this article, we develop an alternative approach based on 
Bayesian nonparametric (BNP) mixture models. 
Mixture models imply a random partition of experimental units linked
to different atoms in the mixture \citep{dahl_2006}.  
We exploit this property to propose a BNP common
atoms mixture (CAM) model  to  introduce matched clusters of
patients in a treatment-only trial data set and a (typically much larger)
RWD.
We show how such matched clusters allow a \textit{density free
importance resampling scheme} to generate a subpopulation of the RWD such
that the distribution of covariates in the subpopulation can be 
considered to be equivalent to the single-arm trial.
That is, the patients in a matching RWD cluster can be considered 
\textit{digital clones} of patients in a matching cluster 
in the single-arm trial.

The proposed CAM model
 allows, among other things, the following two 
alternatives for inference on treatment effects.
Having established equivalent patient populations, 
inference can in principle proceed as if treatment had been
assigned at random, using inference for RCTs.
Alternatively, we propose model-based inference using an extension of
the CAM model with a sampling model for the outcome.
While both alternatives are based on the same underlying CAM
model, 
we prefer the model-based inference on treatment effect
as a more explicit and principled approach.

The proposed CAM model builds on related BNP models in the literature,
including the 
hierarchical Dirichlet process (DP) \citep{teh_HDP} which allows for
information sharing across multiple groups through common atoms,  the
nested DP \citep{rodriguez2008nDP} which can identify distributional
clusters, and 
\citet{camerlenghi2019nDP} who proposed a latent mixture of shared and
idiosyncratic processes across the sub-models.  \citet{denti_common}
proposed a CAM model for the analysis of nested datasets where the
distributions of the units differ only over a small fraction of the
observations sampled from each unit.
In contrast to these constructions, the CAM model proposed here introduces more
structure as needed in our application by setting up two
nonparametric Bayesian mixture models with shared atoms and
constraints on the implied clusters.

The rest of this paper is organized as follows.
Section \ref{sec:GBM_description} describes the glioblastoma study
that motivated this work.  
Section \ref{subsec:CAM_cov} introduces the
proposed common atoms mixture model on the covariates and how it can handle
variable dimensional covariates of different data types; Section
\ref{subsec:IS_reweighing} introduces a novel density-free importance
resampling scheme to achieve equivalent populations; 
and Section \ref{subsec:CA-PPMx} 
discusses the general common atoms
regression model, a flexible mixture of lognormals for censored
survival outcomes and an easy to use graphical tool for model
validation. 
In Section \ref{sec:regression_with_CAM}, we discuss two alternative strategies for 
inference on treatment effects.  
Section \ref{sec:sim_study} presents simulation studies.  
Section \ref{sec:GBM} shows results for the motivating GBM data. 
 Section \ref{sec:discussion} concludes with final remarks.  
Below, in Table \ref{tab:acronyms}, we list the many acronyms used in the paper for
easy reference.


\begin{table}[!ht] \centering
  \caption{List of acronyms}
  \vskip-1.75ex
  \label{tab:acronyms}
  \hspace*{-4ex}\resizebox{.95\textwidth}{!}{ \begin{minipage}[b]{.48\linewidth}
   \begin{tabular}{|p{.2\linewidth}|p{.75\linewidth}|}
      \hline Acronym&\multicolumn{1}{c|}{Full forms} \\
      \hline AUC& area under the receiver operating characteristic curve\\
      BART& Bayesian additive regression tree\\
      BNP& Bayesian nonparametric\\
      CAM& common atoms mixture\\
      CA-PPMx& common atoms PPMx\\      
      \hline
    \end{tabular}
  \end{minipage} \hspace*{.016\linewidth}
\begin{minipage}[b]{.48\linewidth}  	
    \begin{tabular}{|p{.2\linewidth}|p{.8\linewidth}|} \hline
      Acronym& \multicolumn{1}{c|}{Full forms} \\
      \hline 
      DP& Dirichlet process \\
      GBM & glioblastoma \\
      IS& importance sampling\\
      PPMx& product partition model with regression on covariates\\
      PS& propensity score\\
      RCT& randomized controlled trial \\
      RWD& real-world data\\
      \hline
    \end{tabular}
  \end{minipage}}
\end{table}

\section{Motivating Application in Glioblastoma}
\label{sec:GBM_description}
Our motivating application arises from a GBM data science project at
MD Anderson Cancer Center.  
GBM is a devastating disease with the
average life expectancy of less than 12 months in the general population
\citep{ostrom2016}.  
Despite decades of intensive clinical research,
the progress in developing an effective treatment for GBM lags behind
that of other cancers \citep{aldape2019challenges}.  
In the last 30 years, only two drugs (carmustine wafers and temozolomide) have been
approved by the Federal Drug Administration (FDA) for patients with newly diagnosed GBM \citep{gbm_drugs}. 
These drugs extend median survival by
less than three months and neither offers a potential for cure.  
One major cause of the high failure rate of the drug development for GBM
is suboptimal design of phase II trials, in particular,
the lack of a control arm in many studies
 \citep{gbm_trials}.  
A review of phase I/II GBM
trials from 1980 to 2013 found that only 20 (5\%) were randomized
compared to 365 (95\%) single-arm trials \citep{grossman2016}.
Reasons for the dominance of single-arm trials include the
small number of GBM patients available for clinical trials and investigator's desire to speed up drug development and reduce trial costs.  
GBM is a rare disease by the definition of the Orphan Drug Act \citep{fda2020}.  
Unfortunately, the high heterogeneity of GBM patients makes single-arm trials highly susceptible to bias,
contributing to the fact that almost all phase II trials showing
promising treatment effects failed in phase III RCTs
\citep{mandel2017}.  
The objective of the GBM data science project is to address this pressing issue by leveraging historical data collected at the MD Anderson Cancer Center.  
The overarching goal is to develop a platform for future single-arm clinical trials in GBM, with synthetic controls constructed from the historical database to enhance the evaluation and screening of new drugs. 
Working towards this goal, we describe here a method to create synthetic
controls, as the engine of the platform, for future trials.

We work with a database that comprises records from 339
highly clinically and molecularly annotated GBM patients treated
at MD Anderson over more than 10 years.
Once the system is set, the database is expected to be continuously
updated with new patient data collected at MD Anderson Cancer Center
and potentially also be combined with the data from other institutions.

After discarding variables with minimal variability across patients and relying on clinical judgment, we identified
11 clinically important 
 categorical  covariates. 
These covariates are commonly considered as prognostic factors in GBM
treatments \citep{gbm_covs, trippa_insight} and are briefly described
in Table \ref{tab:GBM_covs}.

Figure \ref{fig:GBM_heatplot} shows the categorical covariates in the
historical database and a future treatment-only study which we
elaborate in Section \ref{sec:GBM}. 
Figure \ref{fig:GBM_freq_plots} in the supplementary materials 
highlights the lack of randomization in the two populations.


\begin{table}[!h]
  \centering
  \caption{Description of the covariates in the GBM data.}
  \vskip-1.75ex
  \label{tab:GBM_covs}
  \resizebox{.9\textwidth}{!}{
    \begin{tabular} {|p{.19\linewidth}|p{.8\linewidth}|}
      \hline
      \multicolumn{1}{|c|}{Covariate}&\multicolumn{1}{c|}{Description}  \\
      \hline
      Age    & dichotomized at 55 years\\
      KPS    & Karnofsky performance score, categorized into
               three classes: \\ & ``$\leq 60$'', ``$(60,80]$'' and
                                   ``$>80$''\\ 
      RT Dose& radiation therapy dose: dichotomized at 50 Gray\\
      SOC    & received standard-of-care (concurrent
               radiation therapy and temozolomide): Yes/No   \\[4pt]
      CT     & participation in a therapeutic trial: Yes/No  \\
      MGMT   & status of MGMT ($O^{6}$-methylguanine-DNA
               methyltransferase) gene: methylated (M),
               unmethylated (UM) or uninterpretable (UI)\\[4pt]
      ATRX   & loss of the ATRX chromatin remodeler gene: Yes/No \\
      Gender & gender\\
      EOR    & extent of tumor resection:
               ``total", ``subtotal" or  ``laser interstitial thermal therapy" \citep[LITT,][]{litt}\\
      Histologic grade & grade of astrocytoma:
                         IV (GBM) (most cases), or\\ & I-III (low-grade or anaplastic) (few) \\[4pt]
      Surgery reason  & ``therapeutic" or ``other" (relapse)\\
      \hline
    \end{tabular}}	
\end{table}

\begin{figure}[!ht]
	\centering
	\includegraphics[width=.8\linewidth]{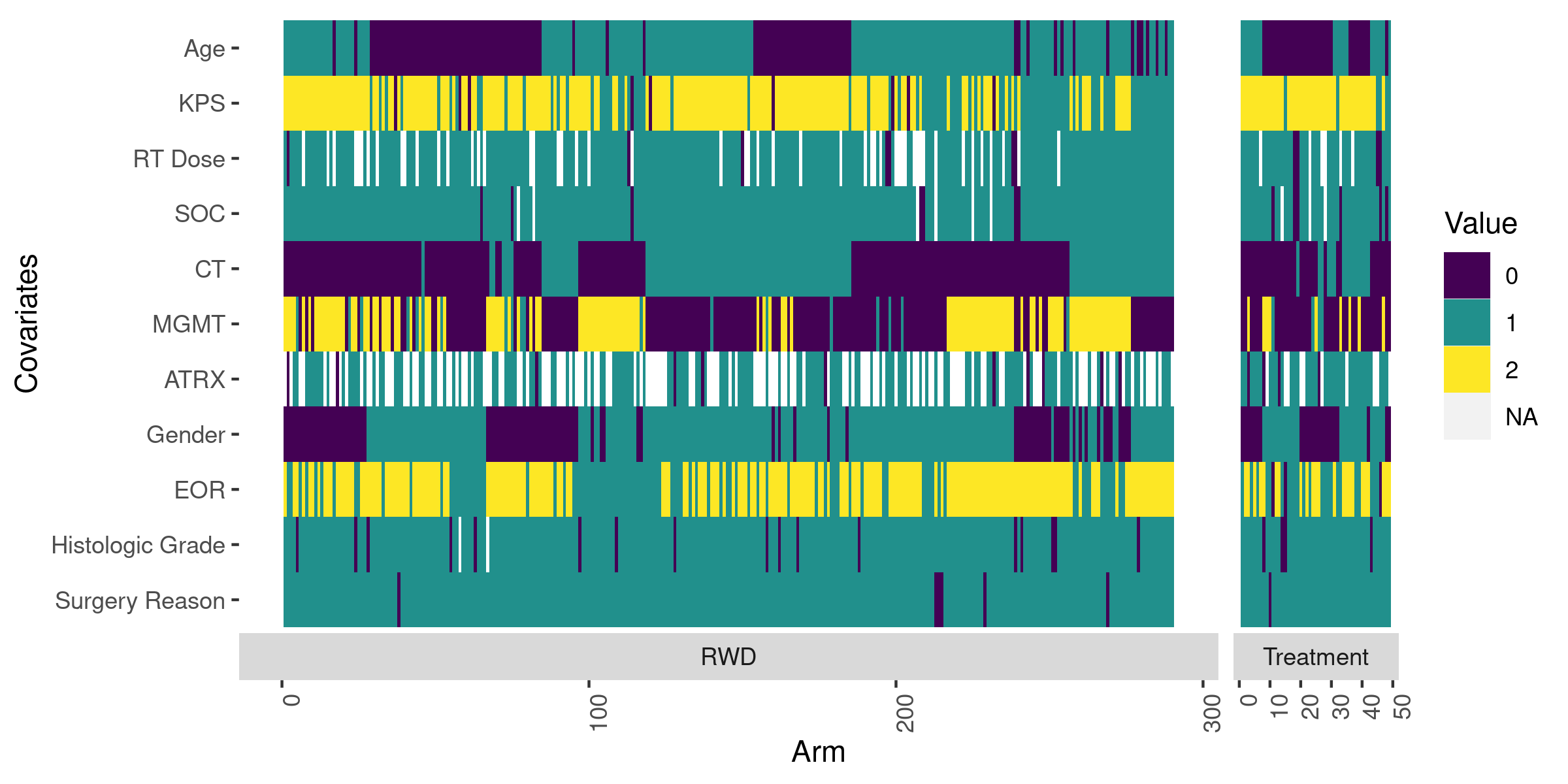}
	\vskip-1.75ex
	\caption{Glioblastoma dataset of 11 baseline categorical
		covariates with missing entries in the two treatment arms.
		The left block shows the historical patients. The
		(smaller) right block shows a hypothetical future trial.}
	\label{fig:GBM_heatplot}
\end{figure}
\vspace*{-2ex}
\section{Common Atoms Mixture Model}
\label{sec:CAM_regression}
We first introduce a model for matching patients with respect to their covariate profiles across different treatment
arms and then an extension of the model to also include outcomes.
Later we will introduce two alternative methods for inference on
treatment effects that build on this model. 

\subsection{Common Atoms Mixture Model on the Covariates}
\label{subsec:CAM_cov}
Suppose we have $S$ datasets $(\bX_{s,i},Y_{s,i})$, $s=1,2,\dots,S$,
comprising $p$-dimensional covariate vectors
$\bX_{s,i}=(X_{s,i,1},\dots,X_{s,i,p})\trans$ and corresponding
responses $Y_{s,i}$ associated with patients $i=1,\dots,n_{s}$. 
In this article, we assume the responses to be univariate.
Let $s=1$ refer to the arm for the (new) experimental therapy, and
$s=2,\dots,S$ denote the RWD datasets.  
Focusing on the motivating GBM
application, we elaborate the model for $S=2$ with a single RWD set.
When we have multiple historical datasets, i.e., when $S>2$, 
we would simply merge them and consider the merged data set to be a single
RWD with increased heterogeneity  
as illustrated in Section \ref{sm subsec:mult_hist} of the
supplementary materials. 
For a valid evaluation of treatment effects, it is
then important to verify equivalent patient populations, i.e., matching
the distributions of $\bX_{s,i}$ under $s=1$ versus $s=2$,
or to otherwise adjust for any detected differences
\citep{synthetic_control_challenges}.   
As the RWD population can be from a variety of sources, such data are
typically more heterogeneous than the patient population in the
ongoing trial.
We develop a novel BNP CAM model with this specific feature to model the two
distributions. 
 The proposed CAM model gives rise to a random partition of
similar $\bX_{1,i}$ and a matching partition of $\bX_{2,i}$.
Clusters under the latter partition can be considered
digital clones of the matching clusters of the earlier partition. 

We first construct the model for covariates $\bX_{2,i}$ in the RWD. 
Let ${\zetat}=\{\zetat_{j} \}_{j=1}^{\infty}$
and $\bpitw=\{ \pi_{2,j} \}_{j=1}^{\infty}$ 
denote cluster-specific parameters and weights, respectively. 
We let
\vspace*{ -2ex}\\
\begin{equation}
  \bX_{2,i}\mid \zetat,\bpitw \simiid
  \overbrace{\textstyle\sum_{j=1}^{\infty} \pi_{2,j} \boldq(\bX_{2,i}\mid \zetat_{j}) }^{\bF_{2} {(\bX_{2,i} \mid \bpitw,\zetat )}},\quad
	\zetat_{j}\mid \bxi \simiid G_{0}(\zetat_{j} \mid \bxi),\quad
	\bpitw\sim \GEM(\alpha_{2}). \label{eq:RWD_model}
\end{equation}
\vspace*{ -6ex}\\
\noindent
Here $\boldq(\cdot \mid \zetat_{j})$ is a
suitably chosen kernel with parameter $\zetat_{j}$, $G_{0}(\cdot
\mid \bxi)$ is a prior distribution for the $\zetat_{j}$'s,
and $\GEM(\alpha)$ is a stick-breaking prior on the mixture weights
corresponding to a DP with mass parameter $\alpha>0$
\citep{ sethuraman_stickbreaking}.
Let $G = \sum_{j=1}^{\infty} \pi_{2,j} \delta_{\zetat_{j}}(\cdot)$ denote a discrete
probability measure with atoms at the $\zetat_{j}$'s. 
An equivalent hierarchical model 
representation of \eqref{eq:RWD_model} is
\vskip -8ex
\begin{equation}
	\bX_{2,i}\mid \bzeta_{i} \simiid \boldq(\bX_{2,i}\mid \bzeta_{i}),\quad
	\bzeta_{i}\mid G  \simiid G, \quad 
	G \mid \alpha_{2}, \bxi \sim \DP\{\alpha_{2}, G_{0}(\cdot\mid
        \bxi) \},
\label{eq:RWDhier}
\end{equation}
\vskip -3ex
\noindent where $\DP(\alpha, G_{0})$ is a DP with base
measure $G_{0}$ and concentration parameter $\alpha$
\citep{ferguson73}.
The discrete nature of the DP random measure $G$ gives rise to
possible ties between the $\bzeta_{i}$'s, which define the desired
clusters.
For later reference we define notations for these ties and clusters.
Let $\zetas = \{\zetas_{j},\; j=1,\ldots,\tilkn\}$ denote the distinct 
values in $\{\bzeta_{i};i=1,\dots,n_{2}\}$,  let
$c_{2,i}=j$ if $\bzeta_i=\zetas_j$ denote cluster membership
indicators defining clusters $C_j=\{i:\; \bzeta_i=\zetas_j\}$.
We assume the distribution of $\bX_{1,i}$ be
a mixture with the same kernel $\boldq$ and the same atoms $\zetas$,
\vspace*{ -2ex}\\
\begin{equation}
  \bX_{1,i}\mid \zetas \simiid
  \overbrace{\textstyle{\sum_{j=1}^{\tilkn} \pi_{1,j} \boldq(\bX_{1,i}\mid \zetas_{j})}}^{\bF_{1} {(\bX_{1,i} \mid \bpion,\zetas )}},\quad
  \textstyle{\bpion  \sim
  \Dir\left\{\frac{{\alpha_{1}}}{\tilkn},\dots,\frac{\alpha_{1}}{\tilkn}\right\}},\label{eq:experimental_model} 
\end{equation}
\vspace*{ -6ex}\\
\noindent where $\bpion=(\pi_{1,1},\dots,\pi_{1,\tilkn})$,
$\Dir(a_{1},\dots,a_{r})$ indicates an $r$-dimensional Dirichlet distribution
with parameters $a_{1},\dots,a_{r}$, 
and $\alpha_{1}>0$ is a concentration parameter. 
Note that model \eqref{eq:experimental_model} is defined conditionally
on
\eqref{eq:RWD_model} and $\zetas$ 
such that $F_1$ and $F_2$ 
share the same set of atoms.
Importantly, the construction avoids the imputation of clusters
(strata) with only $\bX_{1,i}$'s.
There is always a corresponding (non-empty) cluster for the $\bX_{2,i}$'s from the RWD.
This is important for the upcoming constructions. 
The motivation here is that, 
owing to the bigger size of the RWD compared to the trial arm, 
$\bXtw$ can be expected to exhibit greater heterogeneity
than $\bXon$ 
(see, e.g., the right panel in Figure \ref{fig:common_atoms_motivate}).
\begin{figure}[!ht]
  \centering
  \includegraphics[trim={0.5cm 0.4cm 0.5cm 0.25cm},clip,height=.15\paperheight] {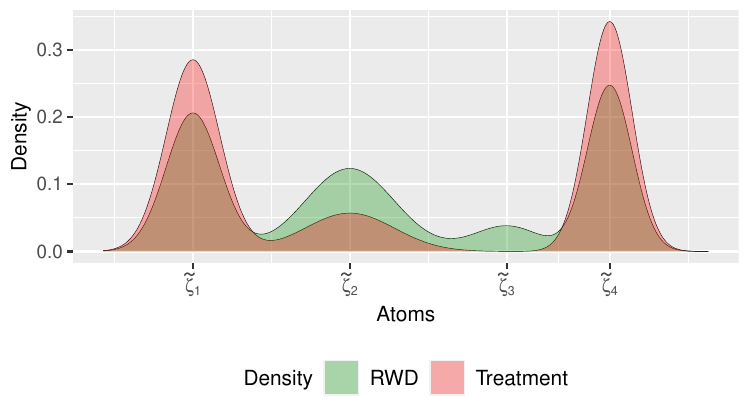}
  \includegraphics[trim={0.5cm 0.375cm 0.35cm 0.1cm},clip,height=.15\paperheight] {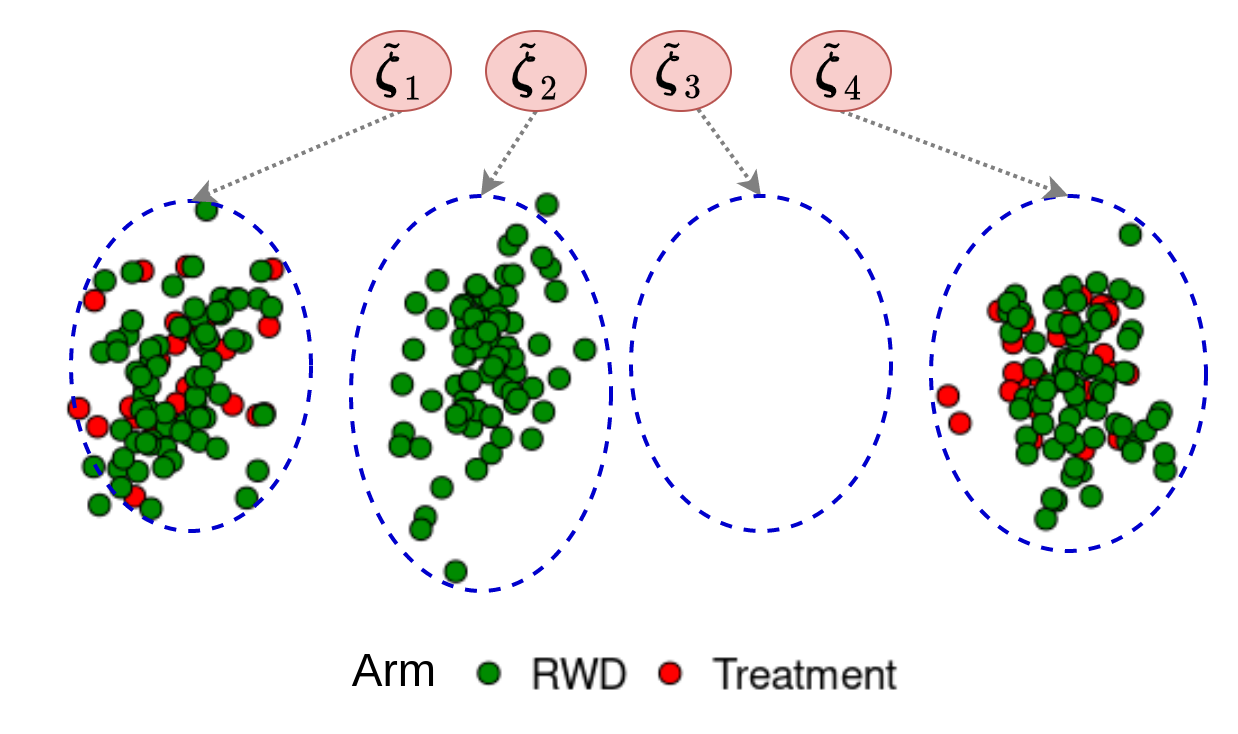}	
  \caption{An illustration of the CAM model: In the 
    generative model, there are a total of four atoms $\zetat_{1:4}$
    shared between RWD and the treatment arm (left panel).  
    Despite having  positive weight, the atom $\zetat_{3}$ is not associated with any
    sample from the RWD (right panel) and hence the density of the treatment arm is also allowed to be supported only on the remaining
    non-empty clusters $\zetat_{1}, \zetat_{2},\zetat_{4}$ of the RWD.
    The atom $\zetat_{2}$ is linked with only
    the RWD (right panel). 
    A cluster for the treatment arm alone is however not permissible.
  }
  \label{fig:common_atoms_motivate}
\end{figure}

In summary, 
we define $\bF_{1}(\bX \mid \bpion,\zetas )=\sum_{j=1}^{\kn} \pi_{1,j}
\boldq(\bX\mid \zetas_{j})$ and $\bF_{2}(\bX \mid \bpitw,\zetat
)=\sum_{j=1}^{\infty} \pi_{2,j} \boldq(\bX\mid \zetat_{j})$,
with the prior on atoms and weights as discussed.
\label{define_the_Fs}
Figure \ref{fig:common_atoms_motivate} shows a stylized representation
of the generative process of the proposed CAM model.
Notice that here atom $\zetat_{3}$ is not linked with any $\bXtw$ observation and hence $\kn=3$.
Accordingly, $\bF_{1}(\cdot \mid \bpion,\zetas )$ is a mixture of
three components.
Finally, no observation from $\bXon$ is linked to $\zetat_{2}$.
The $\bX_{2,i}$'s linked to $\zetat_{1}$ and $\zetat_{4}$ can be regarded as digital clones of the $\bX_{1,i}$'s linked to the same atoms. 

The described CAM model is different from existing BNP mixture models.
In \eqref{eq:RWD_model}-\eqref{eq:experimental_model}, the atoms linked to
$\bXon$ are always a subset of those atoms that are linked to $\bXtw$,
which is not  naturally the case for the hierarchical DP model
\citep{teh_HDP}.
Also, unlike the nested DP
\citep{rodriguez2008nDP} and the common atoms nested DP
\citep{denti_common} models, there is no notion of
clustering distributions.
That is, $p\{\bF_{1}  {(\bX \mid \bpion,\zetas
)}=\bF_{2} {(\bX \mid \bpitw,\zetat )} \}=0$ a priori. Instead,
the intention here is to cluster similar covariate values across the
datasets.

Regarding the concentration parameters $\alpha_{s}$,
	we assume $\log \alpha_{s}\sim
\mn(\mu_{\alpha},\sigma_{\alpha}^{2})$ for $s=1,2$.
\citet{ascolani2022clustering} showed that a hyper-prior on the
concentration parameters can solve the problem of inconsistency of DP
mixtures \citep{miller_DP_inconsistency}.

\paragraph*{Handling mixed data types and missing values:} 
An appealing feature of the proposed CAM model over existing approaches is
the easy use of covariates of different data-types and missing values.  
Covariates in RCTs often comprise different
data-types including continuous, discrete and categorical variables. 
Missing values are also quite common.
For example, in Figure \ref{fig:GBM_heatplot}, 
there are a large number of missing values for the ATRX gene which has
only recently been identified as a therapeutic target for glioma
\citep{atrx_gbm} and was therefore not commonly recorded before.

Many existing methods for handling missing data 
rely on imputation \citep{propensity_missingdata}, 
possibly at the expense of an
additional layer of prediction errors.
Alternatively, data records with missing variables may be dropped altogether, resulting in a reduced sample size.

Assuming missingness completely at random,
the proposed CAM model avoids these issues by accommodating variable
dimensional covariates in a principled manner by considering a
separate univariate kernel for each covariate.
Note that a mixture with independent kernels
can still accommodate marginal dependence between the covariates
\citep[Section 7.2.2, pp
175]{ghosal_book}. \label{pg:dependent_marginal}
Specifically, let $\O_{s,i}=\{j: X_{s,i,j} \mbox{ is recorded}\}$
denotes the set of observed covariates for patient $i$ in dataset $s$.
We use independent kernels
\vspace*{ -5ex}\\
\begin{equation}
\textstyle{	\boldq(\bX_{s,i} \mid \zetas_{j})= \prod_{\ell\in \O_{s,i}} q_{\ell}(X_{s,i,\ell}\mid \zetas_{j,\ell}), \quad G_{0}(\zetas_{j} \mid \bxi)= \prod_{\ell=1}^{p} g_{0,\ell}( \zetas_{j,\ell} \mid \bxi)},
	\label{eq:model_missingdata}
\end{equation} 
\vspace*{ -6ex}\\
\noindent where $q_{\ell}( \cdot \mid \zetas_{j,\ell})$ is a
univariate kernel corresponding to the $\ell\th$ covariate with
parameters $\zetas_{j,\ell}$ and $g_{0,\ell}( \zetas_{j,\ell} \mid
\bxi)$ is a prior on $\zetas_{j,\ell}$ with hyper-parameters $\bxi$.
The likelihood function of $\bX_{s,i}$ is then computed on the basis
of only the observed values.
The kernel $q_{\ell}$ is chosen to accommodate the data-type of the $\ell\th$
covariate. 
The model allows co-clustering of $\bX_{s,i}$ with some missing
variables and another fully observed $\bX_{s,i'}$; see Section \ref{sm
  sec:ppmx_missing} of the supplementary materials for additional details.	 
Missingness patterns other than completely at random can be handled by
introducing additional hierarchy in the model, see, e.g.,
\cite{linero2018_missing} for a review.\label{referee1_concern}  

\subsection{Density-free Importance Resampling of RWD} 
\label{subsec:IS_reweighing}
Building on the fitted CAM for covariates,
we propose an importance resampling method to create a subpopulation
of $\bXtw$ that can be considered to be equivalent to $\bXon$
(see below for a definition of equivalence that is being used here).
Under the assumption of no unmeasured confounders, the
$\bX_{2,i}$'s in the sampled (or weighted) subpopulation can be assumed
to follow the same distribution as $\bX_{1,i}$, and be
considered digital clones of the $\bX_{1,i}$.
With such equivalent populations, in principle, any desired method
for randomized clinical trials 
can subsequently be used to carry out inference on treatment effects.
Such focus on equivalent populations follows recent recommendations by the FDA  
\citep{fdaregulation2021}.

Recall that $\bF_{s}$ denotes the mixture model for $\bX_{s,i}$,
$s=1,2$, under \eqref{eq:RWD_model} and
\eqref{eq:experimental_model}, respectively.
We define equivalent populations as a subset (possibly all) of
$\bXtw$ together with a set of weights such that expectation of
any function of interest $g(\bX_{1,i})$ under $\bX_{1,i} \sim \bF_{1}$ can be evaluated as
a (weighted) Monte Carlo average using these $\bXtw$ (and the weights).
Here we assume that all stated expectations exist and that
the order of taking expectations and limits can be switched.

Recall that $\bF_{2} {(\cdot \mid \bpitw,\zetat )} = \sum_{j=1}^{\infty} \pi_{2,j} \boldq(\cdot \mid \zetat_{j})$.
Alternatively, the joint model of $(\bXtw,\bctw)$ can be expressed as
$\bF_{2}^{H}(\bX_{2,i},c_{2,i}\mid \bpitw,\zetat)=\boldq(\bX_{2,i}
\mid \zetat_{c_{2,i}}) \pi_{2,c_{2,i}}$. 
For easier housekeeping, we assume
$\zetas_{j} = \zetat_{j}$ for $j=1,\dots,\tilkn$, i.e., the first
$\tilkn$ atoms are linked with the  $\bX_{2,i}$'s. 
Accordingly, we let $\bF_{1} {(\cdot \mid \bpion,\zetat )} = \sum_{j=1}^{\tilkn} \pi_{1,j} \boldq(\cdot \mid \zetat_{j})$ using the same first $\tilkn$ atoms observed in the $\bXtw$ population.
This is the exact construction of
\eqref{eq:RWD_model} and \eqref{eq:experimental_model}. 
For an {equivalent population}, we require weights
$\wi$ attached to
 {$(\bX_{2,i} {,c_{2,i}})$} (using $\wi=0$ to drop samples)
such that: 
\vspace*{ -8ex}\\
\bse
	\eE_{\bF_{1} {(\cdot \mid \bpion,\zetat )} } \left\{g(\bX_{1,i}) \right\} = \eE_{\bF_{2}^{H} (\cdot \mid \bpitw,\zetat )} \{\gh(\bXtw {,\bctw})\}
	\mbox{ with }
	\gh(\bXtw {,\bctw}) = \sum_{i=1}^{n_{2}}  \wi\, g(\bX_{2,i} {,c_{2,i}}).
\ese
\vspace*{-8ex}\\
\noindent
The weights $\wi$ are functions of $c_{2,i}$ and $\pi_{1,j}$ as
follows. 
Define $n_{2,j}=\abs{ C_{2,j}}$, the cardinality of the earlier
introduced clusters $C_{2,j}$.  
Then $\frac{1}{n_{2,j}}
\sum_{i\in C_{2,j}} g(\bX_{2,i})$ is an unbiased estimator of
$\eE_{\boldq(\cdot \mid \zetat_{j})}g(\bX)$  
and
\vspace*{-8ex}\\
\be
	\textstyle{\gh =
	\sum_{j} \pi_{1,j} \left\{\sum_{i\in C_{2,j}} \frac{1}{n_{2,j}} g(\bX_{2,i})\right\}
	= \sum_{i=1}^{n_{2}} \frac{\pi_{1,c_{2,i}}}{n_{2,c_{2,i}}} g(\bX_{2,i})} \label{eq:ref_needs}
\ee
\vspace*{-8ex}\\
\noindent
is an unbiased estimator of $\eE_{\bF_{1} {(\cdot \mid \bpion,\zetat )}} g(\bX)$. 
We then recognize $ \pi_{1,c_{2,i}}/n_{2,c_{2,i}}$ as the ideal
weights. 
 {Since we only observe $\bX_{s}$ but not $\bpion$ and $\bctw$, we}  
replace $\pi_{1,c_{2,i}}/n_{2,c_{2,i}}$ in $\gh$ by a Monte Carlo
average under 
posterior MCMC simulation to get the desired equality simulation-exact
(i.e., in the limit as $n_{1}, n_{2}$ and the number of MCMC
simulations increases).   
Let $m=1,\dots,M$ index
the posterior sample and use $\pi_{1,j}\mm$, $n\mm_{2,j}$, etc. to
indicate parameter values in the $m\th$ sample. 
We use
\vspace*{-5ex}\\
\begin{equation}
	\textstyle \gh = \sum_i w_i g(\bX_{2,i}),\quad
	w_{i} \propto  \sum_{m=1}^{M} \pi\mm_{1,c\mm_{2,i}}/n\mm_{2,c\mm_{2,i}},
	\label{eq:imp_dist}
\end{equation}
\vspace*{-6ex}\\
\noindent with $w_{i}$ being the
\textit{importance sampling} weight for $\bX_{2,i}$.
The $\bX_{2,i}$'s can be resampled with these weights to
obtain the desired subpopulation with distribution $\bF_{1}(\bx \mid \bpion,\zetat)$
\citep{importance_resampling}.  
This resampled subpopulation
of $\bXtw$ can then be regarded as equivalent in distribution to
$\bXon$.
Algorithm \ref{algo:importance_resampling} summarizes the procedure. 
\begin{algorithm}[!ht]
	\caption{Density-free importance resampling of RWD and validation}
	\label{algo:importance_resampling}
	\SetAlgoLined
	\DontPrintSemicolon
	\textbf{Input} two data sets $\bXon$ and $\bXtw$.\;
	\textbf{Fit the CAM model} to the data using MCMC simulation.
	Let $\bpion^{(m)}$ and $\bctw^{(m)}$ be the $m\th$ MCMC sample of $\bpion$ and $\bctw$, respectively, 
	and $n^{(m)}_{2,j}$ be the size of cluster $C_{2,j}$ in the $m\th$ MCMC iteration for $m=1,\dots,M$.\; 
	\textbf{Calculate} importance sampling weights
	\hspace*{65pt}$w_{i} \propto \sum_{m=1}^{M}
        {\pi^{(m)}_{1,c^{(m)}_{2,i}}}/ {n^{(m)}_{2,c^{(m)}_{2,i}}}, 
        ~~~i=1,\dots,n_{2}$.\;
	\textbf{Resample} a subpopulation of size $n_{1}$ from $\bXtw$
        with importance resampling weights $w_{i}$ with
        replacement.\; 
	\textbf{Test for equivalence} of $\bXon$ and the resampled subpopulation of $\bXtw$ using a supervised classification algorithm (e.g., a BART as described in the text).	
\end{algorithm}	
\vspace*{-5pt}

To test the equivalence of the two populations, 
we use a Bayesian additive regression tree \citep[BART,][]{bart2010} in Step 5 of Algorithm \ref{algo:importance_resampling}.
In extensive simulation studies in Section \ref{sec:sim_study}, we
notice that an AUC (area
under the receiver operating characteristic curve) 
less than 0.6 
yields excellent empirical performance. 
Once equivalence is achieved, in principle any existing approach for
inference on treatment effects can be used
(see Section  \ref{sec:regression_with_CAM} and later).

 {Note that even if the RWD population is not a heterogeneous superset of the current trial, one can still fit the CAM model.
	In case the RWD is not comparable, Step 5 of Algorithm \ref{algo:importance_resampling} can discriminate the two populations and 
	the AUC can quantify the degree of incongruence.}\label{pg:quantify_comparability}

In general, importance sampling schemes need the ratio of the target
density (in our case, $\bF_{1}$) and the importance sampling density
(in our case, $\bF_{2}$). 
For our problem, this would require high-dimensional
density estimation. 
Even if the densities were known, importance sampling
would be plagued by unbounded weights
\citep{au2003importance_sample}. \label{pg:why_density_free}
Exploiting the common atoms structure, our proposed scheme however
avoids evaluation of the marginal multivariate densities.
We therefore refer to this as a \textit{density-free importance resampling scheme}, and for brevity often simply as an \textit{IS scheme}. 
{In the denominator of $w_{i}$, the use of $n_{2,j}$ (which by definition are $\geq 1$) 
avoids complications arising from
unbounded weights.}
 {Conventional importance sampling schemes are asymptotically consistent. 
This is seen to hold in numerical experiments with our algorithm as well.}\label{pg:IS_unbiased}
Additional discussions on Algorithm \ref{algo:importance_resampling} are in Section \ref{sm sec:importance_resamping} of the supplementary materials.


\subsection{Regression with CAM Model on Covariates}
\label{subsec:CA-PPMx}
Note that up to here we only concerned ourselves with the covariates,
without any reference to the outcomes $\bY$.
In preparation for one of the strategies in the upcoming discussion of
treatment comparison (Section \ref{sec:regression_with_CAM}), we now augment the CAM model to include a sampling model for the outcomes.  
That is, we add a response model on top of the CAM model on covariates.

The extended model defines a regression of $Y_{s,i}$ on covariates
$\bX_{s,i}$ by first grouping patients
with similar covariate profiles into clusters
and then adding a cluster-specific sampling model for the outcome $Y_{s,i}$.
That is, the overall model specifies a regression of $Y_{s,i}$ on
$\bX_{s,i}$ via a random partition.
A major advantage of this approach is
that it allows a variable-dimension covariate vector -- a feature that
is not straightforward to include in a regression otherwise.
Similar product partition models with regression on covariates (PPMx,
see also \ref{subsec:ppmx} in the supplementary materials)
were considered by \citet{ppmx} and \citet{ppmx_missing}, albeit
without any notion of common atoms. 
We will therefore refer to the model proposed below as the common
atoms PPMx (CA-PPMx).
Formally,
we introduce cluster-specific parameters
$\btheta_{s}=\{\btheta_{s,j};~ j=1,\ldots,\tilkn\}$, and assume 
\vskip-5ex
\begin{equation}	
	(Y_{s,i}\mid \btheta_{s}, c_{s,i}=j) \simind \h(Y_{s,i}\mid \btheta_{s,j}), \label{eq:sampligmodel_cappmx}
\end{equation}
\vskip-2ex
\noindent for a suitable choice of $\h$.
For example, for an event-time response,
$\h$ could be a lognormal, exponential or Weibull model.
The response model \eqref{eq:sampligmodel_cappmx} depends on the
covariates indirectly via $c_{s,i}$'s, i.e., the partition induced by the
covariates.
 Within stratum $C_{j} = C_{1,j} \cup C_{2,j}$, the response models
allows for a treatment comparison based on $(\btheta_{1,j},
\btheta_{2,j})$, which can then be averaged with respect to the
assumed distribution of $\bX_{s,j}$
to define an average treatment effect. 

For the implementation in the motivating case study, we let $Y_{s,i}$
denote the $\log$ OS (overall survival) times and assume 
$h(Y_{s,i} \mid \btheta_{s,j})$ to be a normal kernel with $\btheta_{s,j}=(\mu_{s,j}, \sigma^{2}_{s,j})$. 
Such mixtures are highly flexible \citep{ghosal_density99}, making them an attractive choice for many applications.
%
We complete the model with conjugate
normal-inverse-gamma $(\NIG)$ priors on the $(\mu_{s,j},
\sigma^{2}_{s,j})$'s. 
In summary, we have 
\vspace*{-3ex}
\begin{equation}
	\textstyle{Y_{s,i}\mid c_{s,i}=j, \btheta_{s} \simind \mn(\mu_{s,j}, \sigma_{s,j}^{2}),~
	\mu_{s,j}\mid \sigma^{2}_{s,j} \simind \mn\left(\mu_{0}, \frac{\sigma^{2}_{s,j}}{ \kappa_{0}}\right),~
	\sigma^{-2}_{s,j} \simiid \Ga(a_{0},b_{0})}, \label{eq:mixture_lognormal}
\end{equation}
\vskip-3ex
\noindent where $\Ga(a_{0},b_{0})$ is a gamma distribution with mean $a_{0}/b_{0}$. 
We add the hyper-priors
$\mu_{0}\sim \mn(m_{\mu}, s_{\mu}^{2})$
and  $\log b_{0}\sim \mn(m_{b}, s_{b}^{2})$
on the main location-scale controlling hyper-parameters $\mu_{0}$ and
$b_{0}$ while fixing the precision hyper-parameters $\kappa_{0}$ and
$a_{0}$. 
Choices of these hyperparameters are discussed in Section \ref{sm sec:hyperparam} of the supplementary materials.
Finally, 
for a goodness-of-fit test under the proposed model, we 
use the approach of
\citet{valen_model_fit_2007} to build a
graphical tool based on quantile plots.
Such visual tools are often quite effective for detecting departures from model assumptions \citep[Chapter
2]{meloun2011}. 
See Section \ref{subsec:goodnessoffit} in the supplement for more details.


\section{Inference on Treatment Effects}
\label{sec:regression_with_CAM}

\subsection{Two-step Importance Sampling (IS) Approach}
\label{subsec:two_step_IS}

We already described the use of the weights $\pi_{s,j}$ in the CAM
model to achieve equivalent patient populations.  
This allows a straightforward approach to treatment    
comparison. 
Using the adjusted (resampled) subpopulation of $\bXtw$, one can 
proceed with inference on the treatment effect using any method
relying on equivalent patient populations across the two arms.
We refer to this approach as the {\bf ``two-step IS"} and use it 
in the simulation studies and applications in Sections
\ref{sec:sim_study} and \ref{sec:GBM}, respectively.
This approach does not make use of the outcome model of 
Section \ref{subsec:CA-PPMx}. 
%

\subsection{Model-Based Inference for Treatment Effects}
\label{subsec:treatment_effect}
Alternatively, we implement inference using the response model of
Section \ref{subsec:CA-PPMx}, i.e., the full CA-PPMx. 
We refer to this approach as
{\bf ``model-based inference"}. 
We assume that the desired inference on treatment effects takes the
form of inference for some notion of difference $\delta(\cdot,\cdot)$
of the marginal distributions
under the two treatment arms,
$\Delta_{\theta}=\delta\left\{
f_{1}(\cdot \mid \btheta_{1},\bpion),
f_{2}(\cdot \mid \btheta_{2},\bpitw) \right\}.$
However, 
since the covariate populations in the two treatment arms can be substantially different, 
comparison between the marginal
(with respect to the covariates) 
outcome models
$f_{1}(\bY_{1,i} \mid
\btheta_{1},\bpion)$ and $f_{2}(\bY_{2,i} \mid \btheta_{2},\bpitw)$
can be biased.
We need to appropriately adjust for the differences
in the two populations. 
We do this by replacing $f_{2}$ as
follows. 
Exploiting the common atoms structure of the proposed CA-PPMx, there is
an operationally simple method to carry out this adjustment and 
infer treatment effects. 
Since within each cluster, the covariate populations can be considered
equivalent, the adjustment for the lack of randomization amounts to
adjusting the corresponding cluster weights.
We define 
\vskip-5ex
\begin{equation*}
	\textstyle \wt{f}_{2}(Y \mid  \btheta_{2}, \bpion) = \sum_{j=1}^{\tilkn} \pi_{1,j} \h(Y \mid \btheta_{2,j}),
\end{equation*}
\vskip-2ex
\noindent
where the mixture components $\h(Y \mid \btheta_{2,j})$ of the
response model in the RWD are weighted by $\bpion$, i.e., the cluster
weights associated with $\bXon$ (rather than $\bpitw$).
Thus $\wt{f}_{2}$ is the distribution of 
outcomes under control in the treatment population or in other words, the
response of an average individual from the trial arm potentially
treated with the control therapy.
With these notions, we define the population adjusted treatment effect as 
\vskip-8ex
\begin{equation}
  \wt {\Delta}_{\theta}=\delta\{f_{1}(\cdot \mid \btheta_{1},\bpion),
  \wt{f}_{2}(\cdot \mid \btheta_{2},\bpion) \}.
\label{eq:Delta3}      
\end{equation} 
\vskip-2ex
\noindent
For example, when $Y$ is a univariate response variable
and $\delta(f_{1},f_{2})={\eE_{f_{1}} (Y)- \eE_{f_{2}}(Y)}$,
$\wt{\Delta}_{\theta}$ simplifies to
$\wt{\Delta}_{\theta}= \sum_{j=1}^{\tilkn} \pi_{1,j} \left\{
\eE_{\h(Y\mid \btheta_{1,j})} (Y) - \eE_{\h(Y\mid \btheta_{2,j})} (Y) \right\}$,
which further reduces to $\wt\Delta_{\theta} = \sum_{j=1}^{k(n_{2})}
\pi_{1,j} (\mu_{1,j}-\mu_{2,j})$ when $\mu_{s,j}=\eE_{\h(Y\mid
  \btheta_{s,j})}\{T(Y)\}$.

In general, each cluster of covariates in the CAM model can be interpreted as
a homogeneous sub-population of patients.
For the $j\th$ group, the average treatment effect is
$\delta\{\h(Y\mid \btheta_{1,j}), \h(Y\mid  \btheta_{2,j})\}$
and its proportion in the target population is $\pi_{1,j}$.
%
%
The reported treatment effect \eqref{eq:Delta3} includes the
adjustment with the sub-population proportions 
$\pi_{1,j}$.
On a related point, the proposed
model-based inference on treatment effects in the CA-PPMx model can be interpreted as a stochastic
propensity score stratification approach.
See Section \ref{sec:SM-PS} in the supplementary materials for the details. 
\vskip-1ex
{We prefer the Bayesian model-based approach to avoid discarding unmatched patient records from the RWD from the analysis.
The two-step IS can be useful to validate the results obtained
by the model-based approach. 


\section{{Posterior Computation}}
\label{sec:posterior_computation}
\label{editor_comment}
We develop an efficient Gibbs sampler 
for posterior inference in the proposed CAM model for
non-conjugate mixture of lognormals on survival outcomes.
One potential complication arises from 
the varying dimension of $\bpion$ depending on the observed atoms in $\bXtw$.
Posterior simulation with variable dimensional parameters 
generally involves complicated trans-dimensional
Markov chain Monte Carlo \citep{rjmcmc}, often
resulting in poor mixing and computational inefficiencies. 
Our posterior sampling algorithm avoids such complications while
rigorously maintaining the architecture of the CAM model.
See Section \ref{SM subsec:MCMC} in the supplementary materials for more details.

\section{Simulation Study}
\label{sec:sim_study}
We first describe the simulation scenarios.

\noindent 
\underline{CAM scenario:}
We first consider a scenario where the covariates are generated from a CAM model.
In this scenario, we take the first $q=p-3$ 
covariates to be continuous and the remaining $3$ to be binary. 
For the trial arm $s=1$, we generate $\bX_{1,i, 1:q} \simiid \sum_{j=1}^{2} \pi_{1,j} \mn_{q}(\bmu_{j}, \sigma_{j}^{2} \bI_{q})$ 
and $X_{1,i,\ell} \simiid \Bern(\varrho_{1})$ for $\ell={q}{+1},\dots, p$.
For the RWD arm, $s=2$, we generate $\bX_{2,i,1:q} \simiid
\sum_{j=1}^{3} \pi_{2,j} \mn_{q}(\bmu_{j}, \sigma_{j}^{2} \bI_{q})$ and  
$X_{2,i,\ell} \simiid \sum_{j=1}^{2} \iota_{j} \Bern(\varrho_{j})$ for $\ell={q}{+1},\dots, p$ where $\iota_{1}=\pi_{2,1}+\pi_{2,2}$ and $\iota_{2}=\pi_{2,3}$.
We take $\iota_{1}\ll \iota_{2}$ ensuring that the $\bXtw$ population  is
substantially different from $\bXon$ in having more heterogeneity. 

\noindent \underline{MIX scenario:}
In this scenario, we generate 
$\bX_{s,i}\simiid \sum_{j=1}^{k} \pi_{s,j} \mn_{p}(\bmu_{s,j}, 0.05 \bI_ {p})$.
We take $\bmu_{1,j}=\bmu_{2,j}$ for all $j< k$ but set $\bmu_{1,k}\neq \bmu_{2,k}$ so that the atoms in the treatment arm are not exactly a subset of those in the RWD.
Given the typically larger heterogeneity of the RWD, this is not
a realistic scenario.
We include it to evaluate the approach under model misspecification.
Different weights attached to the atoms 
in the two populations result in significantly different
marginal densities. 

\noindent
\underline{Interaction scenario:} \label{interaction_scenario}
In this scenario, we resample from the historical GBM
database of 339 patients to create a future single-arm trial population. 
 Let $\bF(\bX)$ denote the (unknown) distribution of the covariates in the
database, and $Z$ be an indicator variable such that
$Z=s$ if $\bX$ is selected into arm $s$. 
That is, we sample $\bX_{1,i}$ i.i.d. from
$p(\bX_{1,i}) \propto \bF(\bX_{1,i})\cdot e(\bX_{1,i})$ 
and $\bX_{2,i}$ from 
$p(\bX_{2,i}) \propto \bF(\bX_{2,i})\cdot \{1-e(\bX_{2,i})\}$ 
where
$e(\bX)=\Pr(Z= 1 \mid \bX)$
is the PS of assignment to the treatment arm.
We set $e(\bX)$ to be a 
logistic regression with  {pairwise}\label{pairwise_interaction} interactions between some covariates.
We can sample $\bX_{s,i}$ by simple weighted resampling of the
historical database, without explicitly knowing $\bF(\cdot)$.
\smallskip

\noindent
\underline{Oracle scenario:} 
In this fourth and final scenario, we proceed as in the Interaction scenario but
now with $e(\bX)$ defined as a logistic regression with main effects of the true predictors only, 
i.e., as if an oracle had revealed the right predictors. 

\noindent \underline{Outcome model:}
\label{psrwe}
Under the CAM and MIX scenarios, we generate
$Y_{1,i}=\delta+ f(\bX_{1,i}) + \epsilon_{1,i} $ and
$Y_{2,i}= f(\bX_{2,i}) + \epsilon_{2,i} $
where $f(\cdot)$ is a nonlinear function;
in the Interaction and Oracle scenarios we generate
$Y_{1,i}=\delta+ \bX_{1,i}\trans \bbeta + \epsilon_{1,i} $ and
$Y_{2,i}= \bX_{2,i}\trans \bbeta + \epsilon_{2,i} $, where
$\epsilon_{s,i}\simiid \mn(0, 1)$ for 
$i=1,\dots,\ns$ and $s=1,2$, implying $\delta$ as the true treatment effect. 
We repeat the experiments for $\delta=-1,0,1,3$.

We repeat the simulations in the CAM and MIX scenarios for $p=10,20$,
$n_{1}=50,100,150$ and set $n_{2}=6 \times n_{1}$
for all setups, keeping the ratio of the population sizes consistent 
with the GBM application.
For each $(n_{1},p,\delta)$ combination in the CAM and MIX scenarios,
we perform 500 independent replications. 
Under the Interaction and Oracle scenarios, there are $p=11$ covariates and we use $n_{1}=49$.
To avoid reporting summaries that might just hinge on a lucky
choice of the logistic regression coefficients in $e(\cdot)$
and to remove one source of randomness unrelated to the methods
under comparison, 
we independently sample
different sets of regression coefficients (from a discrete
mixture distribution) for each of the 500 repeat simulations.
 {Further details are provided in Section \ref{sm subsec:simulation_truth} of the supplementary materials.} \label{ae:true_sim_details}

\vspace*{-2ex}
\paragraph*{Analyses:} 
We compare the CA-PPMx model
with the PS-integrated power prior
and composite likelihood approaches
\citep{propensity_power_prior,propensity_CL_first,propensity_CL} as implemented  in the
\texttt{psrwe} \texttt{R} package, 
and a two-step population \textit{matching} approach. 
We perform  {seven} different analyses
for each of the four scenarios to estimate the treatment effect
$\Delta_{\theta}$ which 
we define here as the difference in mean outcomes, i.e., $\delta$. 
The analyses are 
(i)  \textit{CA-PPMx}: The proposed \textit{CA-PPMx}
model of Section \ref{subsec:treatment_effect};
(ii) \textit{IS-LM}: The two-step IS approach introduced in Section
  \ref{subsec:IS_reweighing}.
We first sample a subpopulation of size
$n_{1}$ from $\bXtw$ following the importance resampling scheme
proposed in Section \ref{subsec:IS_reweighing} and subsequently
estimate the treatment effect between the subpopulation and the
treatment arm by fitting a linear model; 
(iii) and (iv) \textit{PP-Logistic} and \textit{PP-RF}: Two PS-based power prior approaches using 
logistic regression and random forest \citep{randomforest}, respectively;
(v) and (vi) \textit{CL-Logistic} and \textit{CL-RF}: Two composite likelihood based approaches with logistic and
random forest classifier based PSs, respectively; 
and finally, 
(vii) \textit{Matching}: A
distance based bipartite matching method designed to
match treatment and control groups in observational
studies \citep{optmatch_paper} and subsequently using a linear model
for detecting treatment effects as implemented in the \texttt{optmatch
  R} package. 

\vspace*{-2ex}
\paragraph*{Equivalence of populations:} 
In preparation for inference under the two-step IS approach, 
we generate equivalent populations using the density-free
importance resampling scheme discussed in Section
\ref{subsec:IS_reweighing} based on the fitted CAM model.
To formally test for equivalence of the adjusted datasets, 
we implement Step 5 in Algorithm \ref{algo:importance_resampling}. 
We first merge the datasets and then try to classify patients in
the merged sample as originally RWD or single-arm treatment cohort
($s=2$ vs. $s=1$ in our earlier notation).
For classification, we use  BART 
and report the boxplots of the area under the receiver operating
characteristic curve (AUC) of the classification accuracy across the
independent experiments for all simulation settings in Figure
\ref{fig:AUC_boxplots}.
For comparison, 
we also subsample randomly (instead of using the IS weights) and
report the AUCs in the same figure.
We refer to the two sampling strategies as \textit{IS} and
\textit{Random}, respectively. 

\begin{figure}[!ht]
	\centering
	\begin{subfigure}[b]{.75\linewidth}
		\includegraphics[width=\linewidth]{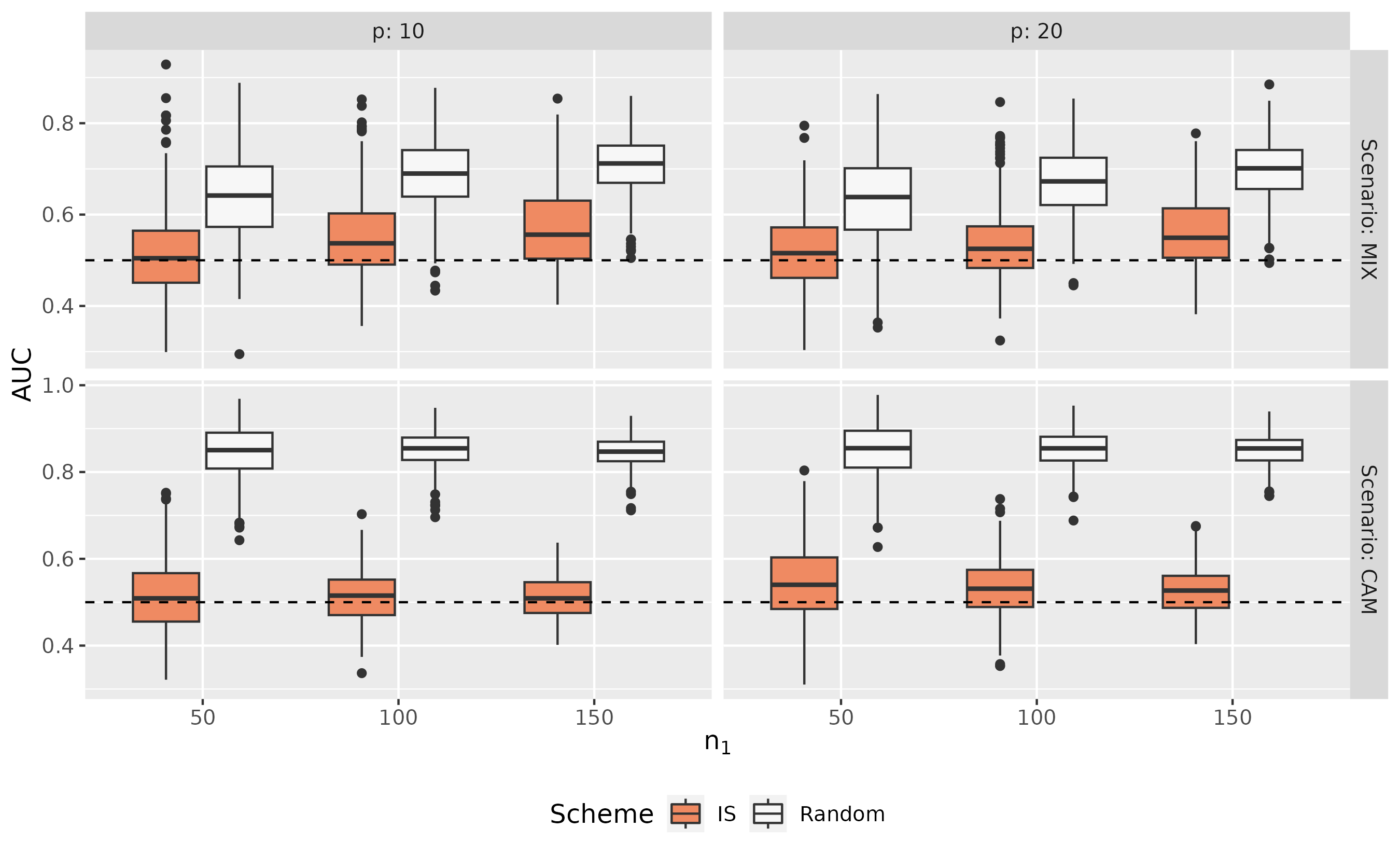}
		\vskip-1.5ex
		\caption{}
		\label{fig:CAM_MIX_auc}
	\end{subfigure}
	\begin{subfigure}[b]{.225\linewidth}
		\centering        
		\includegraphics[width=\linewidth,trim={.275in 0cm 0cm 0cm},clip]{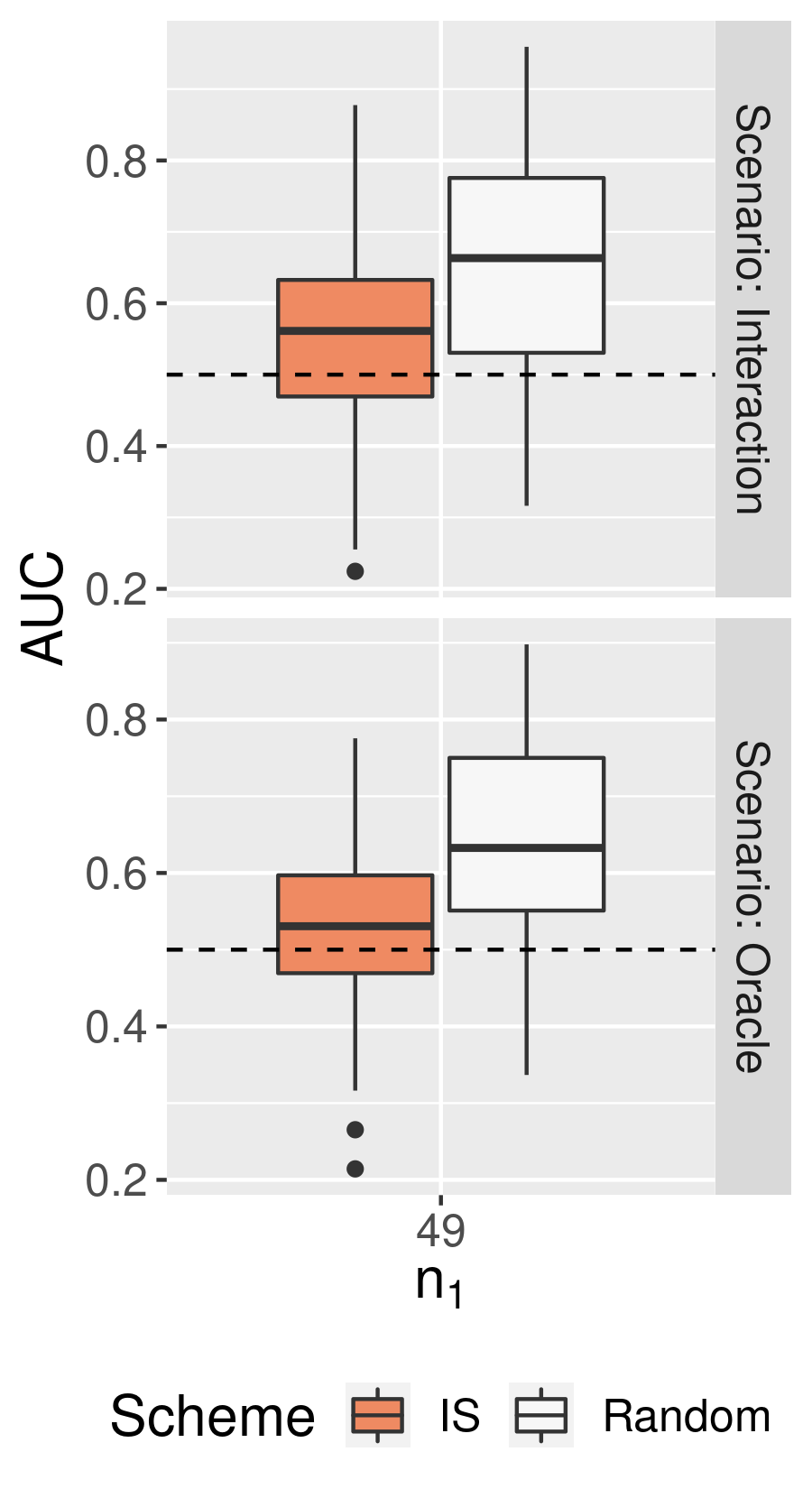}
		\vskip-1.5ex
		\caption{}
		\label{fig:GBM_auc}
	\end{subfigure}
	\vskip-1.25ex
	\caption{Boxplots of the area under the receiver operating
		characteristic curve (AUC) of the classification accuracy for a merged
		dataset consisting of $\bXon$ and the subsampled $\bXtw$ using BART,
		and trying to classify into originally $\bXon$ versus $\bXtw$.  Two
		subsampling schemes are used - the importance sampling (IS) strategy
		in Section \ref{subsec:IS_reweighing} and simple random resampling.
		Here AUCs close to 0.5 imply near equivalence between the
		populations.
		Panel (a) shows the AUCs in the CAM and MIX scenarios across
		different sample sizes and number of covariates; 
		and (b) shows the AUCs under the Interaction and Oracle scenarios.
	}
	\label{fig:AUC_boxplots}  
\end{figure}

In Figure \ref{fig:CAM_MIX_auc}, the Random resampling strategy
yields high AUC, indicating that the two populations are
substantially different and adjustment in the RWD population is
necessary before using it as synthetic control.  For both, the CAM and MIX
scenarios, the performance of the IS scheme improves with increasing sample size. 
 This is expected as for small sample sizes $\bXtw$ is lacking enough data
to produce a subsample equivalent to $\bXon$.
AUC values close to 1 under the CAM scenario imply 
that the true populations are indeed very different in this
case.
In contrast, the AUC values close to 0.5 under the IS scheme
indicate near equivalence after adjustment. 
In both scenarios, AUC is substantially reduced under the IS
resampling scheme, implying that the proposed CAM model indeed adjusts for
the lack of randomization.

Results under the last two scenarios are shown in Figure
\ref{fig:GBM_auc}.
Recall that in both scenarios the simulation truth is not based on the
CAM model. Still, the fit under the proposed CAM model achieves 
near perfect adjustment as shown in the figure. 
\begin{figure}[!t]
	\begin{subfigure}[b]{\linewidth}
		\includegraphics[width=\linewidth]{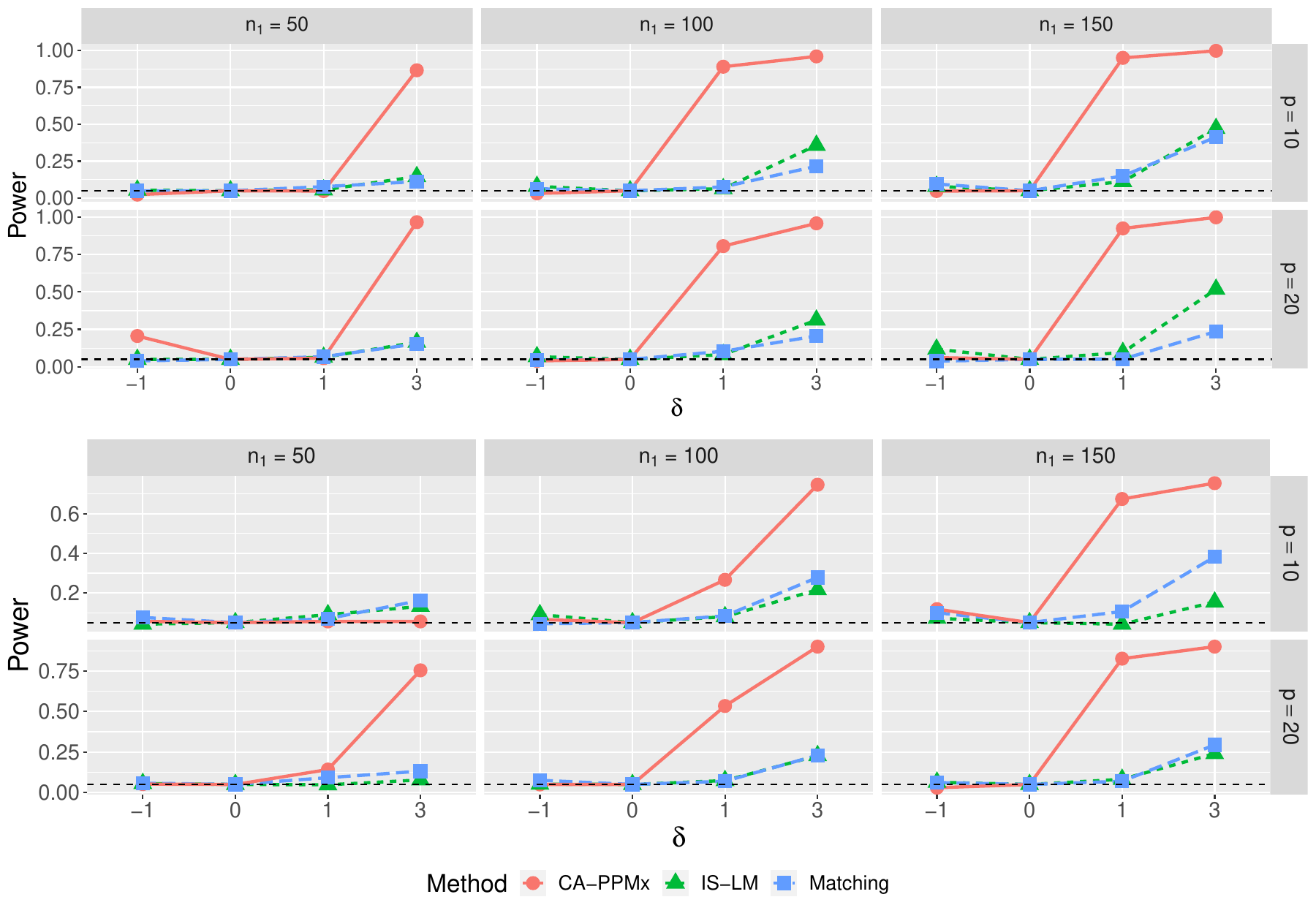}
		\caption{CAM (top) and MIX (bottom) scenarios.}
		\label{fig:mixture}
	\end{subfigure}
	\begin{subfigure}[b]{\linewidth}
		\includegraphics[width=\linewidth]{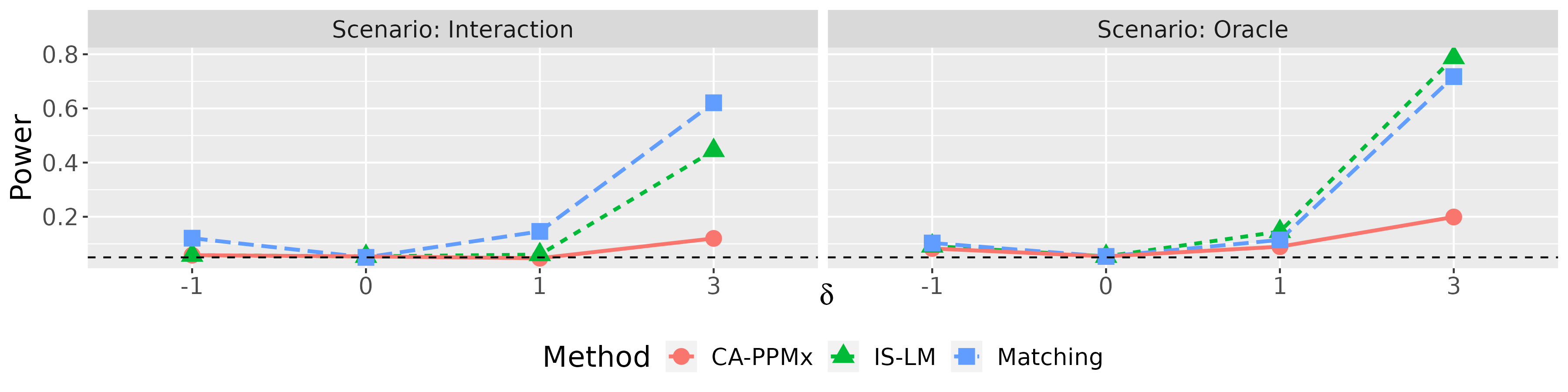}
		\caption{Interaction (left) and oracle (right) scenarios.}
		\label{fig:gbm_bias}
	\end{subfigure}
	\caption{Power of detecting treatment effects in different simulation setups for a $5\%$ level of significance test:		
		{Seven methods are used to estimate the effects where
			IS-LM and CA-PPMx are based on the proposed CAM model. 
			Panel (a) corresponds to the CAM (top) and MIX (bottom) scenarios.
			Panel (b) shows results under the Interaction (left side) and
			the Oracle (right side) scenarios.} }
	\label{fig:Biases}
\end{figure}
\vspace*{-2ex}
\paragraph*{Inference on treatment effects:}\label{referee_complain}
 {In each simulation setup, we test $H_{0}:\delta=0$ versus $H_{1}:\delta \neq 0$ at $5\%$ level of significance.
We elaborate the testing procedure in Section \ref{sm subsec:sim_testing_procedure} of the supplementary materials.}
We report power in Figure \ref{fig:Biases}, with detailed numerical
results appearing in Tables 
\ref{tab: sim_sim}, \ref{tab: sim_gbm} and \ref{sm tab:ps_sim} in the
supplementary materials.
Under the PS-based approaches, the power remains below 15\%
across all scenarios (not shown in the figure). 
Fully model-based nonparametric CA-PPMx has higher power than IS-LM
and Matching when the true response models are non-linear.
In contrast, the IS-LM and Matching perform comparably
and have higher power than the CA-PPMx approach in Interaction and
Oracle scenarios where the true response model is linear, 
but are susceptible to model misspecification as reflected in the CAM
and MIX scenarios.  \label{pg:AE_specific}
This is because IS-LM and Matching assume a linear model for
the outcome, which happens to match the simulation truth in the
Interaction and Oracle scenarios.
Except under the PS-based approaches, power increases with increasing
sample size, indicating that PS-based methods may require a
much larger population size in the RWD to adjust for the lack of
randomization.

 \label{pg:editor}

\section{Application in Glioblastoma}
\label{sec:GBM}

We return to the motivating case study of creating  
a synthetic control for a hypothetical upcoming single-arm GBM trial. The
sample size of the trial is $n_{1}=49$, similar to past trials
\citep{samplesize_gbm}. 
The endpoint of interest is overall survival (OS).
We evaluate the operating characteristics of the proposed design by
simulating $L=100$ trial replicates.
See \citet[Section 2.5.4]{BerryAl:10}
for a discussion of the role of frequentist operating
characteristics in Bayesian inference.
To create treatment arm data, we first select covariates
$\bX_{1,i}$ by randomly selecting patients from the historical
database.
To generate \label{pg:logistic_selection} a realistic non-equivalent patient population, 
we select not uniformly but using a logistic regression on the covariates
(as described in the Interaction scenario in Section \ref{sec:sim_study}). 
The treatment effect is quantified by the hazard ratio (HR) between
the treatment arm and the (synthetic) control arm, with the null and
alternative hypotheses  
$H_{0}: \HR = 1$ vs. $H_{1}: \HR \leq 0.6$ at 50 weeks. 
The HR of 0.6 was suggested by clinical collaborators as a meaningful clinical target.

We show results under two alternative scenarios
(a) $H_{0}$:  no treatment effect (i.e., $\HR = 1$),
created by keeping the OS for the patients in the treatment arm as 
originally observed in the historical database (since the patients received treatments with similar efficacy); and
(b) $H_{1}$: there is a clinically meaningful treatment effect.
We created $H_{1}$ by 
increasing the OS of patients in the treatment arm with an increment that
would correspond to a HR of 0.6 under an exponential model.

We apply three methods to make inference on the treatment effect: 
(i) \textit{IS-based two-step procedure}: Here we first create equivalent patient populations using Algorithm \ref{algo:importance_resampling}
and then proceed with inference on the
treatment effect
as if patients were randomly assigned to treatment and control;
(ii) \textit{Matching-based two-step procedure}: Operationally similar to (i) 
but now the Matching method discussed in Section \ref{sec:sim_study} is used to create equivalent patient populations; and
(iii) \textit{Model-based inference}: The extension of the CAM model to
include the outcomes $Y_{s,i}$, as described in Section  \ref{subsec:treatment_effect}.

\vspace*{-1.5ex}
\paragraph*{(i) IS-based two-step procedure:}
In preparation for inference, we start with a test 
for equivalence of the subsampled population in each of the
$L=100$ repeat simulations. 
Figure \ref{fig:density_adjustment} plots the relative
frequencies for each covariate in the treatment arm (red) and 
in the synthetic control arm constructed from the RWD using:
(a) the IS sampling following Algorithm
\ref{algo:importance_resampling} (green) and 
(b) random sampling (blue). 
Very different frequencies in the two arms under random resampling
indicate significant differences in the covariate distributions 
between the treatment and the control arms. 
For most covariates, the differences are however greatly reduced by the IS scheme. 
\begin{figure}[h] \centering
	\includegraphics[width=\linewidth]{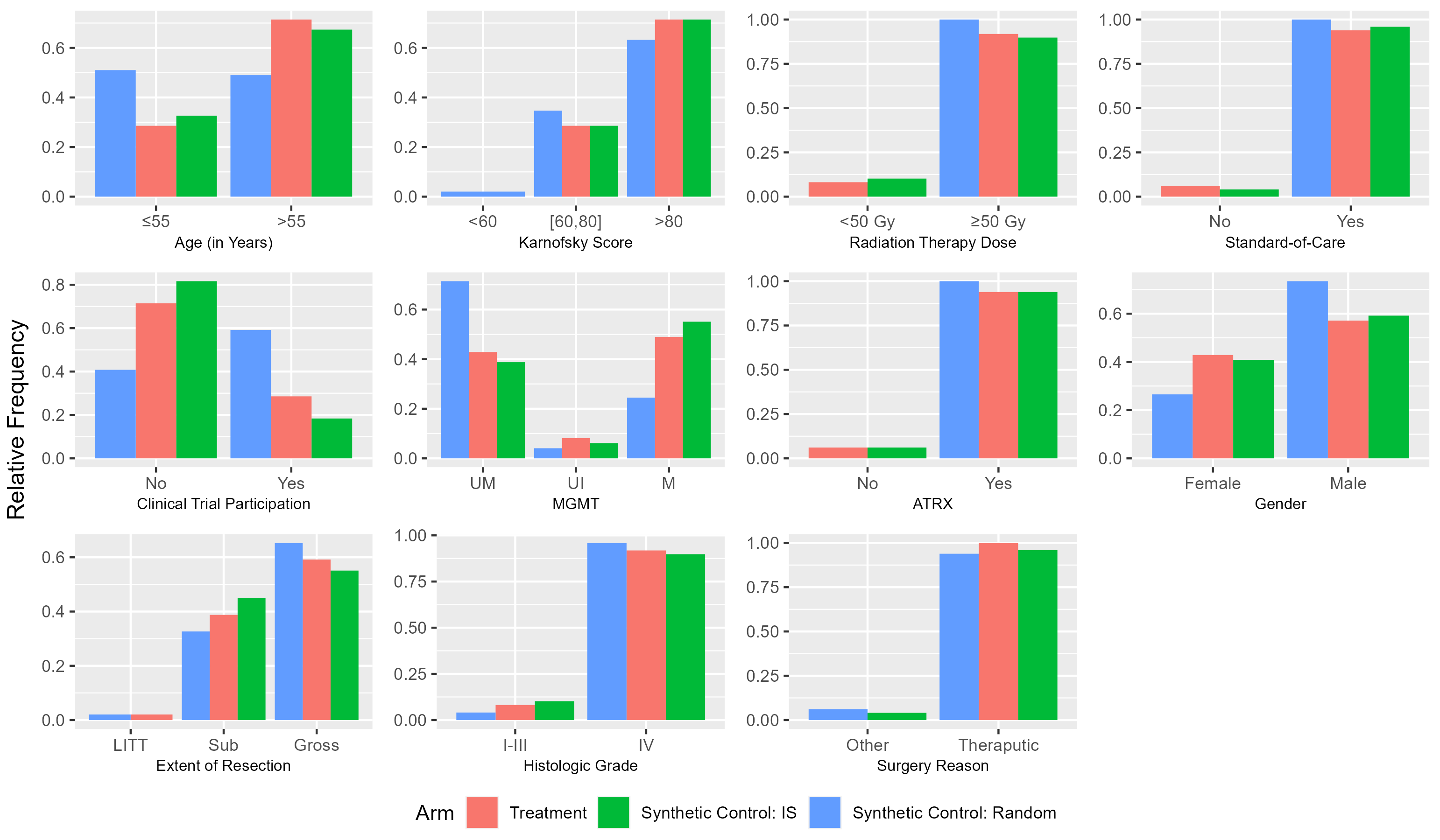} \vskip-1ex
	\caption{Covariate distributions before and afte adjustments. 
		The red bars show the distributions of the covariates in the treatment arm.
		The green and blue bars show the distributions of the covariates in the synthetic control arms formed using the IS and random  resampling schemes, respectively.}
	\label{fig:density_adjustment}
\end{figure}

\begin{figure}[!t]
	\centering
	\begin{subfigure}[b]{.3\linewidth} \centering
		\includegraphics[width=.9\linewidth]{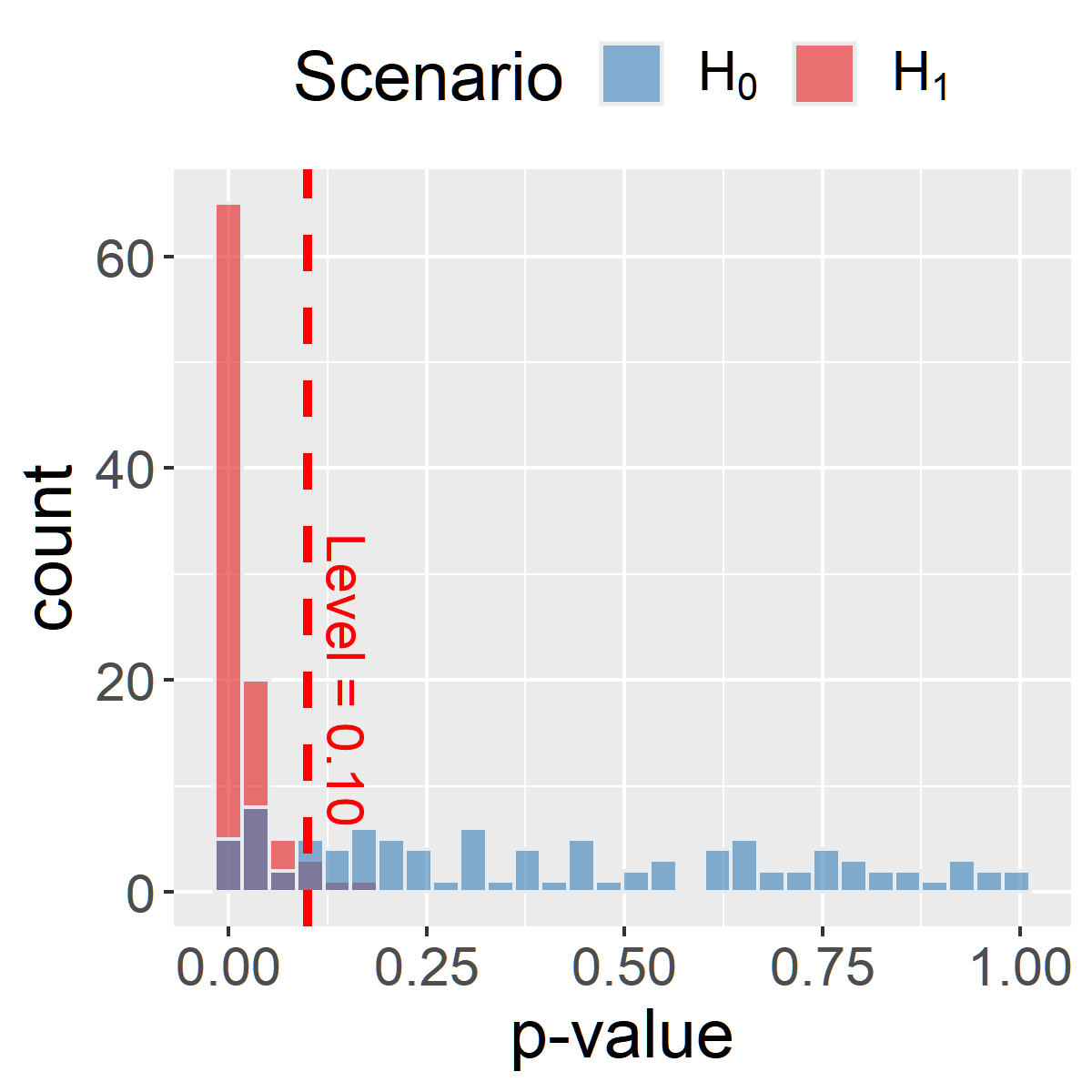} 
		\includegraphics[trim={0cm 0cm 0cm 1.5cm}, clip,width=.9\linewidth]{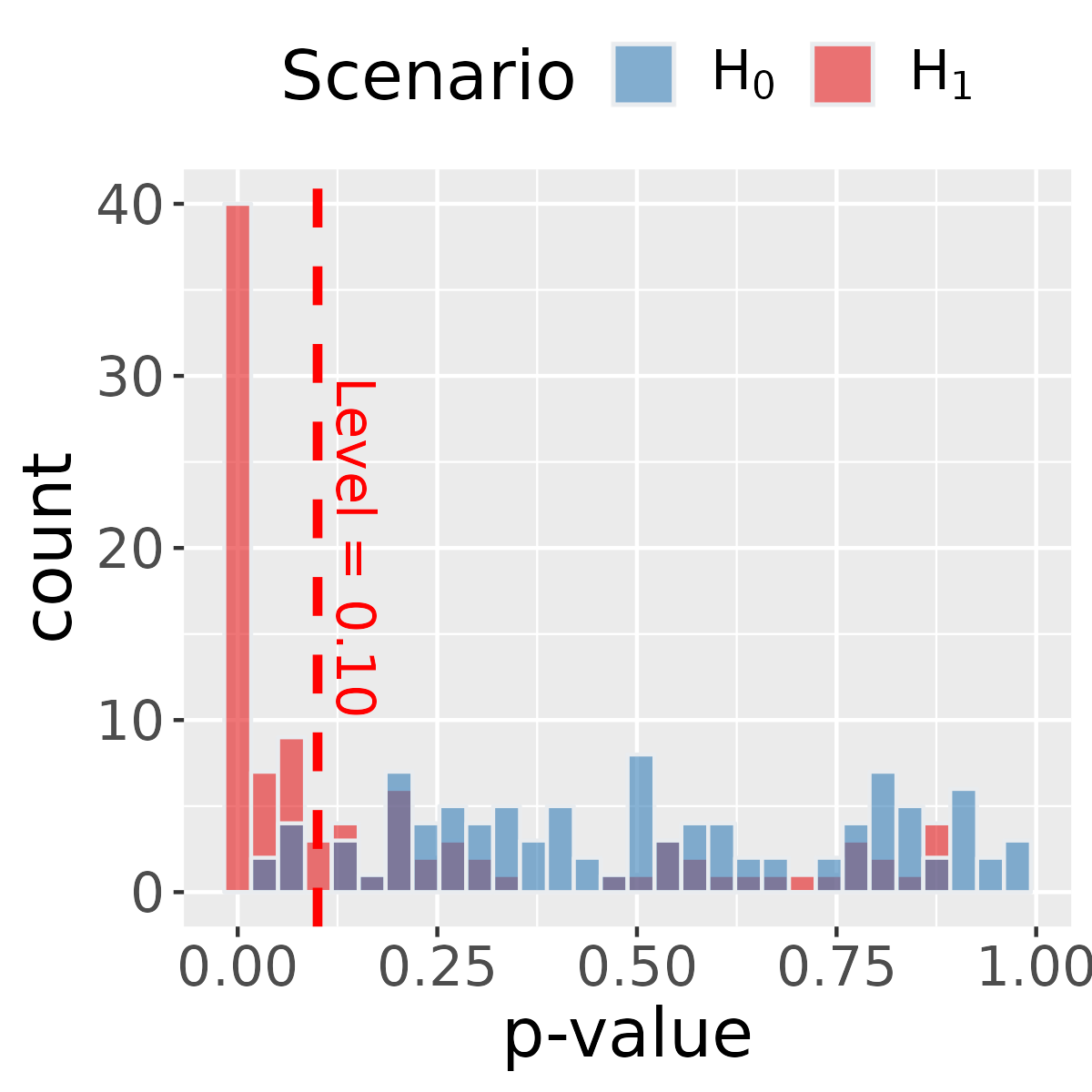}\vskip-1ex
		\caption{$p$-values under the Cox PH model.}
		\label{fig:pvals_GBM}
	\end{subfigure}	
	\begin{subfigure}[b]{.68\linewidth} \centering
		\includegraphics[width=.47\linewidth]{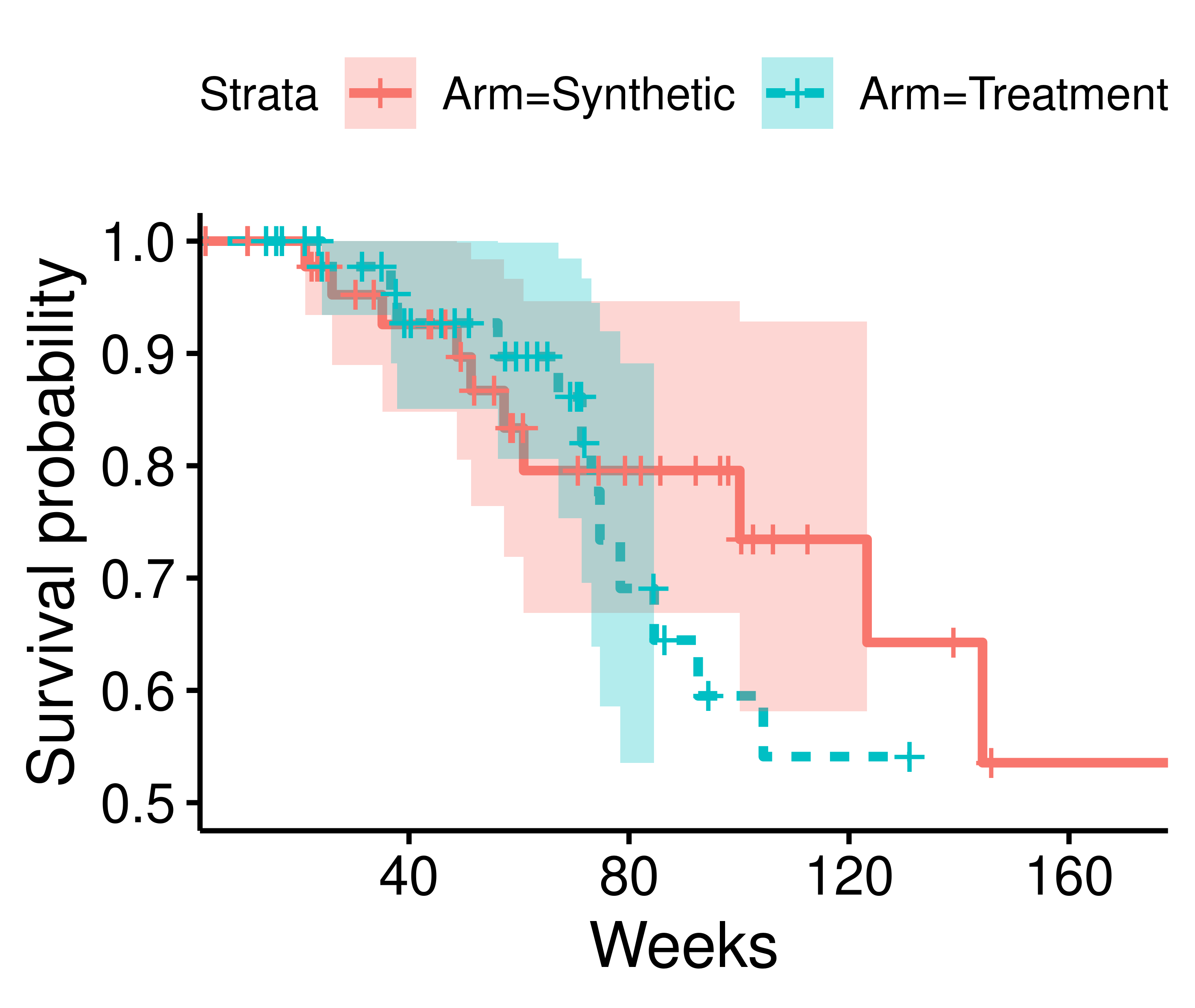}
		\includegraphics[width=.47\linewidth]{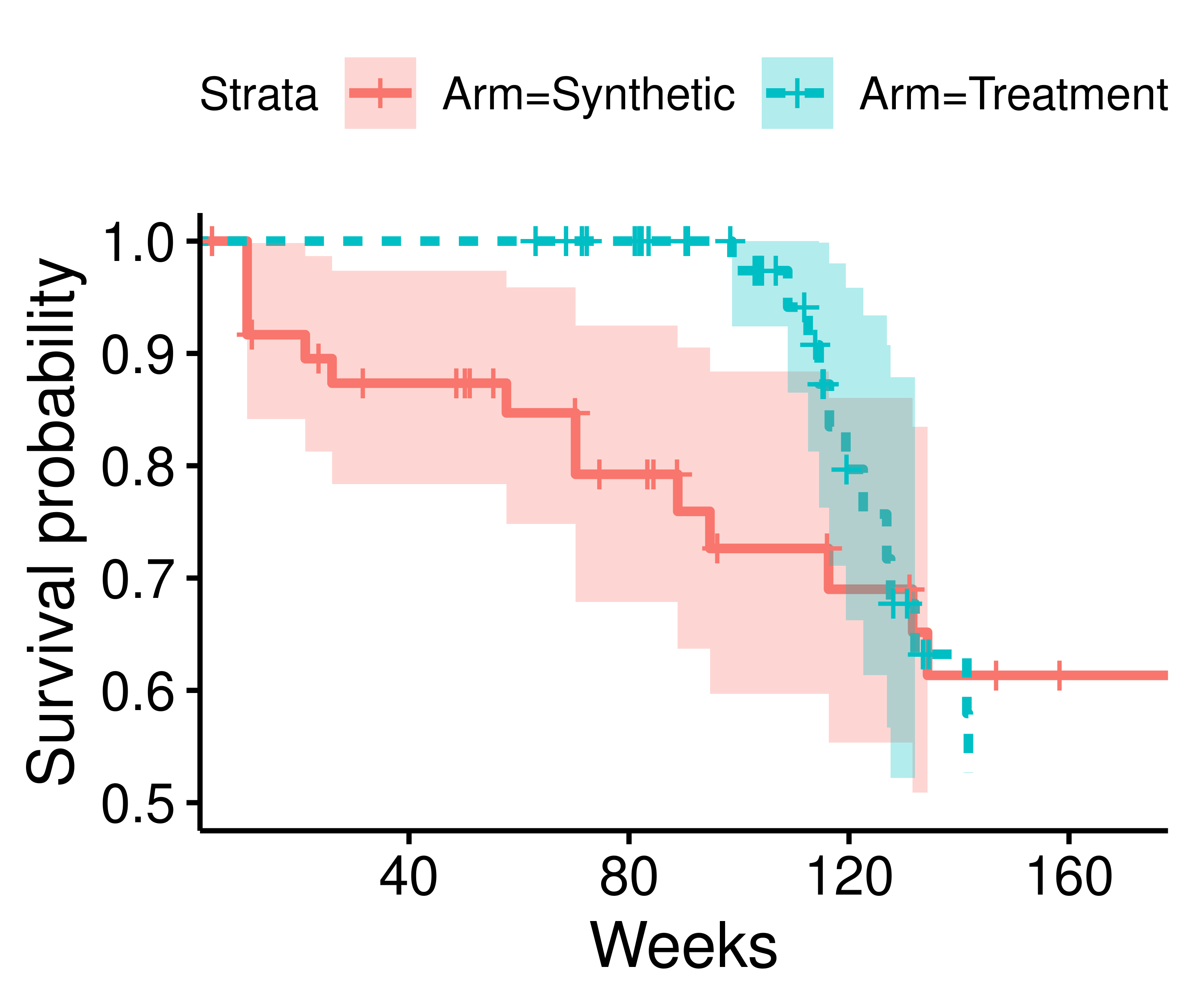} 
		\includegraphics[trim={0cm 0cm 0cm 1.5cm}, clip,width=.47\linewidth] {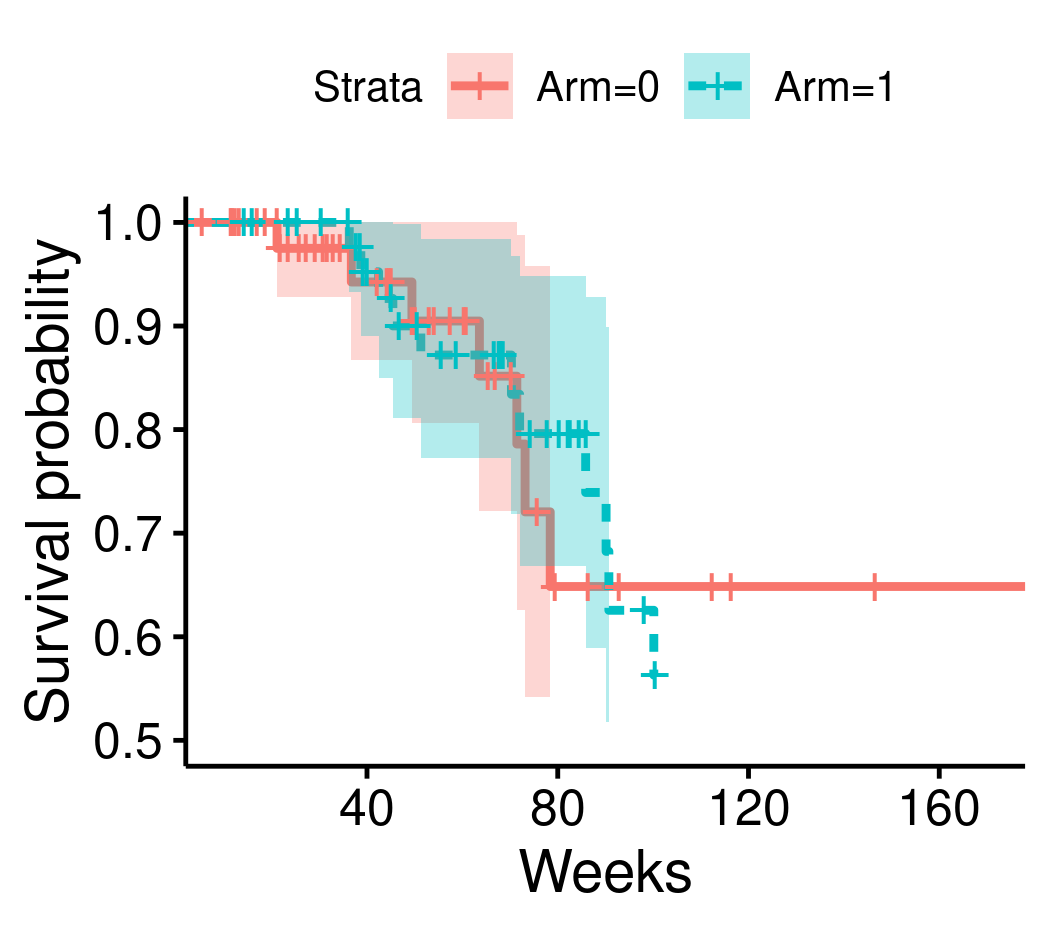}
		\includegraphics[trim={0cm 0cm 0cm 1.5cm}, clip,width=.47\linewidth] {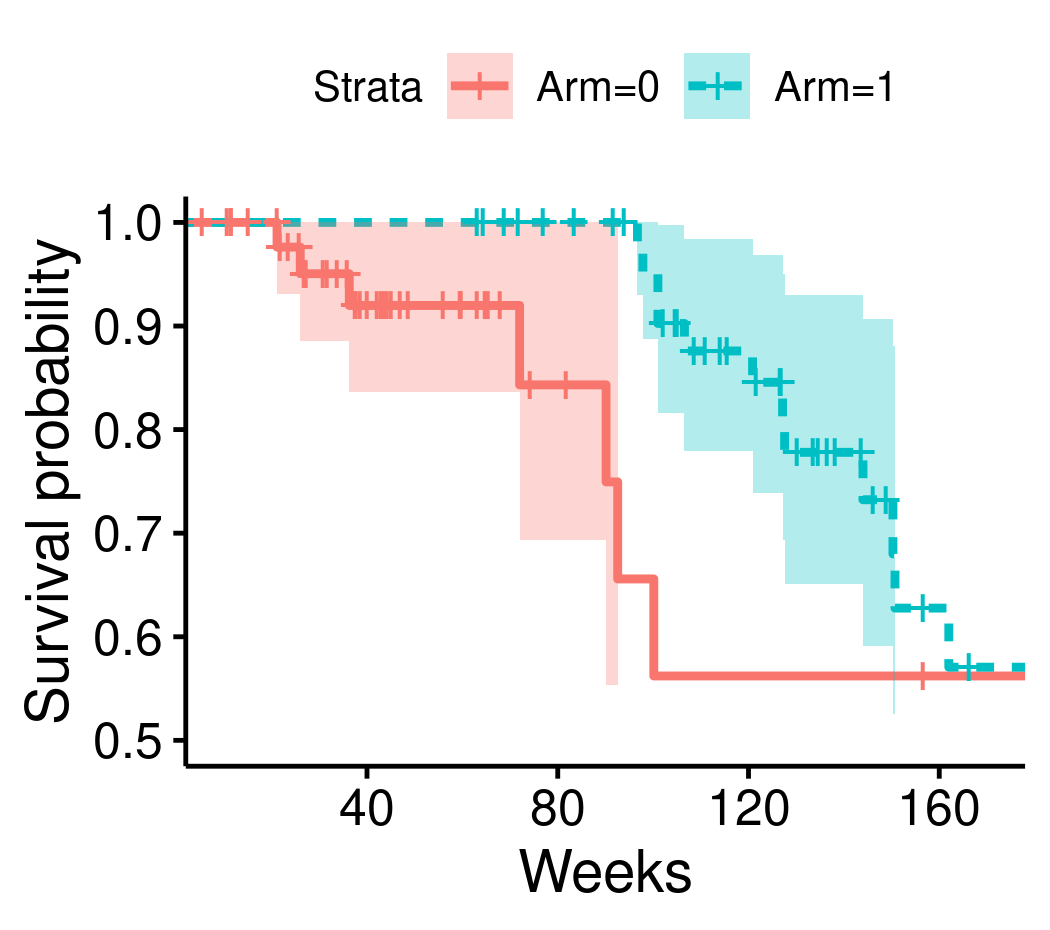} 
		\vskip-1ex
		\caption{Kaplan-Meier curves under $H_{0}$ (left) and $H_{1}$ (right) scenarios.}
		\label{fig:KM}
	\end{subfigure}		
	\caption{Inference under treatment effects under the two-step procedures:
		Panel (a) shows histograms of the $p$-values
		corresponding to a logrank test under the Cox PH model comparing the survival curves
		between the treatment arms; 
		Panel (b) shows the Kaplan-Meier curves and pointwise confidence intervals for treatment (blue) and control (red) arms under scenarios $H_{0}$ (left) and $H_{1}$, respectively.
	The top and bottom panels of (a) and (b) show the results corresponding to the IS and Matching based approaches, respectively.}
	\label{fig:two_step_GBM}
\end{figure}

Once we establish equivalence of the patient populations, we proceed
with inference for the treatment effect.  
We use a Cox proportional hazard (PH) model \citep{coxph1972} and the
logrank test \citep{logrank_test} to compare the survival functions.
The top panel of Figure \ref{fig:pvals_GBM} shows inference summaries over the $L=100$
repetitions. 
The figure shows the histograms of
$p$-values  
under $H_{0}$  (blue) and $H_{1}$ (red).  
Under $H_{0}$, $p$-values are almost uniformly
spread out over $[0,1]$.  
In contrast, under $H_{1}$, the histogram of $p$-values over repeat simulations is peaked close to zero.

Finally, we identify representative simulations from the $L$ repetitions under each of the two scenarios by finding the instance with $p$-value closest to the median of the respective histograms.  
For these two representatives, we show Kaplan-Meier (KM) survival
curves in the top panels of Figure \ref{fig:KM},
respectively.   
We observe that the survival curves in the
two arms are quite alike with wide confidence intervals under the $H_{0}$ scenario, 
whereas significant improvements in the survival times can
be observed for the treatment arm for the first 80 weeks under the $H_1$ scenario.


\vspace*{-1.5ex}
\paragraph*{(ii) Matching-based two-step procedure:} We use the
Matching procedure to create a synthetic control and then
follow the same routine of (i) for inference on treatment effects.
The results are provided in the bottom panels of Figures
\ref{fig:pvals_GBM} and \ref{fig:KM}.
The distribution of the $p$-values under the $H_{1}$ scenario is less
peaked around 0 compared to the IS-based procedure.  This is also
reflected in the representative KM plot under $H_{1}$ in having a much
wider confidence interval around the survival curve possibly
indicating the IS-based approach is doing better than Matching in
creating equivalent populations.

\paragraph*{(iii) Model-based inference:}
As it is not straightforward to account for the uncertainty in creating the synthetic control in the aforementioned two-step procedures, we consider a fully model-based approach. 
For inference on treatment effects, we first assess
goodness-of-fit of the CA-PPMx model
(see Section \ref{subsec:goodnessoffit} in the supplementary materials for details).
Quantile-quantile plots for the two scenarios are shown in Figure
\ref{fig:qqplot_GBM}.
Near diagonal lines indicate no evidence for a lack of fit. 
\begin{figure}[!ht]
	\centering  
	\begin{subfigure}[b]{\linewidth} \centering
		\includegraphics[width=.7\textwidth]{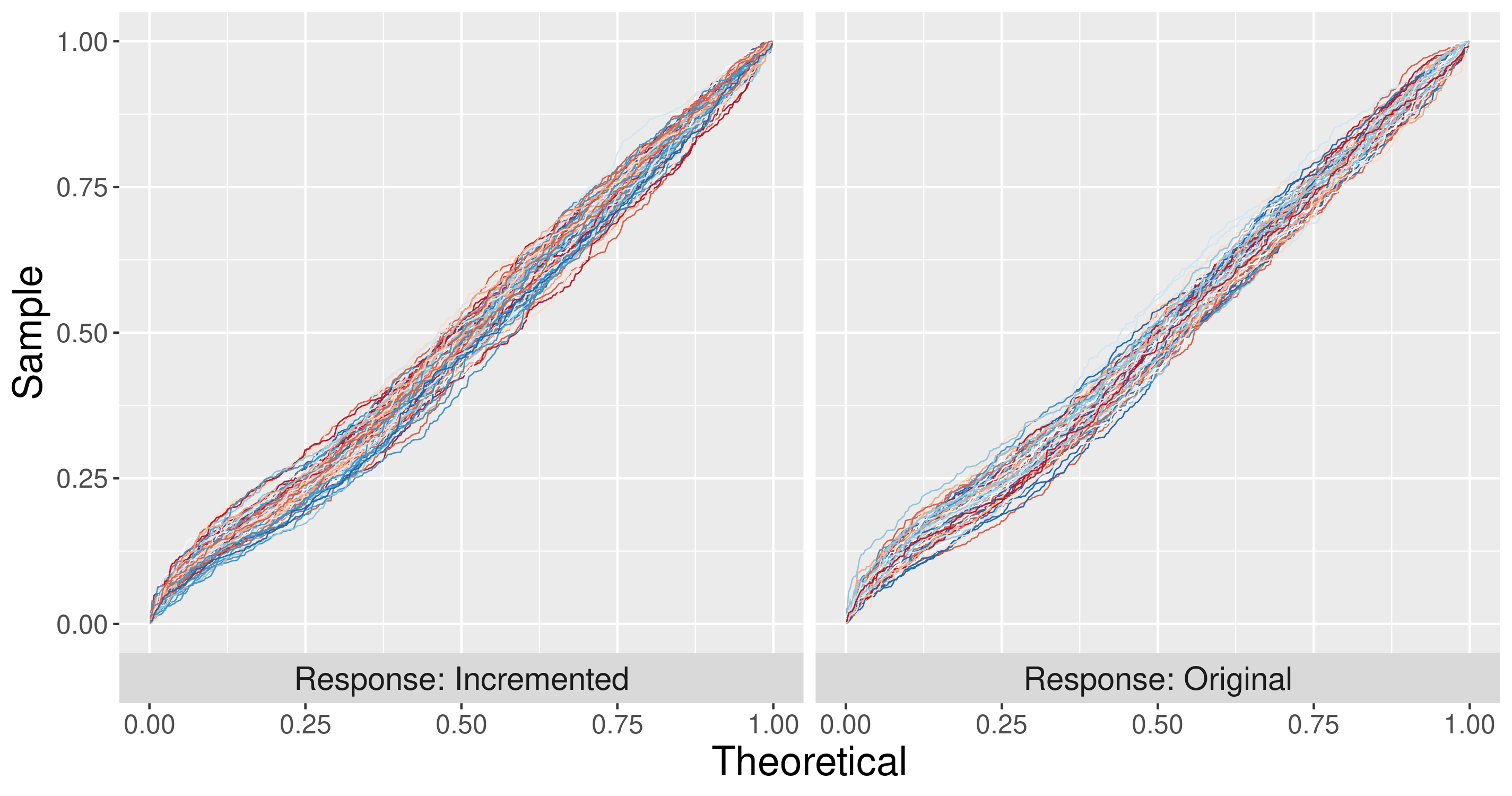} \vskip-.5ex
		\caption{QQ plots for goodness of fit where the $y=x$ line indicates perfect fit.}

		\label{fig:qqplot_GBM}
	\end{subfigure}
\vskip.5ex
	\begin{subfigure}[b]{\linewidth} \centering
		\includegraphics[width=.8\textwidth]{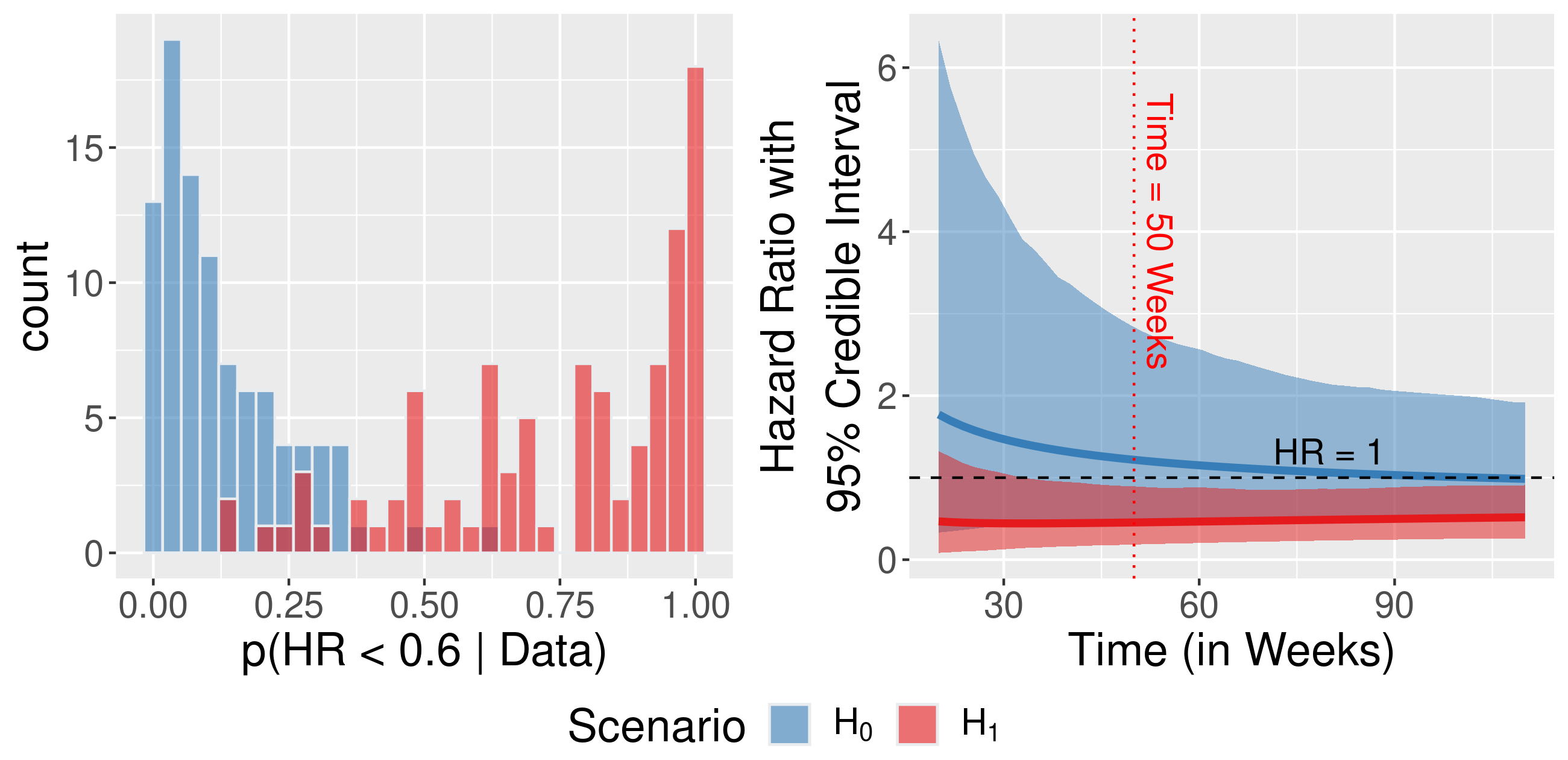} \vskip-.5ex
		\caption{$p(H_{1}\mid \Data)$ (left) and hazard ratio with 95\% posterior credible interval (right). }
		\label{fig:HR_GBM}
	\end{subfigure}  
  \caption{Inference under treatment effects under the model-based approach:
    Panel (a) shows
    quantile-quantile plots to assess model fit from Section \ref{subsec:goodnessoffit} in the supplementary materials;
    the left plot of panel (b) shows posterior
    probabilities $p(\HR < 0.6 \mid \Data)$ under repeat simulations, 
    and  on the right posterior estimated hazard ratios for OS with pointwise 95\%  credible regions are shown under $H_{0}$ and $H_{1}$.} 
  \label{fig:QQ}  
\end{figure}
We then evaluate the posterior probability $p_{\ell} \equiv
p(\HR<0.6 \mid \Data)$ (with $\ell$ indexing the $L=100$ repeat
simulations) at $t=50$ weeks under the proposed model.
The left panel of Figure \ref{fig:HR_GBM} shows histograms of 
$p_{\ell}$
 under $H_{0}$ (in blue) and under $H_{1}$ (in red). 
As desired, the posterior probabilities are clustered near $0$ under $H_{0}$, but are peaked near $1$ under $H_{1}$. 

Finally, we identify a representative simulation again by
selecting the repeat simulation $\ell$ with posterior probability $p_\ell$
closest to the median of the respective histograms under each of the
two scenarios.
For each of the two scenarios, 
we plot the posterior estimated hazard ratios (blue and
red for simulation under $H_{0}$ and $H_{1}$, respectively), together with
pointwise $95\%$ posterior credible intervals in the right panel of Figure \ref{fig:HR_GBM}.
Under $H_{0}$ (blue), HR is almost equal to 1 with wide credible
intervals, whereas under $H_{1}$ (red),
HR is significantly below 1 with high posterior probability.
The median (over the $L$ simulations) posterior
probabilities $p_{{\ell}}(\HR < 0.6 \mid \Data)$ 
are 0.08 and 0.98 under $H_{0}$ and $H_{1}$, respectively.

\section{Discussion}
\label{sec:discussion}
With a long term goal of setting up a platform for future
single-arm early-phase clinical trials in GBM, where new patients only
receive experimental therapies, in this article we developed a Bayesian
nonparametric approach for creating synthetic controls from RWD. 
We introduced a Bayesian CAM model that clusters
covariates with similar values across different treatment arms.

The flexibility of the CAM model makes it easily generalizable to
other problems, e.g., to create two synthetic treatment
arms to compare 
two treatments based on RWD from electronic health records.

Another direction for extensions could build on extracting propensity
scores as inference summaries under the CA-PPMx model. This is briefly
discussed in Section \ref{sec:SM-PS} of the supplementary materials.

A limitation of the current model is {scalability to high-dimensional covariates}. 
In the GBM application, we rely on 11 clinically important
categorical covariates that are commonly considered as prognostic factors in
GBM treatments.  
However, in many applications candidate covariates can be 
high-dimensional.
Implicit in the current construction is the assumption that the
recorded covariates are clinically relevant for the disease or
condition under consideration, 
and the approach may not be appropriate when large numbers of unscreened candidate covariates are used.
Recent advances in Bayesian model-based clustering by
\citet{chandra_lamb} could be useful to construct high-dimensional
generalizations.

\baselineskip=14pt
\subsection*{Supplementary Materials}
Supplementary materials include 
additional discussion of the motivating dataset,
a brief review on the PPMx,
detailed discussion of the graphical goodness-of-fit test for the
regression model, 
an alternative interpretation of our model-based inference approach,
choices of hyperparameters,
details of the posterior simulation scheme,
additional simulation studies and associated details, and MCMC convergence diagnostics.
\texttt{C++} and \texttt{R} programs implementing the methods developed in this article and \texttt{R Markdown} files with instructions are provided in a separately attached \texttt{Codes.zip} folder.

\subsection*{Acknowledgments}
We thank the Editor, Dr. Michael Stein, 
an anonymous Associate Editor and 
two anonymous referees for comments that led to significant improvements in the clarity and presentation of the paper.


\baselineskip=14pt
\bibliographystyle{natbib}
\bibliography{refs}


\clearpage\pagebreak\newpage
\newgeometry{textheight=9in, textwidth=6.5in,}
\pagestyle{fancy}
\fancyhf{}
\rhead{\bfseries\thepage}
\lhead{\bfseries SUPPLEMENTARY MATERIALS}

\baselineskip 20pt
\begin{center}
{\LARGE{Supplementary Materials for\\}} 
\papertitle
\end{center}

\setcounter{equation}{0}
\setcounter{page}{1}
\setcounter{table}{1}
\setcounter{figure}{0}
\setcounter{section}{0}
\numberwithin{table}{section}
\renewcommand{\theequation}{S.\arabic{equation}}
\renewcommand{\thesubsection}{S.\arabic{section}.\arabic{subsection}}
\renewcommand{\thesection}{S.\arabic{section}}
\renewcommand{\thepage}{S.\arabic{page}}
\renewcommand{\thetable}{S.\arabic{table}}
\renewcommand{\thefigure}{S.\arabic{figure}}
\baselineskip=15pt

\vspace{0cm}

\authors

\vskip 10mm
Supplementary materials present
additional discussion on the motivating dataset,
a brief review on the PPMx,
detailed discussion on the graphical goodness-of-fit test of our regression model,
an alternative interpretation of our model-based inference approach,
choices of hyperparameters,
detailed posterior simulation scheme,
additional simulation studies and associated details, and MCMC convergence diagnostics.

\newpage
\baselineskip=16pt

\section{Historical Data and Potential Future Trial}
Figure \ref{fig:GBM_freq_plots} 
shows summaries for the covariates described in Section \ref{sec:GBM_description} in the historical database and a potential future single-arm trial. 
Marginal frequencies for each of the covariates are plotted clearly highlighting the differences between the two populations.

	\begin{figure}[!ht]
		\centering
		\includegraphics[width=\linewidth]{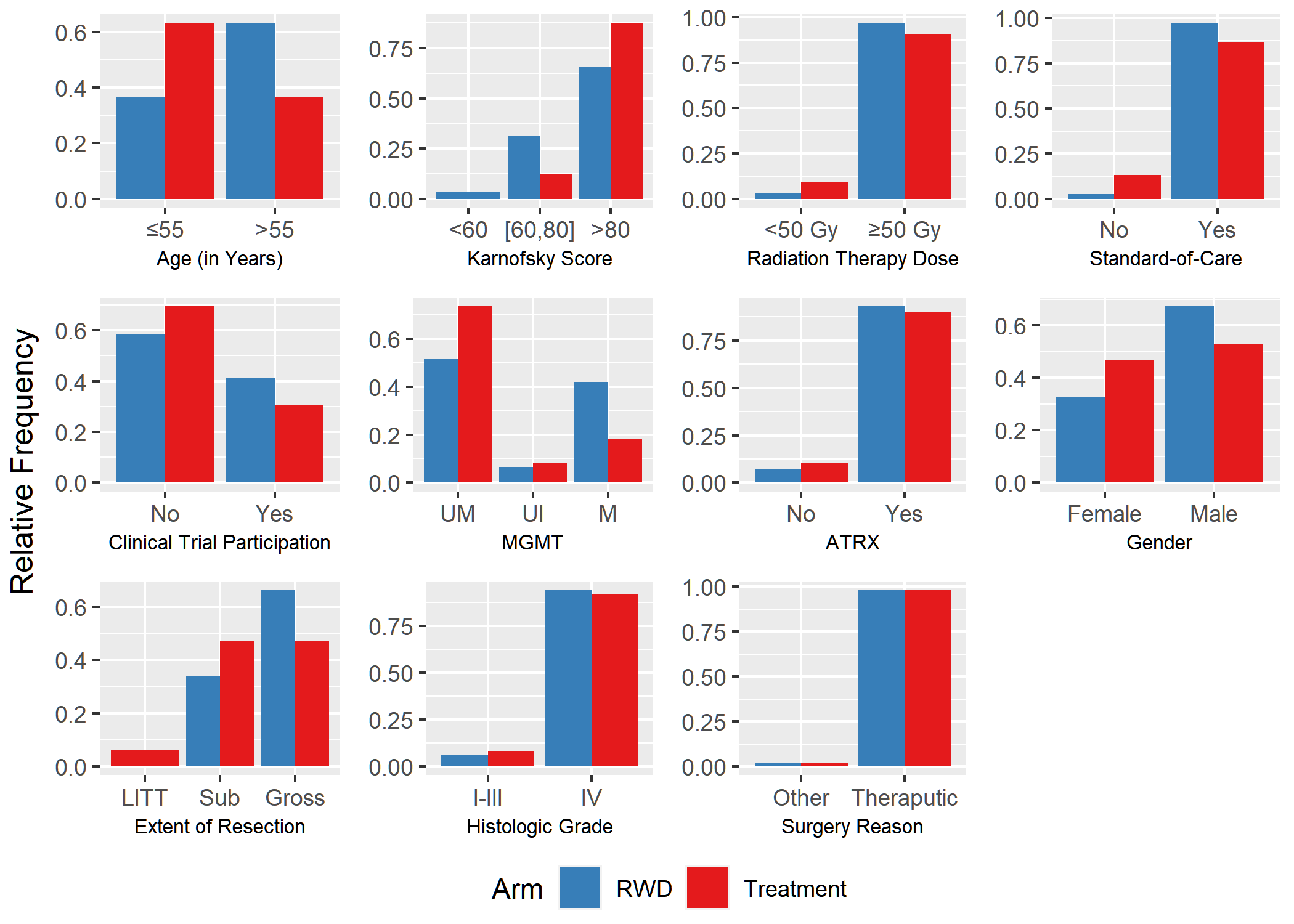}
		\caption{Relative frequency plots of the covariates in the two
			treatment arms.} 
		\label{fig:GBM_freq_plots}	
	\end{figure}

\newcommand{\knn}{k_{n}}

\section{Product Partition Model with Regression (PPMx)} 
\label{subsec:ppmx}
Let $i=1,\dots,n$ be the indices of $n$ data points. For the $i\th$
unit (patient, in our case), the data consists of covariates
$\bX_{i}=\left(X_{i,1},\dots, X_{i,p} \right)\trans$ and response
variables $\bY_{i}$.  Let $\bX=\left\{\bX_{1},\dots, \bX_{n} \right\}$
and $\bY=\left\{ \bY_{1},\dots, \bY_{n} \right\}$ be the complete set
of covariates and responses respectively.  Let
$\brhon=\left\{S_{1},\dots,S_{\knn}\right\}$ denote a partition of the
$n$ units into $\knn$ subsets, where $1\leq\knn\leq n$.  An equivalent
representation of $\brhon$ introduces cluster membership indicators
$c_{i}=j$ if and only if $i\in S_{j}$.
Let $\bXs_{j}$ be the covariates corresponding to the samples in
$S_{j}$.  In the PPMx, it is believed that data points with more
similar covariate values are more likely to \textit{a priori} be in
the same cluster and the corresponding responses are also very
similar.  The prior consists of two functions -
(i) a cohesion function denoted by $c(S_{j} \mid \alpha )\geq 0$ for
$S_{j} \subset \{1, \dots , n\}$ associated with a hyper-parameter
$\alpha$ discerning the prior belief of co-clustering of the elements
of $S_{j}$, and
(ii) a similarity function denoted by $\bg(\bXs_{j} \mid \bxi)$ and
parametrized by $\bxi$, formalizing the `closeness' of the $\bX_{i}$’s
in the cluster $S_{j}$ by producing larger values of $\bg(\bXs_{j}
\mid \bxi)$ for $\bX_{i}$’s that are more similar.
Using the similarity and cohesion functions, the PPMx assumes
\vskip -3.5ex
\begin{equation}
	\Pi\left(\brhon \mid \bX, \alpha,\bxi\right) \propto
        \prod_{j=1}^{\knn} c(S_{j}\mid \alpha)
        \bg(\bXs_{j}\mid\bxi). \label{eq:productprior} 
\end{equation}
A default choice for the first factor is \ech
$c(S_{j}\mid \alpha)=\alpha
\times (\abs{S_{j}}-1)!$, where $\alpha>0$ and $\abs{\cdot}$ being the
cardinality of a set, which is identical to
probability function for a random partition under the 
Chinese restaurant process \citeplatex{ferguson73}.
For the second factor,  
\citetlatex{ppmx} suggested the following default choice for
similarity functions 
\vskip -1.5ex
\begin{equation}
	\bg(\bXs_{j}\mid \bxi)=\int\prod_{i\in S_{j}}\boldq(\bX_{i}\mid \bzeta_{j})G_{0}(\bzeta_{j} \mid \bxi)\de\bzeta_{j}.\label{eq:similarityfn}
\end{equation}
\vskip -1ex
With a conjugate sampling model and prior pair of $\boldq$ and
$G_{0}$, the integral in \eqref{eq:similarityfn} is analytically
available, facilitating easy computation.  The pair is used to assess
the agreement of the data points in $S_{j}$ rather than any notion of
statistical modeling.

The model construction is concluded by specifying a sampling model for
the response variable $\bY_{i}$'s. 
Let $c_{i}=j$ if $i \in S_{j}$ denote cluster membership indicators for all $i=1,\dots,n$. 
For a given partition $\brho_{n}$, we introduce 
cluster-specific parameters
$\btheta=\{\btheta_{1},\dots,\btheta_{\knn}\}$ and assume 
\vskip -2ex
\begin{equation}
	\bY_{i}\mid \btheta, c_{i}=j \simind \h(\bY_{i}\mid \btheta_{j}), \quad
	\btheta_{j} \mid \bvarphi \simiid \Pi(\btheta_{j}\mid \bvarphi),
	\label{eq:samplingmodel}
\end{equation}
\vskip -1ex
\noindent where $\h$ is a sampling model and $\Pi(\cdot \mid
\bvarphi)$ is a prior on $\btheta_{j}$ with possible hyper-parameters
$\bvarphi$. 

Recognizing that $\bX_{i}$'s may not be random, with slight abuse of
notations,
under the similarity function \eqref{eq:similarityfn}
the PPMx can be equivalently stated as  
\begin{equation}
	\bX_{i} \mid c_{i}=j, \bzeta \simiid \boldq(\bX_{i}\mid \bzeta_{j}),\quad
	\bzeta_{j}\mid \bxi \simiid G_{0}(\bzeta_{j} \mid \bxi),\quad
        p(\brhon) \propto \prod c(S_{j} \mid \alpha).
        \label{eq:mixturemodel}
\end{equation}

\section{Missing Data in PPMx}
\label{sm sec:ppmx_missing}
Following the thread of the discussion on handling missing
data from Section \ref{subsec:CAM_cov} of the main paper, we would
like to point out that
the model never rules out the possibility of co-clustering a unit with
missing entries with fully observed units.	 
For the following argument consider
\eqref{eq:mixturemodel} with 
$$
	\bX_{i}\mid c_{i}=j, \bzeta_{j}
        =(\zeta_{j,1},\dots,\zeta_{j,p})\trans ~  \simind 
                                  \textstyle {\prod_{\ell=1}^{p} q_{\ell}(X_{i,\ell} \mid \zeta_{j,\ell})},\\
$$
that is, with $\boldq(\bX_{i}\mid \bzeta_{j})$ factoring over
covariates. 
While implementing inference using 
a Gibbs sampler, we then update the $c_{i}$ as follows

\begin{equation}
  \textstyle{\Pi(c_{i}=j \mid \bX_{i},\bzeta_{1:K}, \bc_{-i}) \propto \Pi(c_{i}=j \mid  \bc_{-i}) \times \prod_{\ell=1}^{p} q_{\ell}(X_{i,\ell} \mid \zeta_{j,\ell})},\label{eq:gibbs_step}
\end{equation}
where $\bc_{-i}$ is the set of $c_{\ell}$'s for $\ell=1,\dots,n$ excluding $c_{i}$.
      
Now consider the case where we have missing observations in some components of $\bX_{i}$ and let $\O_{i}=\{1\leq\ell\leq p:X_{i,\ell}\text{ is observed}\}$ be the indices of the observed variables in $\bX_{i}$.
In this case \eqref{eq:gibbs_step}  changes to
\begin{equation*}
	\textstyle{\Pi(c_{i}=j \mid \bX_{i},\bzeta_{1:K}, \bc_{-i}) \propto \Pi(c_{i}=j \mid  \bc_{-i}) \times   \prod_{\ell\in \O_{i}} q_{\ell}(X_{i,\ell} \mid \zeta_{j,\ell})}.
\end{equation*}
While updating the cluster membership of the units, {only the observed variables $X_{i,\ell}$'s in $\bX_{i}$} are matched with the corresponding $\zeta_{j,\ell}$ for all $\ell \in \O_{i}$.
A more detailed discussion can be found in \citelatex{ppmx_missing}.

\section{Variations of the Importance Resampling Scheme}
\label{sm sec:importance_resamping}
\subsection{Number of Patients to Resample from the RWD}
\label{sm subsec:importance_number}
Due to various reasons \citeplatex[see, e.g.,][for a review]{hey77},
in two-arm designs the allocation of patients
in the treatment and control arms are generally considered to be equal,
including in particular early-phase GBM trials
\citeplatex{centric_trial,core_trial,samplesize_gbm}.
As a rule of thumb, we thus recommend the size of the resampled
population to be equal to the treatment arm population. 

However, if desired any different ratio of sample sizes in
treatment and control arm, say $R:1$, could be used. 
In that case, even if the the distribution of the covariates in the
two arms are same after the importance resampling population
adjustment, the AUC of any classifier used in step 5 of Algorithm
\ref{algo:importance_resampling} would be $R/(R+1)$, rather than $0.5$. 

\subsection{Averaging over Multiple Resamplings}
\label{sm subsec:multiple_resamping}
It may be tempting to average over multiple, say $R$,
instances of the random importance-resampling, to remove one
source of variability. 
But this gives rise to some fundamental
problems.
For illustrative purpose, we refer to Section \ref{sec:GBM} of the
main manuscript where we discuss the application in GBM.
There we use the importance resampling strategy to generate an
equivalent subpopulation of the treatment arm and then use the Cox
proportional hazard model to test for treatment effects.
In Figure \ref{fig:pvals_GBM}, we plot the histogram of $p$-values
under the null scenario which resembles the $\Unif(0,1)$ distribution.
Now for $R$ resamplings we would have multiple $p$-values corresponding to each of the $R$ resampled populations. 
Subsequently we need a statistic to summarize the $p$-values, let us denote it by $T$.
Letting $p_{1},\dots,p_{R}$ be the $p$-values thus obtained, the {distribution of $T(p_{1},\dots,p_{R})$ will not be $U(0,1)$} anymore under the null.
We therefore recommend against it.
As importance resampling schemes are asymptotically unbiased \citeplatex{importance_resampling},
 under reasonably large sample sizes, a single resampled population should be adequate.

\section{Goodness-of-Fit Test for Continuous Responses}
\label{subsec:goodnessoffit}
We use the approach of \citetlatex{valen_model_fit_2007} to suggest a
graphical goodness-of-fit tool to validate the mixture of lognormals model for the CA-PPMx.
The procedure is valid as
long as $\h$ in \eqref{eq:sampligmodel_cappmx} is a univariate
continuous density, i.e., as long as the response variables are
univariate and continuous.
For the moment, we suppress the additional ${s}$ subindex on
$(\bX_i,Y_i)$, $i=1,\ldots,n$. Let $m(\bY\mid \bX)$ be the marginal
distribution after integrating out all model parameters
\vskip -2ex
\begin{equation*}
	m(\bY\mid \bX)= \sum_{\bc} \int \left\{\prod_{i=1}^{n}
          \h(Y_{i}\mid \btheta_{c_{i}})\right\} \de
        p(\btheta,\bc_{1:n}\mid \bX). 
\end{equation*}
\vskip -1ex \noindent
We implement a test of fit based on the following result. 
Assuming that $m(\bY\mid \bX)$ is the true marginal distribution of
$\bY$, we have:
\begin{prop}
  \label{th:goodnessoffit}
  Let $\bomega=({\btheta},{\bc}_{1:n})$ be a sample from their posterior, 
  $\H(y\mid \btheta)=\int_{-\infty}^{y} \h(z\mid \btheta) \de z$ be the CDF, 
  and $U_{i}=\H(Y_{i} \mid {\btheta}_{{c}_{i}}),~i=1,\dots,n$.
  Then, $U_{i}\simiid \Unif(0,1)$.
\end{prop}
\begin{proof}
	Let $\bu_{1:n}=\{u_{1},\dots,u_{n}\}$
	and define $A(\bu_{1:n}; {\bomega} )=\cap_{i=1}^{n} \{y: \H(y\mid {\btheta}_{ {c}_{i}} ) \leq u_{i} \}$.
	Then,
	\vskip-2ex
	\begin{equation*}
		\Pr(U_{i} \leq u_{i}\text{ for all } i=1,\dots,n) 
		= \int \int_{A(\bu_{1:n};{\bomega})}  \de  \Pi( \bomega \mid \bX,  \bY)   m( \bY\mid \bX) \de \bY.		
	\end{equation*}
	\vskip-1ex
	\noindent Note that $\Pi(\bomega \mid \bX,  \bY)= { \left\{\prod_{i=1}^{n} \h( Y_{i}\mid {\btheta}_{{c}_{i}})\right\}  \Pi(\bomega \mid \bX) }/{m( \bY\mid \bX)}$.
	Substituting this in the above equation, we get	
	\vskip-7ex
	\begin{equation*}
		\Pr(U_{i} \leq u_{i}\text{ for all } i=1,\dots,n)
		= \int  \left\{\int_{A(\bu_{1:n}; \bomega)} \prod_{i=1}^{n} \h( Y_{i}\mid {\btheta}_{ {c}_{i}}) \de Y\right\} \de \Pi(\bomega \mid \bX).
	\end{equation*}
	\vskip-2ex
	\noindent Now, the term inside the parenthesis integrates to $\prod_{i=1}^{n} u_{i}$ which is independent from $\Pi(\bomega \mid \bX)$. 
	Hence the proof.
\end{proof}

To understand the implications,
consider the distribution $(\bY,\bomega \mid \bX)$ for a
hypothetical data set $(\bX,\bY)$. 
First sample $\omegat=(\tht,\ct_{1:n})$ from $p(\bomega \mid \bX) =
p(\bc \mid \bX)\; p(\thb \mid \bc, \bX)$ and then $(\bY \mid \omegat,
\bX)$ from the sampling model \eqref{eq:sampligmodel_cappmx}.  Letting
$\wt{U}_{i} = \H(Y_{i} \mid \wt{\btheta}_{\wt{c}_{i}})$, we then have
$\wt{U}_{i}\simiid \Unif(0,1)$.  Assuming that the observed data $\bY$
do in fact arise from the assumed marginal model $m(\bY \mid \bX)$,
Proposition \ref{th:goodnessoffit} sets up sampling from the
alternative factorization
$p(\bY, \bomega \mid \bX) = m(\bY \mid \bX) \cdot p(\bomega \mid \bY,\bX)$.
It follows that  
$\wt{\bU}_{1:n}$ and $\bU_{1:n}$ are indistinguishable in distribution.
The latter, $\bU_{1:n}$, can be readily obtained from the posterior
samples of $\bomega$. 
Letting $\bU_{1:n}\mm$ denote the evaluation under the $m\th$ posterior
MCMC sample $\bomega\mm$,
a goodness-of-fit test can then be carried out to validate the uniform distribution.

Note that the $\bU_{1:n}\mm$'s vary across different posterior samples
$\bomega^{(m)}$ while also having hierarchical dependence since all of
them are  sampled conditionally on the same $\bY$ (and
$\bX$). 
Although in principle formal prior-predictive-posterior based tests be
carried out
\citeplatex{valen_model_fit_2007,cao2010_goodnessoffit},
it can be numerically infeasible for complex models
like ours.  As a practical alternative, goodness-of-fit can be
assessed by inspecting the quantile-quantile plots of
$\bU_{1:n}^{(m)}$.  
Such visual tools can be effective for
detecting departures from model assumptions \citeplatex[Chapter
2]{meloun2011}. 
We use it to assess the model fit in Section \ref{sec:GBM}.

To assess the goodness-of-fit in the GBM application, 
where the outcomes are right-censored survival data, 
we extend the result in the following corollary.

\begin{corollary}
	\label{cor:survival_outcome}
	Suppose we have right-censored survival outcomes $(Y_{i},\nu_{i})$ with covariate $\bX_{i}$ where $\nu_{i}=1$ if $Y_{i}$ is an observed failure time, for $i=1,\dots,n$.
	Following the notations of Theorem \ref{th:goodnessoffit}, 
	define $U_{i}=\H(Y_{i} \mid \btheta_{c_{i}})$ if $\nu_{i}=1$,
	else if $\nu_{i}=0$ define $U_{i}=\H(Y_{i} \mid \btheta_{c_{i}})+ \gamma_{i}\{1-\H(Y_{i} \mid \btheta_{c_{i}})\}$, 
	where $\gamma_{i}\simiid \Unif(0,1)$ independent from $Y_{i}$.	
	If the observed failure times are independent of the censoring times, 
	then $U_{i}\simiid \Unif(0,1)$.
\end{corollary}

\begin{proof}[Proof of Corollary \ref{cor:survival_outcome}]
	Let
	$\wt{ Y}_{i}$ be the true failure time of the $i\th$ individual, that is $\wt{ Y}_{i}\geq  Y_{i}$ with equality if and only if $\nu_{i}=1$.
	Letting $\wt{U}_{i}= \H( \wt{ Y}_{i} \mid \btheta_{c_{i}} )$, Theorem \ref{th:goodnessoffit} implies $\wt{\bU}_{1:n} \simiid \Unif(0,1)$. 
	Note that
	\vskip-6ex
	\begin{equation*}
		\H(\wt{ Y}_{i} \mid \btheta_{c_{i}})= 		\nu_{i} \H(\wt{ Y}_{i} \mid \btheta_{c_{i}}) + (1-\nu_{i}) \left[\H({ Y}_{i} \mid \btheta_{c_{i}} )+\{\H(\wt{ Y}_{i} \mid \btheta_{c_{i}} ) - \H({ Y}_{i} \mid \btheta_{c_{i}} )\}\right].
	\end{equation*}
	\vskip-3.5ex \noindent
	Since $\H(\wt{ Y}_{i} \mid \btheta_{c_{i}}) \sim \Unif(0,1)$
        and is independent of $ Y_{i}$,  $\H({ Y}_{i} \mid
        \btheta_{c_{i}} )+\{\H(\wt{ Y}_{i} \mid \btheta_{c_{i}} ) -
        \H({ Y}_{i} \mid \btheta_{c_{i}} )\} \mid
        Y_{i},\btheta_{c_{i}} \sim \Unif\{\H( Y_{i} \mid
        \btheta_{c_{i}}), 1\}$ which follows the same distribution as
        $\gamma_{i}\{1-\H( Y_{i} \mid \btheta_{c_{i}})\}$. 
	Hence the proof.
\end{proof}

\subsection{Illustrating Example for the Graphical Goodness-of-Fit Test}
\label{sm subsec:example_testoffit}
We illustrate the Bayesian goodness-of fit test in a linear regression problem.
We simulate data $(\bX_{i},Y_{i})$, $i=1,\dots,n~(=1,000)$ from the following mixture distribution
\vskip-2ex
\begin{equation}
	Y_{i}\mid \bX_{i} \simind \pi_{0} \mn(\alpha_{0}+ \bbeta_{0}\trans \bX_{i},\sigma_{0}^{2}) +(1-\pi_{0}) \Exp(\alpha_{0}+ \bbeta_{0}\trans \bX_{i}),	\label{sm eq:true_reg}
\end{equation}
\vskip-1ex
\noindent where $\bX_{i}$'s are $p~(=5)$-variate continuous covariates and
 $\Exp(a)$ denotes an exponential distribution with mean $a$.
However, we fit the following misspecified Bayesian linear regression model on the data using the \texttt{MCMCpack R} package
\vskip-5ex
\begin{align}
	&\text{\textbf{likelihood: }} Y_{i}\mid \bX_{i} \simind \mn(\alpha+ \bbeta\trans \bX_{i},\sigma^{2});\nonumber\\
	&\text{\textbf{prior: }} (\alpha,\bbeta) \sim \mn_{p+1}(\bzero,10\times \bI_{p+1}),~ \sigma^{-2} \sim \Ga(0.1, 0.1).\label{sm eq:bayes_reg}
\end{align}
\vskip-2ex
For varying values of $\pi_{0}$, we show quantile-quantile plots in Figure \ref{sm fig:sim_testoffit} where we see deviation from the diagonal $y=x$ straight-line aggravates as $\pi_{0}\to 0$, i.e., with increasing model misspecification. 
\begin{figure}[ht]
	\centering
	\includegraphics[width=.9\linewidth]{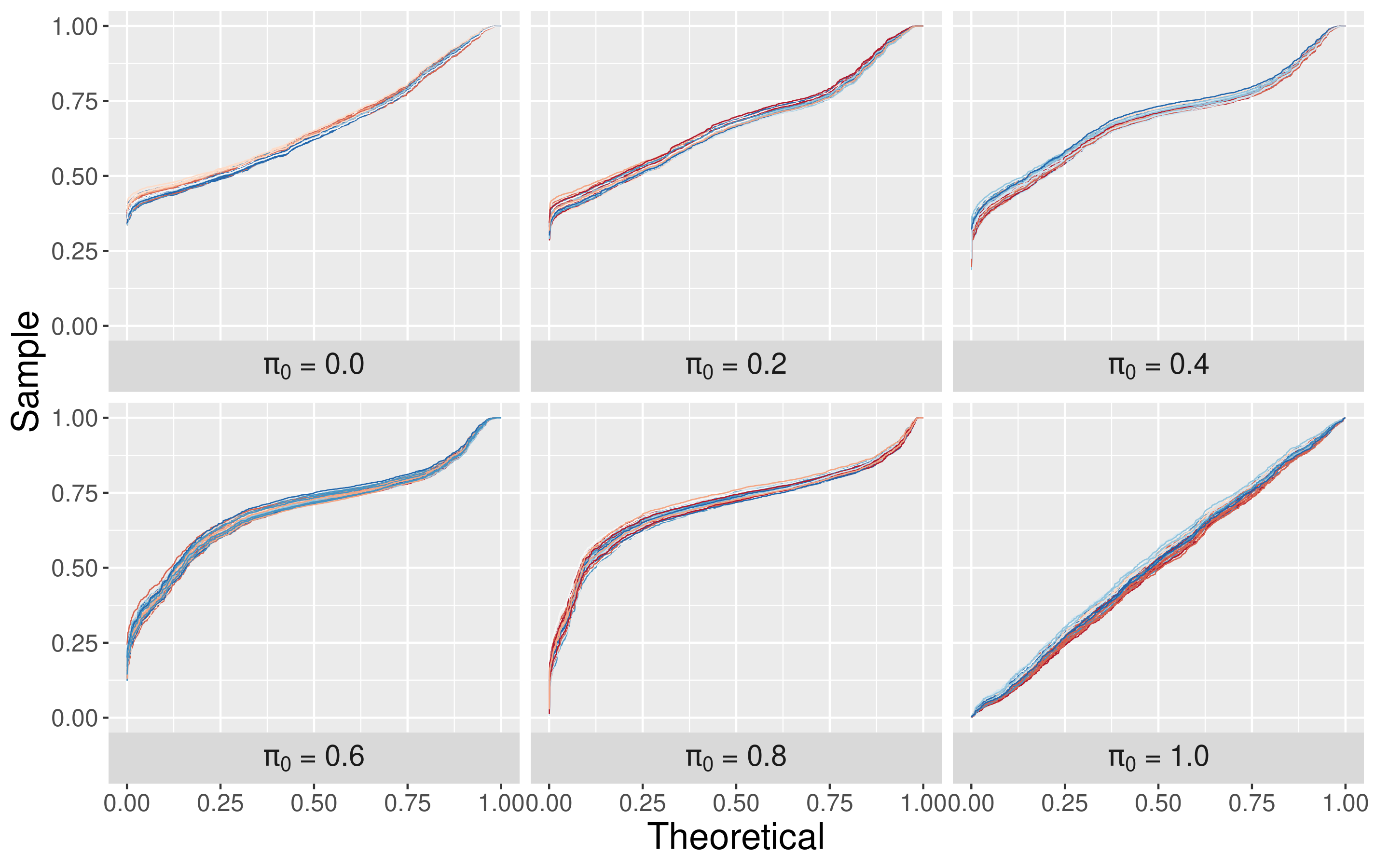}
	\caption{Quantile-quantile plots for increasing model misspecification:  Data are generated from model \eqref{sm eq:true_reg} for different values of $\pi_{0}$  and the Bayesian linear regression model in Eqn \eqref{sm eq:bayes_reg} is fitted where $\pi_{0}=0.0$ and $1.0$ denote the extreme misspecified model and the true model, respectively.
	Deviation from the diagonal $y=x$ straight-line aggravates with increasing model misspecification.}
	\label{sm fig:sim_testoffit}
\end{figure}

\section{Alternate Interpretation of the CA-PPMx}
\label{sec:SM-PS}
In Section \ref{subsec:treatment_effect}, we introduced a model-based approach for inference on
treatment effects in the CA-PPMx model.
An alternative interpretation of the approach arises from observing
the following connection with methods based on PS stratification
\citeplatex{propensity_power_prior, propensity_CL, propensity_multiple}.
The CAM model can be interpreted as a stochastic PS stratification.
To see this, first re-index all patients and patient specific
variables across $s=1,2$ as $i=1,\ldots,N=n_{1}+n_{2}$ and define $Z_{i} \in \{1,2\}$
if patient $i$ was originally in data set $s=1$ or $2$, respectively.
Assuming equal sample sizes $n_{1}=n_{2}$, 
we 
have $p(Z_{i}=1 \mid c_{i}=j)/p(Z_{i}=2 \mid c_{i}=j) = \pi_{1,j}/\pi_{2,j}$.
That is, the terms in the CAM model correspond to different PS ratios for
the selection of a patient into $s=1$ versus $s=2$.
Grouping patients in clusters $C_{j}$ is then interpreted as
stratification by PS, with clusters $C_{j}$ defining the strata.
Within each stratum we report treatment effect
$\delta_{j} = {\delta\{\h(Y\mid \btheta_{1,j}), \h(Y\mid
	\btheta_{2,j})\}}$.
Compare the discussion in Section \ref{subsec:treatment_effect}.  

Whereas fixed consolidated unidimensional PSs may be inadequate in
matching multivariate covariates
\citeplatex{propensity_limitation1,propensity_limitation2} and hence
sensitive to the specification of the PS model \citeplatex{zhao2004},
inference under the proposed CAM model overcomes limitations by
naturally including uncertainty in the stratification.

\section{CA-PPMx Specifications and Hyperparameters}
\label{sm sec:hyperparam}
Recall the setup from Section \ref{subsec:CAM_cov} and the notations from Eqn \eqref{eq:model_missingdata}.
For categorical covariate $X_{s,\ell}$ with categories $1,\dots,m_{\ell}$,
 we choose $q_{\ell}(X_{s,\ell} \mid \bzeta_{\ell})= \Mult(1; \zeta_{\ell,1},\dots,\zeta_{\ell,m_\ell})$
 and $g_{0,\ell}(\zeta_{\ell,1},\dots,\zeta_{\ell,m_\ell})=\Dir(1,\dots,1)$  {to choose a uniform distribution over the simplex}. 
 For continuous $X_{s,\ell}$, we choose $q_{\ell}(X_{s,\ell} \mid \bzeta_{\ell})=\mn(X_{s,\ell}; \mu_{X,\ell}, \sigma_{X,\ell}^{2})$ with $\bzeta_{\ell}=(\mu_{X,\ell}, \sigma_{X,\ell}^{2})$ and $g_{0,\ell}(\mu_{X,\ell}, \sigma_{X,\ell}^{2})=\mathrm{NIG}(\mu_{X,\ell}, \sigma_{X,\ell}^{2};0,1,\alpha_{X},1)$, i.e., $\mu_{X,\ell} \mid \sigma_{X,\ell}^{2} \sim \mn(0, \sigma_{X,\ell}^{2})$, $\sigma_{X,\ell}^{-2}\sim \Ga(a_{X},1)$.
Following standard practice, we center $\mu_{X,\ell}$
around zero.  Based on previous experience on Gaussian mixture models,
we set $a_{X}= \# \text{continuous covariates}+30$, as a small prior
variance on $\sigma^{2}_{X,\ell}$'s favors
a larger number of occupied clusters in the mixture model a
posteriori,
allowing for a more flexible fit.
Recall that we have assumed $\log \alpha_{s}\sim
\mn(\mu_{\alpha},\sigma_{\alpha}^{2})$ for $s=1,2$ on the
concentration parameters in models \eqref{eq:RWD_model} and
\eqref{eq:experimental_model}.  To specify weakly informative priors,
we set the hyperparameters $\mu_{\alpha}$ and $\sigma^{2}_{\alpha}$
such that $\eE(\alpha_{s})=1$ and $\var(\alpha_{s})=10$ a priori for
$s=1,2$.
 
Regarding the parameters of the sampling model for survival outcomes
in Eqn \eqref{eq:mixture_lognormal}, we set $\kappa_{0}=1$ and
$a_{0}=10$ to ensure a thin-tailed base-measure.
In our experience, 
with too heavy tailed prior distributions, small sample performance
can easily get dominated by the prior.
Regarding the hyperprior on the mean parameter
$\mu_{0}$, we choose $m_{\mu}$ using an empirical Bayes type approach.
Letting $\wt{n}$ be the number of observed failures combining the RWD and the current trial,
we set $m_{\mu}= \frac{1}{\wt{n}}\sum_{s}\sum_{i:\nu_{s,i}=1}Y_{s,i}$, i.e., the grand mean of the $\log$-observed failure times across all arms.
We further set $s_{\mu}^{2}=1$.
Regarding the hyperprior on the scale parameter $b_{0}$, we choose
$m_{b}$ and $s_{b}^{2}$ such that $\eE(b_{0})=5$ and $\var(b_{0})=20$
a priori to set a weakly informative hyperprior. 
 
 Regarding the real-valued continuous responses in the simulation studies in Section \ref{sec:sim_study}, we use the model in Eqn \eqref{eq:mixture_lognormal} on the actual response variables with $\nu_{s,i}=1$ for all $i$ and $s$.
 

\section{Posterior Computation}
\label{SM subsec:MCMC}

For computational convenience in the practical implementation, we
consider the \textit{degree $k$ weak limit approximation}
\citeplatex{Ishwaran_Zarepour02b, Ishwaran_Zarepour02c} of the
$\GEM(\alpha_{2})$ distribution in \eqref{eq:RWD_model}, i.e., we
use a $\Dir(\alpha_{2}/k, \dots, \alpha_{2}/k)$ distribution,
with fixed but large enough $k$. 
We set $k=15$ for all our simulation experiments and applications.

We develop a Gibbs sampler to avoid computational
issues with a Gaussian mixture models on the $\log$ transformed survival outcomes with censoring.  
Without loss of generality we assume $ Y_{s,i}$'s ($\log$ transformed outcomes) are supported on the entire
real line and describe our algorithm for a mixture of Gaussian distributions.
Let $\nu_{s,i}$'s be the censoring indicators
such that $\nu_{s,i}=1$ implies $ Y_{s,i}$ is an observed failure
time; else if it is censored in the interval $( Y_{s,i,l}, Y_{s,i,u})$
then $\nu_{s,i}=0$.
  For left and right censoring, we take $
Y_{s,i,u}=\infty$ and $ Y_{s,i,l}=-\infty$, respectively.  Let $\wt{
Y}_{s,i}$ be the true failure times, that is $\wt{ Y}_{s,i}={
Y}_{s,i}$ if and only if $\nu_{s,i}=1$.  
Off-line, before starting MCMC simulation, we
initialize $\wt{ Y}_{s,i}$ at some admissible value for $\nu_{s,i}=0$
and
cluster membership indicator variables $\bcon$ and $\bctw$.
For the CAM model on covariates, we consider a conjugate pair
$q_{\ell}$ and $g_{0,\ell}$ for $\ell=1,\dots,p$. This allows us to analytically marginalize
with respect to the atoms $\bzeta_{j}$'s. 
This strategy results in substantially improved mixing of the Markov chain.

The sampler iterates through the following steps.
In Step 1, we impute $\wt{Y}_{s,i}$'s for the censored observations;
in Step 2, we update the cluster membership indicators $\bcon$ and $\bctw$;
in Step 3, we update hyper-parameters related to the response model that allows sharing of information via a hierarchical model; 
in Step 4, we update the parameters required to implement the strategies outlined in Sections \ref{subsec:IS_reweighing} and \ref{subsec:treatment_effect};
finally in Step 5 we update the Dirichlet hyperparameters for the two mixture models.

\begin{description}[leftmargin=0pt]	
	\item[Step 1] 
	We define the set $S_{s,j,-i}=\{i: c_{s,i}=j\}\setminus \{i\}$, 
	$n_{s,j,-i}=\abs{S_{s,j,-i}}$,
	$\kappa_{s,j,-i}=\kappa_{0}+ n_{s,j,-i}$, 
	$\wbar{ Y}_{s,j,-i}=\sum_{r\in S_{s,j,-i} }\wt{ Y}_{s,r}/ n_{s,j,-i} $,
	$\mu_{s,j,-i}=(\kappa_{0}\mu_{0}+ n_{s,j,-i} \wbar{ Y}_{s,j,-i} )/ \kappa_{s,j,-i}$, 
	$a_{s,j,-i}=a_{0}+ n_{s,j,-i}/2$, 
	$b_{s,j,-i}=b_{0}+ \sum_{r\in S_{s,j,-i} }(\wt{ Y}_{s,r}- \wbar{ Y}_{s,j,-i})^{2}/2+ n_{s,j,-i}\kappa_{0}(\wbar{ Y}_{s,j,-i}-\mu_{0})^{2}/\kappa_{s,j,-i}$.
	Then for all $i=1,\dots,\ns$ and $s=1, 2$, generate
	\vskip-3ex
	\begin{equation*}
		\wt{ Y}_{s,i}\sim\begin{cases}
			{ Y}_{s,i}\text{ with probability 1 if } \nu_{s,i}=1;\\
			t_{2 a_{s,j,-i}} \left\{\mu_{s,j,-i}, \frac{b_{s,j,-i}(\kappa_{s,j,-i}+1)}{a_{s,j,-i} \kappa_{s,j,-i}} \mid  ( Y_{s,i,l}, Y_{s,i,u}) \right\},
		\end{cases}
	\end{equation*}
	\vskip-2ex
	where $t_{df}\{\mu, \sigma^{2} \mid (a,b)\}$ is a central Student's $t$-distribution, with degrees of freedom $df$, median $\mu$ and scale parameter $\sigma$, truncated to the set $(a,b)$.
	
	\item[Step 2] Letting $f_{t}\{\cdot\mid df, \mu, \sigma^{2} \}$ and $F_{t}\{\cdot\mid df, \mu, \sigma^{2} \}$ denote the pdf and cdf of a central Student's $t$-distribution with degrees of freedom $df$, median $\mu$ and scale parameter $\sigma$, respectively, we define
	\vskip-7ex
	\begin{equation*}
		\psi_{Y;s,j}(i)=
		\begin{cases}
			f_{t}\left\{ Y_{s,i} \mid 2 a_{s,j,-i}, \mu_{s,j,-i}, \frac{b_{s,j,-i}(\kappa_{s,j,-i}+1)}{a_{s,j,-i} \kappa_{s,j,-i}}\right\} \text{ if }
			\nu_{s,i}=1;\\
			F_{t}\left\{ Y_{s,i,u} \mid 2 a_{s,j,-i}, \mu_{s,j,-i}, \frac{b_{s,j,-i}(\kappa_{s,j,-i}+1)}{a_{s,j,-i} \kappa_{s,j,-i}} \right\}\\ \qquad\qquad- F_{t}\left\{ Y_{s,i,l} \mid 2 a_{s,j,-i}, \mu_{s,j,-i}, \frac{b_{s,j,-i}(\kappa_{s,j,-i}+1)}{a_{s,j,-i} \kappa_{s,j,-i}}\right\} \text{ otherwise}.
		\end{cases}
	\end{equation*}	
	\vskip-3ex
	Recall from Section \ref{subsec:CAM_cov} (see page \pageref{eq:model_missingdata})  that $\O_{s,i}$ is the set of indices of the covariates observed for $\bX_{s,i}$,
	and define the sets $\C_{j,\ell}= \cup_{s=1}^{2} \left\{i:i\in S_{j},\ell\in \O_{s,i} \right\}$ and $\bX^{*o}_{j,\ell}=\cup_{s=1}^{2} \left\{ X_{s,i,j}:i\in \C_{j,\ell}\right\}$.
	Define the functions $g_{\ell}(\bX^{*o}_{j,\ell} \mid \bxi_{\ell})=\int \prod_{i\in \C_{j,\ell}} q_{\ell}(X_{s,i,\ell}\mid \bzeta_{j,\ell}) g_{0,\ell} ( \bzeta_{j,\ell} \mid \bxi_{\ell})\de\bzeta_{j,\ell}$ and 
	$\psi_{X;s,j}(i)=\prod_{\ell\in \O_{s,i} }  \frac{g_{\ell}(\bX^{*o}_{j,\ell} \mid \bxi_{\ell})} {g_{\ell}\left[\bX^{*o}_{j, \ell}\setminus \left\{ X_{s,j,\ell} \right\} \mid \bxi_{\ell}\right] } $.	
	Then, $\bcon$ can be updated as 
	\vskip-2.5ex
	\begin{equation*}
		\Pi(c_{1,i}=j\mid -) \propto (n_{1,j,-i}+{{\alpha_{1}}}/{\tilkn}) \times \psi_{Y;1,j}(i)\times \psi_{X;1,j}(i)\text{ for } j=1,\dots,\tilkn.
	\end{equation*}
	\vskip-2.5ex
	Similarly $\bctw$ can be updated as\\
	$~~~~~~~\Pi(c_{2,i}=j\mid -)=1$ if $n_{1,j}>0$ and $n_{2,j,-i}=0$;\\
	else $\Pi(c_{2,i}=j\mid -) \propto (n_{2,j,-i}+\alpha_{2}/k) \times \psi_{Y;2,j}(i)\times \psi_{X;2,j}(i)$ for $j=1,\dots,k$.
	
	\item[Step 3] \label{lognorm_hyper} Define $\wt{b}=\log b_{0}$ and
	let $\Pi(\mu_{0}, \wt{b}\mid \wt{\bY}_{1,1:n_{1}}, \wt{\bY}_{2,1:n_{2}})$ be the joint posterior density of $\mu_{0}$ and $\wt{b}$ given $\wt{Y}_{s,i}$'s, $k_{n,1}$ and $k_{n,2}$ be the number of non-empty clusters in the two cohorts respectively.
	Then,
	\vskip-7.5ex
	\begin{multline*}
		\log\Pi(\mu_{0}, \wt{b}\mid \wt{\bY}_{1,1:n_{1}}, \wt{\bY}_{2,1:n_{2}})=K -\frac{(\mu_{0}- m_{\mu})^{2}}{2s_{\mu}^{2}}  -\frac{(\wt{b}- m_{b})^{2}}{2s_{b}^{2}}+ (k_{n,1}+k_{n,2})a_{0}\wt{b}  \\
		-\sum_{j=1}^{\kn} \sum_{s=1}^{2}
		\left(a_{0}+\frac{n_{s,j}}{2}\right)\log\left[e^{ \wt{b}} +\frac{1}{2} \left\{\mu_{0}^{2}\kappa_{0}+\sum_{i\in S_{s,j}} \wt{Y}_{s,i}^{2} -\frac{(\kappa_{0}\mu_{0} +n_{s,j} {\wbar{Y}}_{s,j})^{2} }{\kappa_{0}+n_{s,j}} \right\} \right]  ,
	\end{multline*}
	\vskip-4ex
	where $K$ is a constant and ${\wbar{Y}}_{s,j}=\sum_{i\in S_{s,j}} \wt{Y}_{s,i}$.
	We sample $\mu_{0}$ and $\wt{b}$ using a Hamiltonian Monte Carlo (HMC) algorithm \citeplatex{duane_hmc}.
	
	\item[Step 4] For $j=1,\dots, k$, we define the set $S_{s,j}=\{i: c_{s,i}=j\}  $, 
	$\kappa_{s,j}=\kappa_{0}+ n_{s,j}$, 
	$\mu_{s,j}=(\kappa_{0}\mu_{0}+ n_{s,j} \wbar{Y}_{s,j} )/ \kappa_{s,j}$, 
	$a_{s,j}=a_{0}+ n_{s,j}/2$, 
	$b_{s,j}=b_{0}+ \sum_{r\in S_{s,j} }(\wt{Y}_{s,r}- \wbar{Y}_{s,j})^{2}/2+ n_{s,j}\kappa_{0}(\wbar{Y}_{s,j}-\mu_{0})^{2}/\kappa_{s,j}$.
	Then, 
	\vskip-6ex
	\begin{align}
		&\mu_{s,j} \sim t_{2a_{s,j}} \left\{\mu_{s,j},\frac{b_{s,j}(\kappa_{s,j}+1)  } {a_{s,j} \kappa_{s,j}}\right\}, \qquad \sigma_{s,j}^{-2} \sim \Ga(a_{s,j}, b_{s,j}), \label{eq:sim_responseparams}\\
		&\bpion \sim \Dir \left(n_{1,1}+\frac{\alpha_{1}}{\tilkn}, \dots, n_{1,\tilkn}+ \frac{\alpha_{1}} {\tilkn} \right),~~ 
		\bpitw \sim \Dir \left(n_{2,1}+\frac{\alpha_{2}}{k}, \dots, n_{2,k}+ \frac{\alpha_{2}}{k} \right).\nonumber
	\end{align}
	\vskip-4ex
	For $s=1$, we only sample for $j=1,\dots,\tilkn$ in \eqref{eq:sim_responseparams}.
	Note that the dimension of $\bpion$ can vary across MCMC samples.
	
	\item[Step 5] 
	With lognormal priors on the Dirichlet mixture hyperparameters $\alpha_{1}$ and $\alpha_{2}$, $\log \alpha_{s} \sim \mn(\mu_{\alpha}, \sigma_{\alpha}^{2})$, $s=1,2$, the log-posterior pdfs are given by
	\begin{align*}
		&\log \Pi(\alpha_{1}\mid -) = K_{1} + \log \frac{\Gamma(\alpha_{1})}{\Gamma(\alpha_{1}+n_{1})} + \sum_{j: n_{1,j}>0 } \log \frac{\Gamma(\alpha_{1}/k(n_{2}) +n_{1})}{\Gamma(\alpha_{1})} - \log \alpha_{1}- \frac{(\log\alpha_{1}-\mu_{\alpha})^{2} } {2 \sigma_{\alpha}^{2}},\\
		&\log \Pi(\alpha_{2}\mid -) = K_{2} + \log \frac{\Gamma(\alpha_{2})}{\Gamma(\alpha_{2}+n_{2})} + \sum_{j: n_{2,j}>0 } \log \frac{\Gamma(\alpha_{2}/k +n_{2})}{\Gamma(\alpha_{2})} - \log \alpha_{2}- \frac{(\log\alpha_{2}-\mu_{\alpha})^{2} } {2 \sigma_{\alpha}^{2}}.
	\end{align*}
	As the respective pdfs are differentiable with respect to $\alpha_{1}$ and $\alpha_{2}$, we sample the parameters using HMC.
\end{description}

\begin{remark}	
	Note that in Step 2, $\C_{j,\ell}$ is the set of data points in $S_{j}$ with observed covariate $\ell$, 
	$\bX_{j,\ell}^{*o}$ is the collection of the observed values of the covariate $\ell$ in $S_{j}$ and 
	$g_{\ell}(\bX_{j,\ell}^{*o} \mid \bxi_{\ell})$ is the joint
        marginal density. 
	A conjugate pair $q_{\ell}$ and $g_{0,\ell}$ ensures the analytical availability of $g_{\ell}$ and
	$\psi_{X;s,j}(i)$ becomes the conditional distribution of $\bX_{s,i}$ given $\bX_{j,\ell}^{*o}$.
	For continuous real-valued $X_{s,j,\ell}$, we may take
        $q_{\ell}(\cdot \mid \bzeta_{j})$ to be the univariate
        Gaussian pdf where $\bzeta_{j}$ is the set of associated mean
        and variance parameters, and $g_{0,\ell}(\bzeta_{j} \mid
        \bxi_{\ell})$ to be a normal-inverse-gamma density
        (compare Section \ref{sm sec:hyperparam}).
	In this case, the ratio $\frac {g_{\ell}(\bX^{*o}_{j,\ell} \mid \bxi_{\ell})} {g_{\ell}\left[\bX^{*o}_{j, \ell}\setminus \left\{ X_{s,j,\ell} \right\} \mid \bxi_{\ell}\right] }$ 
	reduces to a central $t$-distribution density;
	for categorical $X_{s,j,\ell}$, a convenient choice can be the multinomial-Dirichlet pair which again yields an analytical expression of the ratio.
\end{remark}

In the GBM application and simulation studies in Section \ref{sec:sim_study}, we have considered conjugate normal-inverse-gamma  and multinomial-Dirichlet conjugate pairs for continuous real-valued covariates and categorical covariates, respectively.
For all simulation studies and GBM application, we consider 6,000 MCMC iterations, discarded the first 1,000 as the burn-in samples, 
and saved every $5\th$ MCMC sample to reduce autocorrelation. 

Finally we note that the complete conditional for $\pi_{1,j}$ in step 4
could be used to implement Rao-Blackwellization
\citeplatex{robert2021rao} in the evaluation of the weights $w_{i}$ in
\eqref{eq:imp_dist} by replacing $\pi_{1,j}$ with the conditional
posterior means. 

\section{Additional Details on Simulation Studies}
\label{sm sec:simulations}

\subsection{Procedure to Test for Treatment Effects in Section \ref{sec:sim_study}}
\label{sm subsec:sim_testing_procedure}
Recall that in Section \ref{sec:sim_study} 
we test $H_{0}:\delta=0$ versus $H_{1}:\delta \neq 0$ in each simulation setup.	
To compute the power, we first estimate the treatment effect, say $\deltah$ in each setup.
Estimated treatment effects under CA-PPMx are evaluated using the
posterior mean of Eqn \eqref{eq:Delta3}.
To evaluate type-II error rates we use the empirical
distribution of $\deltah$ under simulation truth $\delta=0$
for each of the seven methods under consideration 
across the 500 repeat simulations to obtain their distributions under
$H_{0}$.  
We evaluate the empirical $2.5\%$ and $97.5\%$ quantiles,
say $\deltah_{L}$ and $\deltah_{U}$ and define the test function
$\Phi(\deltah)=\mathbbm{1}_{\deltah \notin[\deltah_{L},\deltah_{U}]}$
controlling the type-I error at 5\% level of significance.

\subsection{Details on Simulation Truths}
\label{sm subsec:simulation_truth}
\paragraph{\uline{CAM scenario:}} 
We set $\mu_{1,1}=\mu_{1,1}=2$ and $\mu_{1,j}=0$ for all $j>2$, and
$\mu_{2,5}=\mu_{2,6}=2$ and $\mu_{2,j}=0$ for all $j\notin \{5,6\}$, 
$\sigma_{j}^{2}=0.05$ for all $j=1,2,3$.
Regarding the mixture weights, we set $\pi_{1,1}=\pi_{1,2}=0.5$ and $\pi_{2,1}=\pi_{2,2}=1/6$ and $\pi_{2,3}=2/3$.
Regarding the categorical covariates we set
 $\varrho_{1}=0.85$, $\varrho_{2}=0.65$.

\paragraph{\uline{MIX scenario:}} We take $k=4$. 
Recall that $\bmu_{1,j}=\bmu_{2,j}$ for all $j<k$, say $\bmu_{j}=(\mu_{j,1},\dots,\mu_{j,p})\trans$.
For each $j<k$, we take $\mu_{j,2j+1}=\mu_{j,2j+2}=2$ and $\mu_{j,\ell}=0$ for all $\ell \notin \{2j+1,2j+2\}$.
Finally for $\bmu_{s,k}=(\mu_{s,k,1},\dots,\mu_{s,k,p})\trans$ with $s=1,2$,
we set 
$\mu_{1,k,7}=\mu_{1,k,8}=2$, $\mu_{1,k,9}=1$ and $\mu_{1,k,\ell}=0$ for all $\ell\notin \{7,8,9\}$;
and 
$\mu_{2,k,2k+1}=\mu_{2,k,2k+2}=2$ and $\mu_{2,k,\ell}=0$ for all $\ell\notin \{2k+1,2k+2\}$.
In each repeat simulation we generate $w_{1,1},\dots,w_{1,k} = \mathrm{SRSWR}_{k}(1,\dots,4)$ where $\mathrm{SRSWR}_{r}(\S)$ denotes the simple random sampling scheme with replacement of size $r$ from the set $\S$.
Then we set $\pi_{1,j}=w_{1,j}/ \sum_{r=1}^{k}w_{1,r}$ for all $j=1,\dots,k$.
we set $\pi_{2,j}=1/ k$ for all $j=1,\dots,k$.

\paragraph{\uline{Interaction scenario:}}
Recall the covariates in the GBM dataset from Table \ref{tab:GBM_covs} in the main manuscript.
We consider pairwise interactions between (Gender, Age) and (RT Dose, Age).
Following that, we have one-hot-encoded the covariates with more than
two categories
(e.g., KPS) so that we are left with all binary covariates (including the interactions).
Let $\bX_{i}=(X_{i,1},\dots,X_{i,p})\trans$ be the covariates corresponding to patient record $i$ with $p$ being the number of covariates.

For each repeat simulation, we then generate
$\bb=(b_{1},\dots,b_{p})\trans = \mathrm{SRSWR}_{p}(-1,0.75)$.
We then assign the patient record $i$ to the treatment arm with
probability $\frac{\bX_{i}\trans \bb +0.8}{1+ \bX_{i}\trans \bb
  +0.8}$. 

\paragraph{\uline{Oracle scenario:}} We follow the exact same
strategy as described in the Interaction scenario but without pairwise
interactions. 

\paragraph{\uline{Outcome model:}} For $\bx=(x_{1},\dots,x_{p})\trans$, we take $f(\bx)= \beta_{1} \mathbbm{1}_{(x_{1}\geq 1.25, x_{2}\geq 1.25)}- 
\beta_{2} \mathbbm{1}_{(x_{3}\geq 1.25, x_{4}\geq 1.25)} + 
\beta_{3} \mathbbm{1}_{(x_{5}\geq 1.25, x_{6}\geq 1.25)} + 
\beta_{4} \mathbbm{1}_{(x_{p-1}\geq 1, x_{p}\geq 1)}$.
In each repeat simulation we let $\beta_{1}, \beta_{2}\simiid \Unif(40,60)$,
$\beta_{3}\sim \Unif(225,275)$ and $\beta_{4}\sim \Unif(-5,-1)$.

In the Interaction and Oracle scenarios we simulate the linear regression coefficients $\bbeta=(\beta_{1},\dots,\beta_{p})\trans\simiid \Unif(-10,10)$.

\subsection{Implementation of Matching and PS-Based Approaches}
\label{sm subsec:alternative_implementation}

\paragraph{\uline{PS-based approaches:}} 
We implemented the composite likelihood and power-prior approaches using the \texttt{psrwe R} package. 
We set the hyperparameters as recommended in the vignette.
We create 5 strata (suggested in the package vignette) and borrow $n_{1}$ patients from the RWD for all simulation studies.
For the PS model, we consider both, linear logistic regression and the random forest classifier.

\paragraph{\uline{Matching:}} 
We implemented these approaches using the \texttt{optmatch R} package. 
Following the recommendations in the vignette, we set one control to be matched to each treatment.
It makes the matched control population to be of the same size as the treatment arm.
We then fit a linear model 
to estimate the treatment effect $\delta$.

\subsection{Bias for the Methods Considered in Section \ref{sec:sim_study}}
\label{sm subsec:bias}
\begin{figure}[H]
	\begin{subfigure}[b]{\linewidth}
		\includegraphics[width=\linewidth]{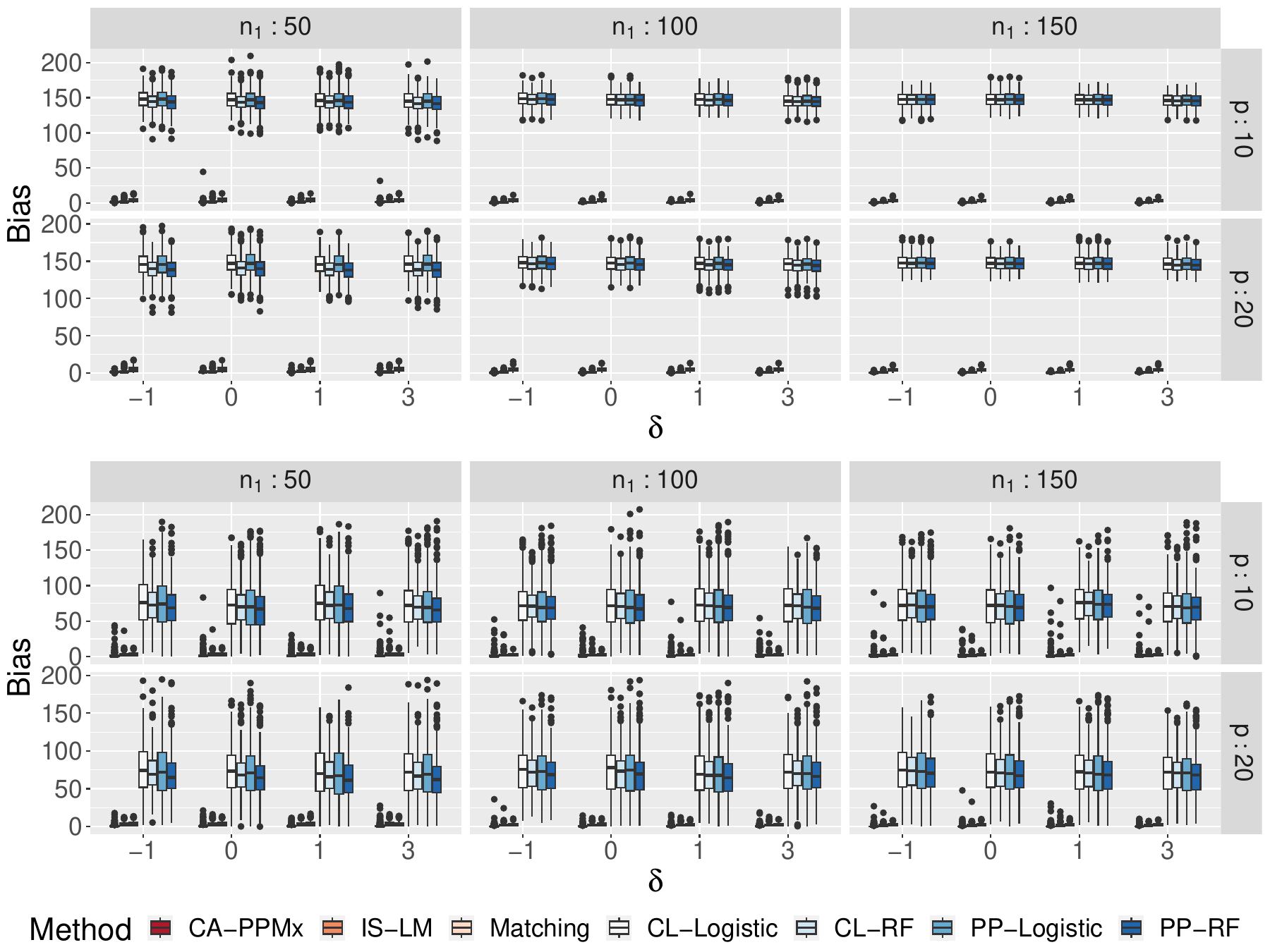}
		\caption{CAM (top) and MIX (bottom) scenarios.}
		\label{sm fig:mixture}
	\end{subfigure}
\end{figure}
\vspace*{-2ex}
\begin{figure}[H]
	\centering	
	\ContinuedFloat
	\begin{subfigure}[b]{\linewidth}
		\includegraphics[width=\linewidth]{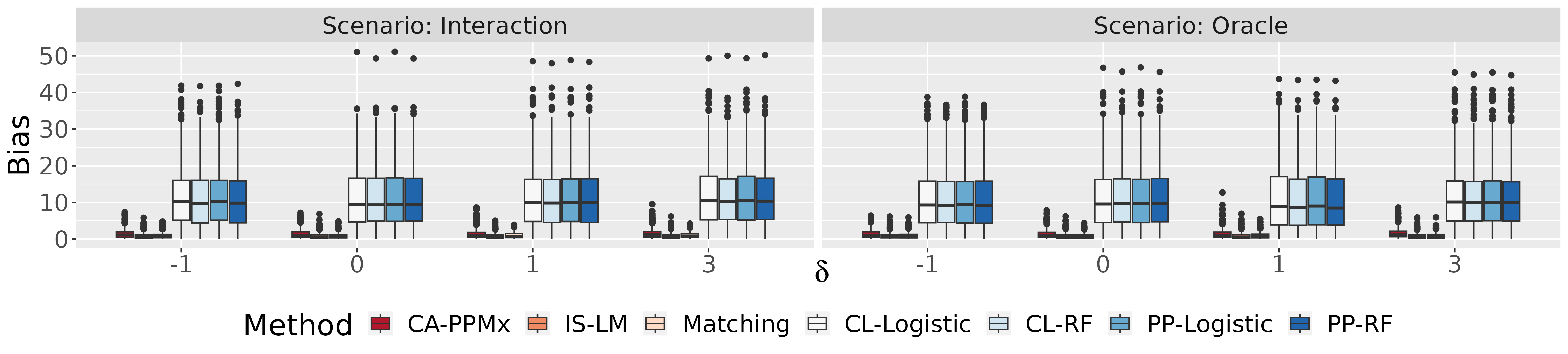}
		\caption{Interaction (left) and Oracle (right) scenarios.}
		\label{sm fig:gbm_bias}
	\end{subfigure}
	\caption{The bias in detecting treatment effects across different simulation setups:		
		{Seven methods are used to estimate the effects where
			IS-LM and CA-PPMx are based on the proposed CAM model. 
			Panel (a) corresponds to the CAM (top) and MIX (bottom) scenarios.
			Panel (b) shows results under the Interaction (left side) and
			the Oracle (right side) scenarios.} }
	\label{sm fig:Biases}
\end{figure}

\subsection{Power for the Methods Considered in Section \ref{sec:sim_study}}
\label{sm subsec:power_prior}
 {The PS-based approaches yield very similar results.
Therefore, for easier apprehension we only show the results for CL-RF together with the other types of methods in Tables \ref{tab: sim_sim} and \ref{tab: sim_gbm}, and the rest of the PS-based methods in Table \ref{sm tab:ps_sim}.}
\begin{table}[!b]
	\centering
	\caption{Power of detecting treatment effects under CAM and MIX scenarios}
	\label{tab: sim_sim}
	\resizebox{.625\textwidth}{!}{ \begin{minipage}[b]{.99\linewidth}			
	\begin{tabular}{c|c|llr|llr||c|llr|llr|}
		\toprule
		\multicolumn{1}{c|}{\multirow{2}[4]{*}{\textbf{$\delta$}}} &       & \multicolumn{3}{c|}{\textbf{Scenario: CAM}} & \multicolumn{3}{c||}{\textbf{Scenario: MIX}} &       & \multicolumn{3}{c|}{\textbf{Scenario: CAM}} & \multicolumn{3}{c|}{\textbf{Scenario: MIX}} \\
		\cmidrule{2-15}    \multicolumn{1}{c|}{} &       & \multicolumn{1}{c}{$n_{1}$} & \multicolumn{1}{c}{$p$} & \multicolumn{1}{c|}{\textbf{Power}} & \multicolumn{1}{c}{$n_{1}$} & \multicolumn{1}{c}{$p$} & \multicolumn{1}{c||}{\textbf{Power}} &       & \multicolumn{1}{c}{$n_{1}$} & \multicolumn{1}{c}{$p$} & \multicolumn{1}{c|}{\textbf{Power}} & \multicolumn{1}{c}{$n_{1}$} & \multicolumn{1}{c}{$p$} & \multicolumn{1}{c|}{\textbf{Power}} \\
		\midrule
		\multirow{6}[2]{*}{-1} & \multirow{24}[8]{*}{\begin{sideways}\textbf{Method: CA-PPMx}\end{sideways}} & 50    & 10    & 0.024 & 50    & 10    & 0.056 & \multirow{24}[8]{*}{\begin{sideways}\textbf{Method: IS-LM}\end{sideways}} & 50    & 10    & 0.054 & 50    & 10    & 0.042 \\
		&       & 100   & 10    & 0.032 & 100   & 10    & 0.066 &       & 100   & 10    & 0.080 & 100   & 10    & 0.090 \\
		&       & 150   & 10    & 0.048 & 150   & 10    & 0.118 &       & 150   & 10    & 0.080 & 150   & 10    & 0.072 \\
		&       & 50    & 20    & 0.206 & 50    & 20    & 0.052 &       & 50    & 20    & 0.050 & 50    & 20    & 0.056 \\
		&       & 100   & 20    & 0.040 & 100   & 20    & 0.050 &       & 100   & 20    & 0.068 & 100   & 20    & 0.052 \\
		&       & 150   & 20    & 0.062 & 150   & 20    & 0.030 &       & 150   & 20    & 0.118 & 150   & 20    & 0.064 \\
		\cmidrule{1-1}\cmidrule{3-8}\cmidrule{10-15}    \multirow{6}[2]{*}{0} &       & 50    & 10    & 0.050 & 50    & 10    & 0.050 &       & 50    & 10    & 0.050 & 50    & 10    & 0.050 \\
		&       & 100   & 10    & 0.050 & 100   & 10    & 0.050 &       & 100   & 10    & 0.050 & 100   & 10    & 0.050 \\
		&       & 150   & 10    & 0.050 & 150   & 10    & 0.050 &       & 150   & 10    & 0.050 & 150   & 10    & 0.050 \\
		&       & 50    & 20    & 0.050 & 50    & 20    & 0.050 &       & 50    & 20    & 0.050 & 50    & 20    & 0.050 \\
		&       & 100   & 20    & 0.050 & 100   & 20    & 0.050 &       & 100   & 20    & 0.050 & 100   & 20    & 0.050 \\
		&       & 150   & 20    & 0.050 & 150   & 20    & 0.050 &       & 150   & 20    & 0.050 & 150   & 20    & 0.050 \\
		\cmidrule{1-1}\cmidrule{3-8}\cmidrule{10-15}    \multirow{6}[2]{*}{1} &       & 50    & 10    & 0.048 & 50    & 10    & 0.056 &       & 50    & 10    & 0.056 & 50    & 10    & 0.090 \\
		&       & 100   & 10    & 0.890 & 100   & 10    & 0.266 &       & 100   & 10    & 0.064 & 100   & 10    & 0.080 \\
		&       & 150   & 10    & 0.950 & 150   & 10    & 0.674 &       & 150   & 10    & 0.112 & 150   & 10    & 0.042 \\
		&       & 50    & 20    & 0.058 & 50    & 20    & 0.142 &       & 50    & 20    & 0.064 & 50    & 20    & 0.048 \\
		&       & 100   & 20    & 0.806 & 100   & 20    & 0.534 &       & 100   & 20    & 0.082 & 100   & 20    & 0.074 \\
		&       & 150   & 20    & 0.924 & 150   & 20    & 0.826 &       & 150   & 20    & 0.096 & 150   & 20    & 0.082 \\
		\cmidrule{1-1}\cmidrule{3-8}\cmidrule{10-15}    \multirow{6}[2]{*}{3} &       & 50    & 10    & 0.866 & 50    & 10    & 0.056 &       & 50    & 10    & 0.146 & 50    & 10    & 0.132 \\
		&       & 100   & 10    & 0.960 & 100   & 10    & 0.746 &       & 100   & 10    & 0.358 & 100   & 10    & 0.216 \\
		&       & 150   & 10    & 0.998 & 150   & 10    & 0.754 &       & 150   & 10    & 0.472 & 150   & 10    & 0.154 \\
		&       & 50    & 20    & 0.966 & 50    & 20    & 0.754 &       & 50    & 20    & 0.164 & 50    & 20    & 0.078 \\
		&       & 100   & 20    & 0.958 & 100   & 20    & 0.900 &       & 100   & 20    & 0.312 & 100   & 20    & 0.228 \\
		&       & 150   & 20    & 0.998 & 150   & 20    & 0.900 &       & 150   & 20    & 0.518 & 150   & 20    & 0.240 \\
		\midrule
		\midrule
		\multirow{6}[2]{*}{-1} & \multirow{24}[8]{*}{\begin{sideways}\textbf{Method: CL-RF}\end{sideways}} & 50    & 10    & 0.056 & 50    & 10    & 0.014 & \multirow{24}[8]{*}{\begin{sideways}\textbf{Method: Matching}\end{sideways}} & 50    & 10    & 0.050 & 50    & 10    & 0.076 \\
		&       & 100   & 10    & 0.056 & 100   & 10    & 0.056 &       & 100   & 10    & 0.064 & 100   & 10    & 0.044 \\
		&       & 150   & 10    & 0.042 & 150   & 10    & 0.060 &       & 150   & 10    & 0.096 & 150   & 10    & 0.100 \\
		&       & 50    & 20    & 0.044 & 50    & 20    & 0.046 &       & 50    & 20    & 0.040 & 50    & 20    & 0.058 \\
		&       & 100   & 20    & 0.076 & 100   & 20    & 0.034 &       & 100   & 20    & 0.046 & 100   & 20    & 0.076 \\
		&       & 150   & 20    & 0.060 & 150   & 20    & 0.046 &       & 150   & 20    & 0.038 & 150   & 20    & 0.062 \\
		\cmidrule{1-1}\cmidrule{3-8}\cmidrule{10-15}    \multirow{6}[2]{*}{0} &       & 50    & 10    & 0.050 & 50    & 10    & 0.050 &       & 50    & 10    & 0.050 & 50    & 10    & 0.050 \\
		&       & 100   & 10    & 0.050 & 100   & 10    & 0.050 &       & 100   & 10    & 0.050 & 100   & 10    & 0.050 \\
		&       & 150   & 10    & 0.050 & 150   & 10    & 0.050 &       & 150   & 10    & 0.050 & 150   & 10    & 0.050 \\
		&       & 50    & 20    & 0.050 & 50    & 20    & 0.050 &       & 50    & 20    & 0.050 & 50    & 20    & 0.050 \\
		&       & 100   & 20    & 0.050 & 100   & 20    & 0.050 &       & 100   & 20    & 0.050 & 100   & 20    & 0.050 \\
		&       & 150   & 20    & 0.050 & 150   & 20    & 0.050 &       & 150   & 20    & 0.050 & 150   & 20    & 0.050 \\
		\cmidrule{1-1}\cmidrule{3-8}\cmidrule{10-15}    \multirow{6}[2]{*}{1} &       & 50    & 10    & 0.044 & 50    & 10    & 0.016 &       & 50    & 10    & 0.078 & 50    & 10    & 0.070 \\
		&       & 100   & 10    & 0.030 & 100   & 10    & 0.062 &       & 100   & 10    & 0.076 & 100   & 10    & 0.084 \\
		&       & 150   & 10    & 0.040 & 150   & 10    & 0.042 &       & 150   & 10    & 0.150 & 150   & 10    & 0.106 \\
		&       & 50    & 20    & 0.038 & 50    & 20    & 0.042 &       & 50    & 20    & 0.068 & 50    & 20    & 0.092 \\
		&       & 100   & 20    & 0.064 & 100   & 20    & 0.040 &       & 100   & 20    & 0.104 & 100   & 20    & 0.070 \\
		&       & 150   & 20    & 0.074 & 150   & 20    & 0.052 &       & 150   & 20    & 0.052 & 150   & 20    & 0.070 \\
		\cmidrule{1-1}\cmidrule{3-8}\cmidrule{10-15}    \multirow{6}[2]{*}{3} &       & 50    & 10    & 0.072 & 50    & 10    & 0.020 &       & 50    & 10    & 0.112 & 50    & 10    & 0.162 \\
		&       & 100   & 10    & 0.056 & 100   & 10    & 0.030 &       & 100   & 10    & 0.216 & 100   & 10    & 0.278 \\
		&       & 150   & 10    & 0.024 & 150   & 10    & 0.044 &       & 150   & 10    & 0.414 & 150   & 10    & 0.382 \\
		&       & 50    & 20    & 0.034 & 50    & 20    & 0.066 &       & 50    & 20    & 0.154 & 50    & 20    & 0.132 \\
		&       & 100   & 20    & 0.064 & 100   & 20    & 0.038 &       & 100   & 20    & 0.206 & 100   & 20    & 0.230 \\
		&       & 150   & 20    & 0.046 & 150   & 20    & 0.030 &       & 150   & 20    & 0.236 & 150   & 20    & 0.294 \\
		\bottomrule
	\end{tabular}%
	\end{minipage}}
	\label{tab:addlabel}%
\end{table}%

\begin{table}[htbp]
	\centering
	\caption{Power of detecting treatment effects under Interaction and Oracle scenarios}
	\label{tab: sim_gbm}
	\vskip-1.5ex
	\hspace*{-1ex}\resizebox{.7\textwidth}{!}{ \begin{minipage}[b]{.99\linewidth}
			\begin{tabular}{c|c|rr||c|c|rr|}
				\toprule
				\textbf{Method} & \textbf{Scenario} & \multicolumn{1}{c}{$\delta$} & \multicolumn{1}{c||}{\textbf{Power}} & \textbf{Method} & \textbf{Scenario} & \multicolumn{1}{c}{$\delta$} & \multicolumn{1}{c|}{\textbf{Power}} \\
				\midrule
				\multirow{8}[4]{*}{\begin{sideways}CA-PPMx\end{sideways}} & \multirow{4}[2]{*}{Interaction} & -1    & 0.079 & \multirow{8}[4]{*}{\begin{sideways}IS-LM\end{sideways}} & \multirow{4}[2]{*}{Interaction} & -1    & 0.085 \\
				&       & 0     & 0.052 &       &       & 0     & 0.052 \\
				&       & 1     & 0.047 &       &       & 1     & 0.083 \\
				&       & 3     & 0.116 &       &       & 3     & 0.497 \\
				\cmidrule{2-4}\cmidrule{6-8}          & \multirow{4}[2]{*}{Oracle} & -1    & 0.077 &       & \multirow{4}[2]{*}{Oracle} & -1    & 0.091 \\
				&       & 0     & 0.053 &       &       & 0     & 0.053 \\
				&       & 1     & 0.084 &       &       & 1     & 0.092 \\
				&       & 3     & 0.132 &       &       & 3     & 0.570 \\
				\midrule
				\multirow{8}[4]{*}{\begin{sideways}CL-RF\end{sideways}} & \multirow{4}[2]{*}{Interaction} & -1    & 0.064 & \multirow{8}[4]{*}{\begin{sideways}Matching\end{sideways}} & \multirow{4}[2]{*}{Interaction} & -1    & 0.108 \\
				&       & 0     & 0.053 &       &       & 0     & 0.050 \\
				&       & 1     & 0.062 &       &       & 1     & 0.121 \\
				&       & 3     & 0.071 &       &       & 3     & 0.575 \\
				\cmidrule{2-4}\cmidrule{6-8}          & \multirow{4}[2]{*}{Oracle} & -1    & 0.061 &       & \multirow{4}[2]{*}{Oracle} & -1    & 0.103 \\
				&       & 0     & 0.053 &       &       & 0     & 0.053 \\
				&       & 1     & 0.039 &       &       & 1     & 0.122 \\
				&       & 3     & 0.043 &       &       & 3     & 0.752 \\
				\bottomrule
			\end{tabular}
	\end{minipage}}
	\caption{Power of the PS-based methods in CAM and MIX scenarios (upper table) and Interaction and Oracle scenarios (lower table)}
	\label{sm tab:ps_sim}
	\hspace*{-5cm}
	\resizebox{.62\textwidth}{!}{ \begin{minipage}[b]{.99\linewidth}
	\begin{tabular}{c|c|llr|llr|c|llr|llr|c|llr|llr|}
		\toprule
		\multicolumn{1}{c|}{\multirow{2}[4]{*}{$\delta$}} &       & \multicolumn{3}{c|}{\textbf{Scenario: CAM}} & \multicolumn{3}{c|}{\textbf{Scenario: MIX}} &       & \multicolumn{3}{c|}{\textbf{Scenario: CAM}} & \multicolumn{3}{c|}{\textbf{Scenario: MIX}} &       & \multicolumn{3}{c|}{\textbf{Scenario: CAM}} & \multicolumn{3}{c|}{\textbf{Scenario: MIX}} \\
		\cmidrule{2-22}    \multicolumn{1}{c|}{} &       & \multicolumn{1}{c}{$n_{1}$} & \multicolumn{1}{c}{$p$} & \multicolumn{1}{c|}{\textbf{Power}} & \multicolumn{1}{c}{$n_{1}$} & \multicolumn{1}{c}{$p$} & \multicolumn{1}{c|}{\textbf{Power}} &       & \multicolumn{1}{c}{$n_{1}$} & \multicolumn{1}{c}{$p$} & \multicolumn{1}{c|}{\textbf{Power}} & \multicolumn{1}{c}{$n_{1}$} & \multicolumn{1}{c}{$p$} & \multicolumn{1}{c|}{\textbf{Power}} &       & \multicolumn{1}{c}{$n_{1}$} & \multicolumn{1}{c}{$p$} & \multicolumn{1}{c|}{\textbf{Power}} & \multicolumn{1}{c}{$n_{1}$} & \multicolumn{1}{c}{$p$} & \multicolumn{1}{c|}{\textbf{Power}} \\
		\midrule
		\multirow{6}[2]{*}{-1} & \multirow{24}[8]{*}{\begin{sideways}\textbf{Method: PP-Logistic}\end{sideways}} & 50    & 10    & 0.082 & 50    & 10    & 0.022 & \multirow{24}[8]{*}{\begin{sideways}\textbf{Method: PP-RF}\end{sideways}} & 50    & 10    & 0.062 & 50    & 10    & 0.022 & \multirow{24}[8]{*}{\begin{sideways}\textbf{Method: CL-Logistic}\end{sideways}} & 50    & 10    & 0.068 & 50    & 10    & 0.016 \\
		&       & 100   & 10    & 0.044 & 100   & 10    & 0.072 &       & 100   & 10    & 0.050 & 100   & 10    & 0.056 &       & 100   & 10    & 0.044 & 100   & 10    & 0.064 \\
		&       & 150   & 10    & 0.036 & 150   & 10    & 0.056 &       & 150   & 10    & 0.048 & 150   & 10    & 0.064 &       & 150   & 10    & 0.036 & 150   & 10    & 0.070 \\
		&       & 50    & 20    & 0.060 & 50    & 20    & 0.044 &       & 50    & 20    & 0.036 & 50    & 20    & 0.026 &       & 50    & 20    & 0.058 & 50    & 20    & 0.040 \\
		&       & 100   & 20    & 0.062 & 100   & 20    & 0.038 &       & 100   & 20    & 0.074 & 100   & 20    & 0.036 &       & 100   & 20    & 0.060 & 100   & 20    & 0.034 \\
		&       & 150   & 20    & 0.060 & 150   & 20    & 0.044 &       & 150   & 20    & 0.066 & 150   & 20    & 0.052 &       & 150   & 20    & 0.048 & 150   & 20    & 0.048 \\
		\cmidrule{1-1}\cmidrule{3-8}\cmidrule{10-15}\cmidrule{17-22}    \multirow{6}[2]{*}{0} &       & 50    & 10    & 0.050 & 50    & 10    & 0.050 &       & 50    & 10    & 0.050 & 50    & 10    & 0.050 &       & 50    & 10    & 0.050 & 50    & 10    & 0.050 \\
		&       & 100   & 10    & 0.050 & 100   & 10    & 0.050 &       & 100   & 10    & 0.050 & 100   & 10    & 0.050 &       & 100   & 10    & 0.050 & 100   & 10    & 0.050 \\
		&       & 150   & 10    & 0.050 & 150   & 10    & 0.050 &       & 150   & 10    & 0.050 & 150   & 10    & 0.050 &       & 150   & 10    & 0.050 & 150   & 10    & 0.050 \\
		&       & 50    & 20    & 0.050 & 50    & 20    & 0.050 &       & 50    & 20    & 0.050 & 50    & 20    & 0.050 &       & 50    & 20    & 0.050 & 50    & 20    & 0.050 \\
		&       & 100   & 20    & 0.050 & 100   & 20    & 0.050 &       & 100   & 20    & 0.050 & 100   & 20    & 0.050 &       & 100   & 20    & 0.050 & 100   & 20    & 0.050 \\
		&       & 150   & 20    & 0.050 & 150   & 20    & 0.050 &       & 150   & 20    & 0.050 & 150   & 20    & 0.050 &       & 150   & 20    & 0.050 & 150   & 20    & 0.050 \\
		\cmidrule{1-1}\cmidrule{3-8}\cmidrule{10-15}\cmidrule{17-22}    \multirow{6}[2]{*}{1} &       & 50    & 10    & 0.058 & 50    & 10    & 0.030 &       & 50    & 10    & 0.072 & 50    & 10    & 0.018 &       & 50    & 10    & 0.060 & 50    & 10    & 0.020 \\
		&       & 100   & 10    & 0.040 & 100   & 10    & 0.054 &       & 100   & 10    & 0.036 & 100   & 10    & 0.054 &       & 100   & 10    & 0.036 & 100   & 10    & 0.058 \\
		&       & 150   & 10    & 0.028 & 150   & 10    & 0.042 &       & 150   & 10    & 0.044 & 150   & 10    & 0.034 &       & 150   & 10    & 0.030 & 150   & 10    & 0.048 \\
		&       & 50    & 20    & 0.034 & 50    & 20    & 0.050 &       & 50    & 20    & 0.028 & 50    & 20    & 0.024 &       & 50    & 20    & 0.028 & 50    & 20    & 0.048 \\
		&       & 100   & 20    & 0.040 & 100   & 20    & 0.038 &       & 100   & 20    & 0.068 & 100   & 20    & 0.032 &       & 100   & 20    & 0.036 & 100   & 20    & 0.040 \\
		&       & 150   & 20    & 0.062 & 150   & 20    & 0.054 &       & 150   & 20    & 0.080 & 150   & 20    & 0.060 &       & 150   & 20    & 0.064 & 150   & 20    & 0.048 \\
		\cmidrule{1-1}\cmidrule{3-8}\cmidrule{10-15}\cmidrule{17-22}    \multirow{6}[2]{*}{3} &       & 50    & 10    & 0.084 & 50    & 10    & 0.018 &       & 50    & 10    & 0.072 & 50    & 10    & 0.020 &       & 50    & 10    & 0.072 & 50    & 10    & 0.026 \\
		&       & 100   & 10    & 0.040 & 100   & 10    & 0.038 &       & 100   & 10    & 0.062 & 100   & 10    & 0.026 &       & 100   & 10    & 0.044 & 100   & 10    & 0.028 \\
		&       & 150   & 10    & 0.038 & 150   & 10    & 0.050 &       & 150   & 10    & 0.028 & 150   & 10    & 0.044 &       & 150   & 10    & 0.034 & 150   & 10    & 0.056 \\
		&       & 50    & 20    & 0.052 & 50    & 20    & 0.064 &       & 50    & 20    & 0.038 & 50    & 20    & 0.042 &       & 50    & 20    & 0.048 & 50    & 20    & 0.064 \\
		&       & 100   & 20    & 0.042 & 100   & 20    & 0.042 &       & 100   & 20    & 0.060 & 100   & 20    & 0.030 &       & 100   & 20    & 0.040 & 100   & 20    & 0.038 \\
		&       & 150   & 20    & 0.040 & 150   & 20    & 0.028 &       & 150   & 20    & 0.058 & 150   & 20    & 0.040 &       & 150   & 20    & 0.038 & 150   & 20    & 0.038 \\
		\bottomrule
	\end{tabular}%
\end{minipage}}
%
\vskip2.5ex
	\vskip-1.5ex
	\resizebox{.65\textwidth}{!}{ \begin{minipage}[b]{.99\linewidth}
	\begin{tabular}{c|c|rr|c|c|rr|}
		\toprule
		\textbf{Method} & \textbf{Scenario} & \multicolumn{1}{c}{$\delta$} & \multicolumn{1}{c||}{\textbf{Power}} & \textbf{Method} & \textbf{Scenario} & \multicolumn{1}{c}{$\delta$} & \multicolumn{1}{c|}{\textbf{Power}} \\
		\midrule
		\multirow{8}[4]{*}{\begin{sideways}CL-Logistic\end{sideways}} & \multirow{4}[2]{*}{Interaction} & -1    & 0.060 & \multirow{8}[4]{*}{\begin{sideways}PP-Logistic\end{sideways}} & \multirow{4}[2]{*}{Interaction} & -1    & 0.058 \\
		&       & 0     & 0.052 &       &       & 0     & 0.052 \\
		&       & 1     & 0.058 &       &       & 1     & 0.058 \\
		&       & 3     & 0.062 &       &       & 3     & 0.058 \\
		\cmidrule{2-4}\cmidrule{6-8}          & \multirow{4}[2]{*}{Oracle} & -1    & 0.065 &       & \multirow{4}[2]{*}{Oracle} & -1    & 0.063 \\
		&       & 0     & 0.053 &       &       & 0     & 0.053 \\
		&       & 1     & 0.036 &       &       & 1     & 0.036 \\
		&       & 3     & 0.043 &       &       & 3     & 0.045 \\
		\midrule
		\multirow{4}[2]{*}{\begin{sideways}PP-RF\end{sideways}} & \multirow{4}[2]{*}{Oracle} & -1    & 0.061 & \multirow{4}[2]{*}{\begin{sideways}PP-RF\end{sideways}} & \multirow{4}[2]{*}{Interaction} & -1    & 0.066 \\
		&       & 0     & 0.053 &       &       & 0     & 0.053 \\
		&       & 1     & 0.039 &       &       & 1     & 0.057 \\
		&       & 3     & 0.043 &       &       & 3     & 0.064 \\
		\bottomrule
	\end{tabular}%
	\end{minipage}}
\end{table}

\subsection{Multiple Historical Controls}
\label{sm subsec:mult_hist}
We consider a setup with historical controls arising from multiple
sources, i.e., with $S>2$. 
As mentioned earlier in Section \ref{subsec:CAM_cov}, we merge the
historical datasets and treat the merged data set as a single RWD population with increased
heterogeneity. 
We study the performance of the CA-PPMx model in this scenario via simulation studies.
We extend the MIX scenario discussed in Section \ref{sec:sim_study}.
We generate the treatment arm $\bX_{1,i}\simiid \sum_{j=1}^{k} \pi_{1,j} \mn_{p}(\bmu_{j}, \sigma^{2} \bI_ {p})$.
We generate two RWD datasets from $\bX_{2,i}\simiid \sum_{j=1}^{k-1}
\pi_{2,j} \mn_{p}(\bmu_{j}, \sigma^{2} \bI_ {p})$ and
$\bX_{3,i}\simiid \sum_{j=2}^{k} \pi_{3,j} \mn_{p}(\bmu_{j},
\sigma^{2} \bI_ {p})$. 
In this construction, the historical populations $\bXtw$ and $\bX_{3}$
are substantially different, with
one distinct atom each, as well as varying weights for the common atoms.
Letting $\bX_{2'}$ denote the merged $\bXtw$ and $\bX_{3}$ population, we fit the CA-PPMx model on $\bXon$ and $\bX_{2'}$.
Note that the current trial population $\bXon$ has an extra atom compared to each of the RWD populations but the merged $\bX_{2'}$ and $\bXon$ share common atoms.

We generate the response
$Y_{1,i}\simind \mn( \delta+ \bX_{1,i}\trans \bbeta , 1 )$ and
$Y_{s,i} \simind \mn(  \bX_{s,i}\trans \bbeta , 1 )$ for $s=2,3$ implying $\delta$ to be the true treatment effect. 
We let $n_{2}$, $n_{2}$ and $n_{3}$ denote the sample sizes in the three populations, respectively where we set $n_{2}=n_{3}= 3\times n_{2}$ in coherence with the simulation studies in Section \ref{sec:sim_study}.
We set the dimension of the covariates $p=10$ and repeat the the experiments for $\delta=-1,0,1,3$ and $n_{2}=50, 100, 150$.
We plot the power of discovering the treatment effect in Figure
\ref{sm fig:mult_RWD} {calculated in the exact same manner as described in Section \ref{sec:sim_study}}.
We observe that the power increases with respect to both sample size and strength of the treatment effect.

\begin{figure}[!h]
	\centering
	\includegraphics[width=.7\linewidth]{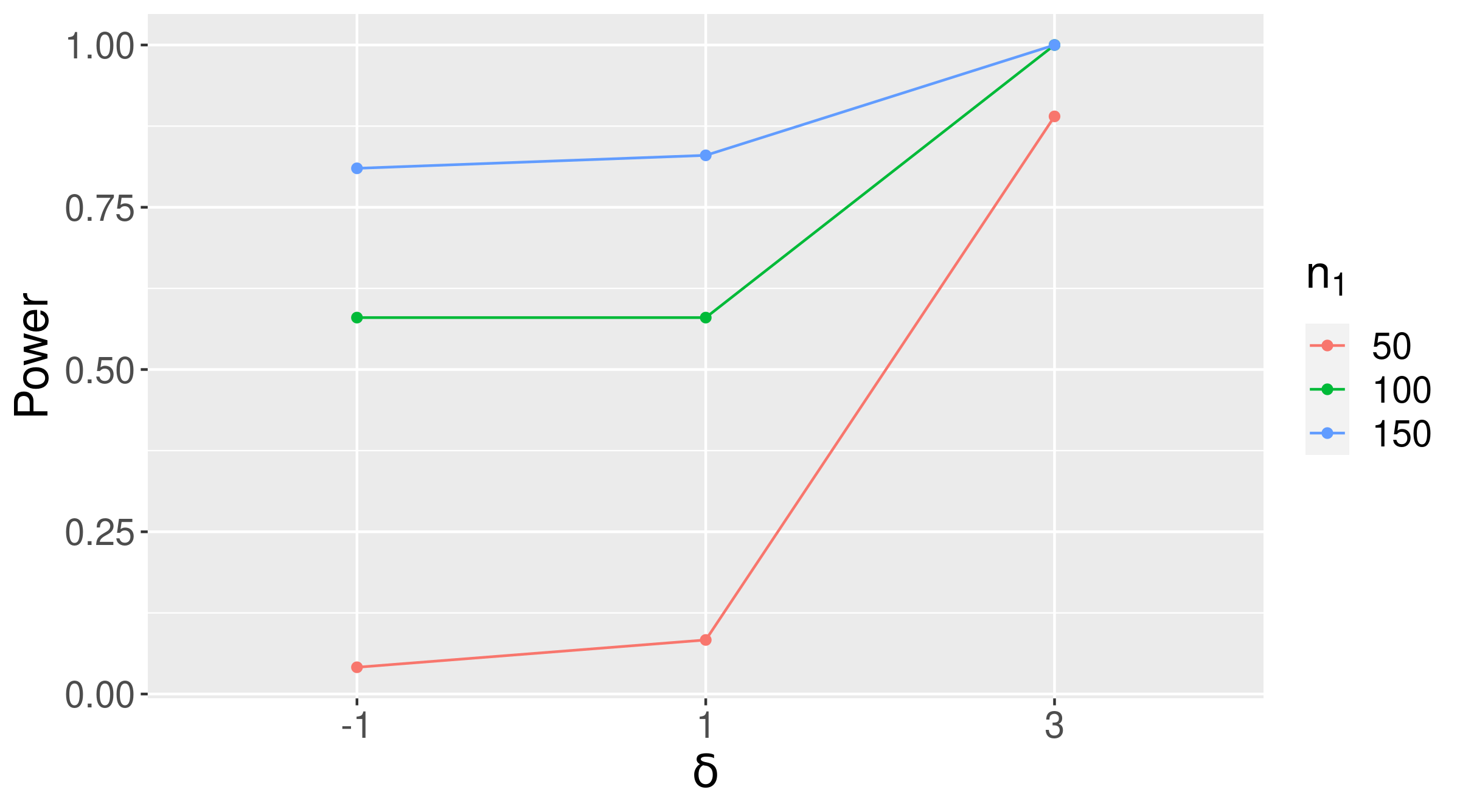}
	\caption{Multiple historical data in the CA-PPMx model: We combine different historical datasets and combine them as a more heterogeneous single population and subsequently fit the CA-PPMx model.
	We observe that the power increases with respect to both sample size and strength of the treatment effect.}
	\label{sm fig:mult_RWD}
\end{figure}



\subsection{Effect of Missing Confounders}
\label{sm subsec:missing_confounder}
In this section we briefly study the effect of missing confounders on
inference under the proposed CA-PPMx model.
In particular we consider the case where a confounding factor is completely unobserved.
In such cases causal inference methods are often biased; see \citelatex{missing_confounder} and the references therein for a detailed review.
However, in many applications, multivariate covariates are often correlated among each other. 
Several imputation methods for partially observed confounders are based on this assumption \citeplatex{missing_confounder2,missing_confounder3}.
In such cases, observing and using another covariate which is correlated to the missing confounder as predictor can reduce bias.
We study this in a simulated example.

We consider a regression setup in a case-control study
$(\bX_{s,i},Y_{s,i})$,  $i=1,\dots,\ns$, $s=1,2$ with bivariate
covariate $\bX_{s,i}=(X_{s,i,1}, X_{s,i,2})\trans$. 
First, we generate $X_{s,i,1} \sim \sum_{j=1}^{k} \pi_{s,j} \mn(\mu_{j},0.01)$ and subsequently generate $X_{s,i,2}= m X_{s,i,1} +\varepsilon_{s,i}$ where $\varepsilon_{s,i} \simiid \mn(0,1)$ and $m\in \rR$.
Then, we generate the responses $Y_{1,i}= \delta + \beta X_{1,i,1}+ \epsilon_{1,i}$ and $Y_{2,i}=  \beta X_{2,i,1}+\epsilon_{2,i}$ where $\epsilon_{s,i}\simiid \mn(0,1)$ implying $\delta$ to be the true treatment effect.
Thus conditionally on the $X_{s,i,1}$'s, the responses $Y_{s,i}$'s are independent of the $X_{s,i,2}$'s.
We take $n_{2}=50$, $n_{2}=300$, $k=3$, $(\mu_{2},\mu_{2},\mu_{3})=(-3,0,3)$, $\delta=3$ and $\beta=1$.
We repeat the simulation experiment independently 100 times and randomly generate the $\pi_{s,j}$'s in each replicate.

\begin{figure}[!ht]
	\centering
	\includegraphics[width=0.65\linewidth]{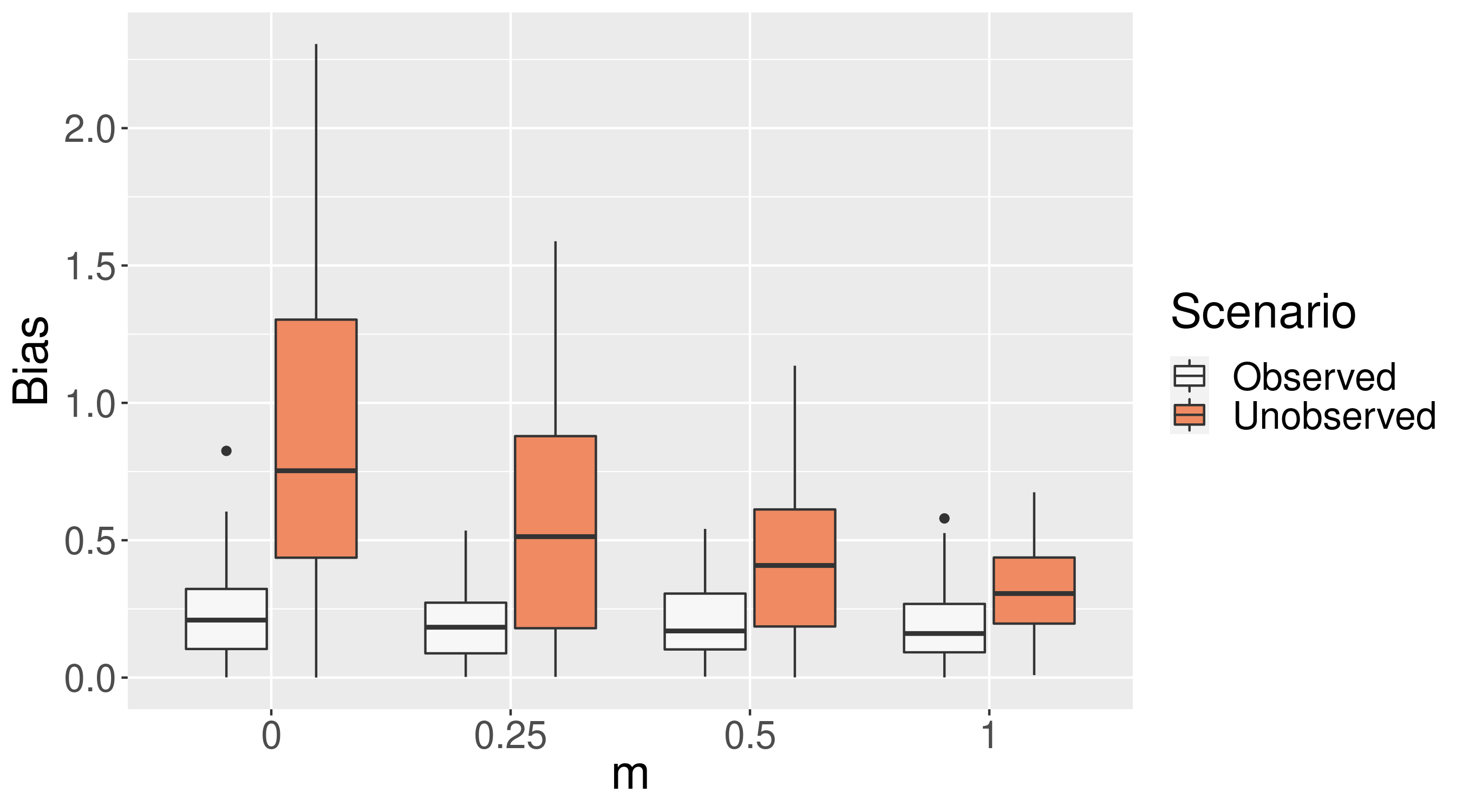}
	\caption{Effect of missing confounder in the CA-PPMx model: The bias in estimating the treatment effect decreases as the correlation between the observed covariate and the unobserved confounder increases.}
	\label{sm fig:compare_missing}
\end{figure}

We consider two analysis scenarios: 
(1) \textit{Unobserved}: $X_{s,i,1}$ is assumed to be unobserved and
the CA-PPMx model is fitted using $(X_{s,i,2},Y_{s,i})$; 
(2) \textit{Observed}: the CA-PPMx model is fitted using $(\bX_{s,i},Y_{s,i})$.
We compute the bias in estimating the treatment effect $\delta$ for varying values of $m$ in both scenarios.
We show boxplots of the biases over the repeat simulations in Figure \ref{sm fig:compare_missing}.

Note that for $m=0$, $X_{s,i,1}$ and $X_{s,i,2}$ are uncorrelated.
Additionally, $\abs{\corr(X_{s,i,1}, X_{s,i,2})}$ is an increasing function of $\abs{m}$.
Coherently, the bias is maximum in the \textit{Unobserved} scenario for $m=0$ as the $X_{s,i,2}$'s carry no information regarding the confounding factor $X_{s,i,1}$'s.
The marginal correlation between the observed covariate and the response increases with $m$ and accordingly we see a reduction in the bias.
This simulation study indicates that the CA-PPMx method will not yield
terribly biased results
as long as the data includes observed covariates that are correlated to the unmeasured confounder.

\subsection{Computation Times for the CA-PPMx Method}
\label{sm subsec:computatiokn_cost}
In this section we report computation times of the MCMC sampler proposed in Section \ref{SM subsec:MCMC} across different sample sizes and covariate dimensions.
We consider the CAM and MIX scenarios and the exact same simulation setups discussed in Section \ref{sec:sim_study} of the main paper.
Since the model implementation times do not depend on the treatment effect size, 
we report the computation times for $\delta=3$ only.
Computation times for $6,000$  MCMC iterations in seconds for a single repeat simulation on an Intel Core i9-13900K CPU with 128GB of RAM are provided in Figure \ref{sm fig:mcmctime} where we see that the computational cost increases with the covariate dimension $p$ as well as the sample size $n_{1}$.

\begin{figure}[H]
	\centering
	\includegraphics[width=0.7\linewidth]{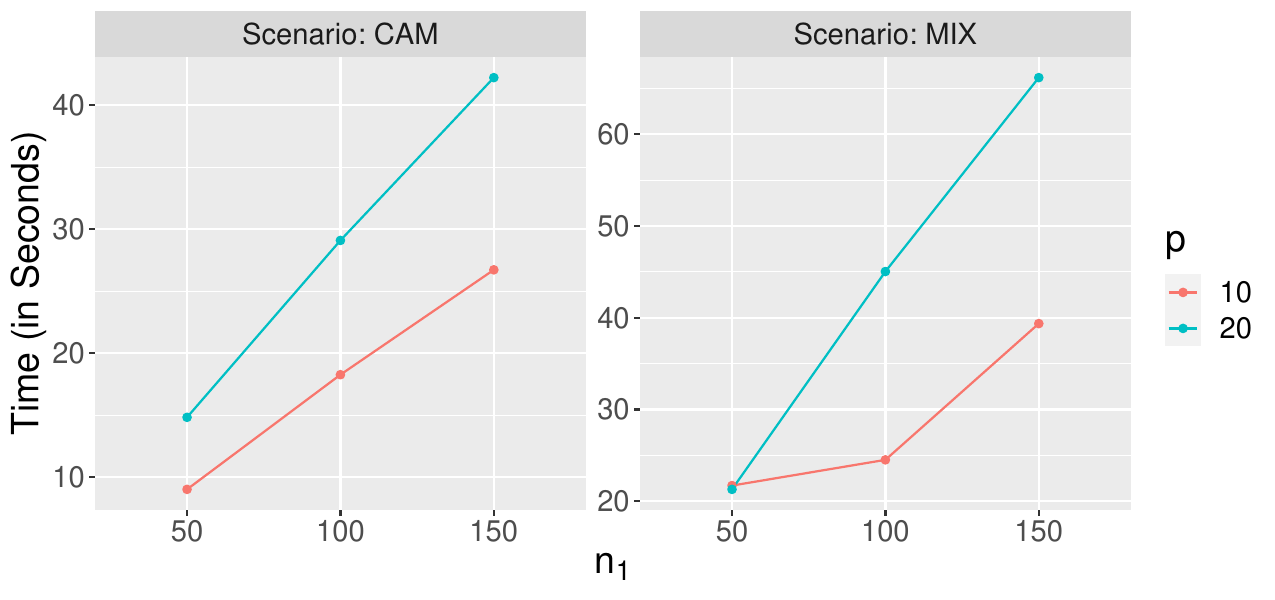}
	\vskip-2ex
	\caption{Computation times of the MCMC sampler in seconds: $n_{1}$ and $p$ denotes the number of patients in the current trial arm and the dimension of the covariates, respectively.}
	\label{sm fig:mcmctime}
\end{figure}

\section{MCMC Diagnostics}
\label{sm subsec:MCMC_diagnose}
In this section, we provide some convergence diagnostics of the MCMC sampler discussed in Section \ref{SM subsec:MCMC} for one trial replicate discussed in Section \ref{sec:GBM}. 
We show traceplots and Geweke's convergence diagnostics
\citeplatex{geweke_mcmc} for some selected parameters, using an implementation in
the \texttt{ggmcmc R} package \citeplatex{ggmcmc}. 

Recall the importance resampling weights $w_{i}\propto \frac{\pi_{1,c_{2,i}}}{n_{2,c_{2,i}}}$ in Eqn \eqref{eq:ref_needs} attached to the historical patients.
We evaluate MCMC convergence diagnostics for the five $w_{i}$'s with the largest posterior means,
the lognormal hyperparameters $\mu_{0}$ and $b_{0}$ mentioned in Step 3 and the Dirichlet mixture hyperparameters $\alpha_{1}$ and $\alpha_{2}$ in Step 5 of the MCMC sampler in Section \ref{SM subsec:MCMC}.
The results, provided in Figure \ref{sm fig:MCMC_lognorm_hyperparams}, do not suggest any convergence or mixing issues.

\begin{figure}[H]
	\centering
	\begin{subfigure}[b]{\linewidth}
		\includegraphics[trim={0cm 0cm 1.5cm 0.15cm},clip,width=\linewidth,height=.475\paperheight]{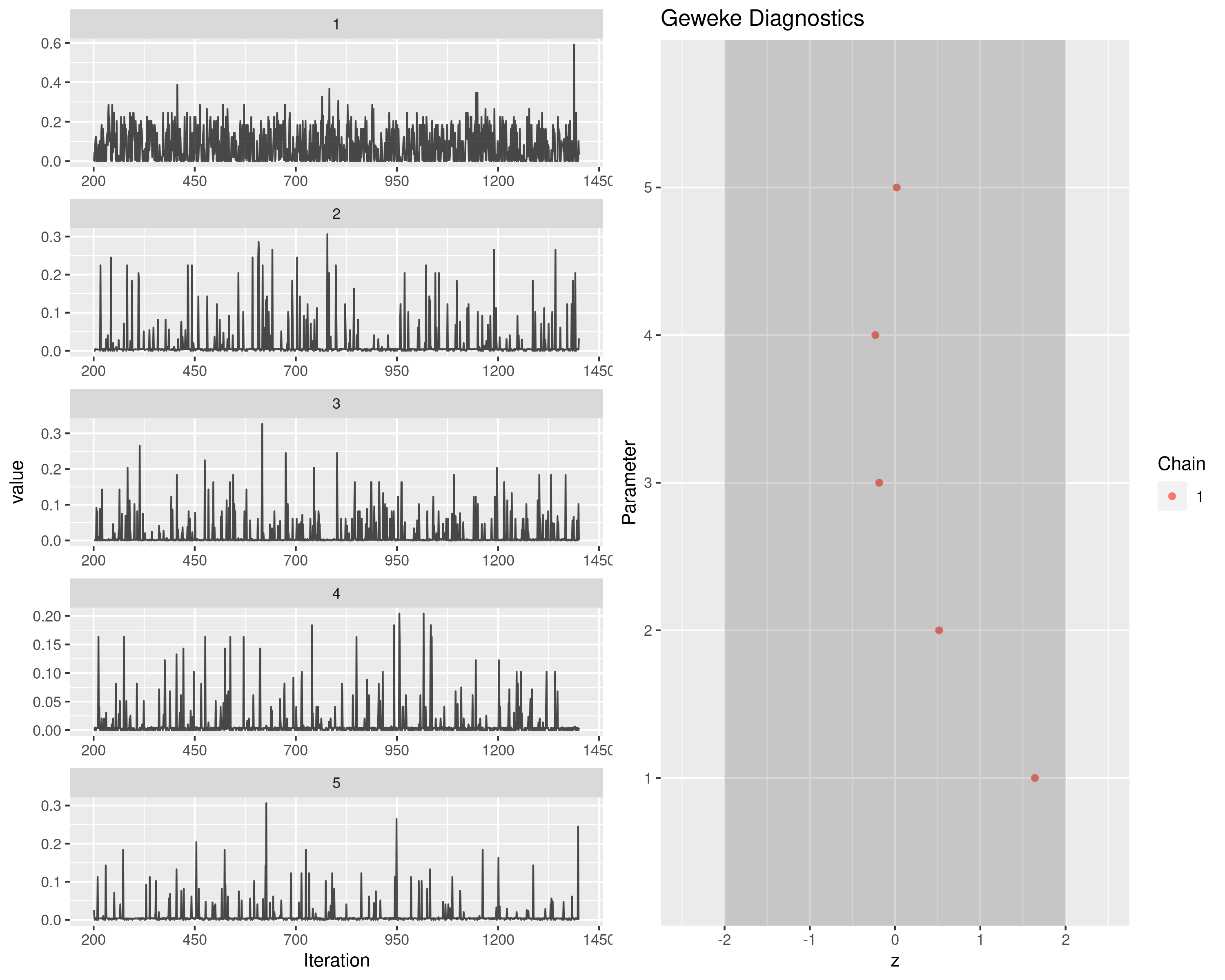}
		\vskip-1.5ex
		\caption{MCMC analysis of the largest 5 importance resampling weights.}
		\label{sm fig:MCMC_importance_weights}
	\end{subfigure}
%
	\begin{subfigure}[b]{\linewidth}
		\includegraphics[trim={0cm 0cm 1.5cm 0.15cm},clip,width=\linewidth]{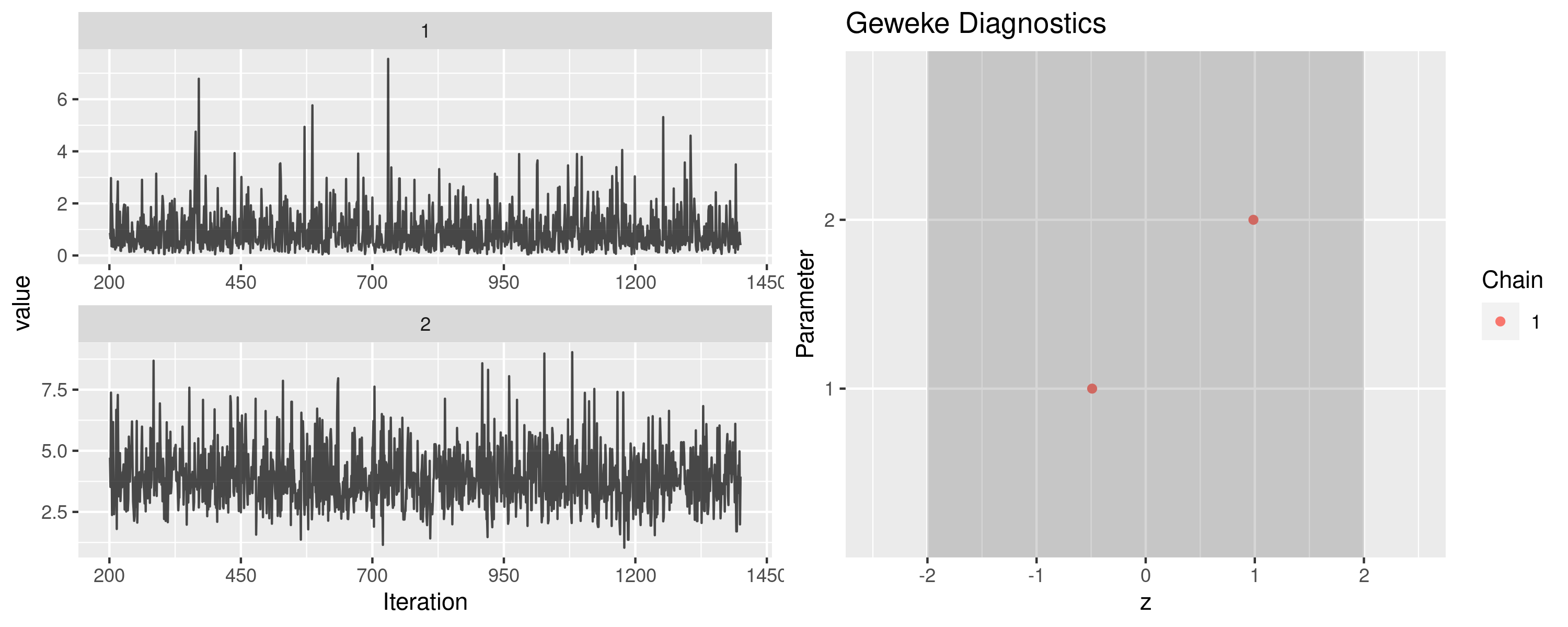}
		\vskip-1.5ex
		\caption{MCMC analysis of the Dirichlet mixture hyperparameters $\alpha_{2}$ and $\alpha_{2}$.}
		\label{sm fig:MCMC_dir_params}
	\end{subfigure}
\end{figure}

\begin{figure}[H]
\ContinuedFloat	
\begin{subfigure}[b]{\linewidth}
	\includegraphics[trim={0cm 0cm 1.5cm 0.2cm},clip,width=\linewidth]{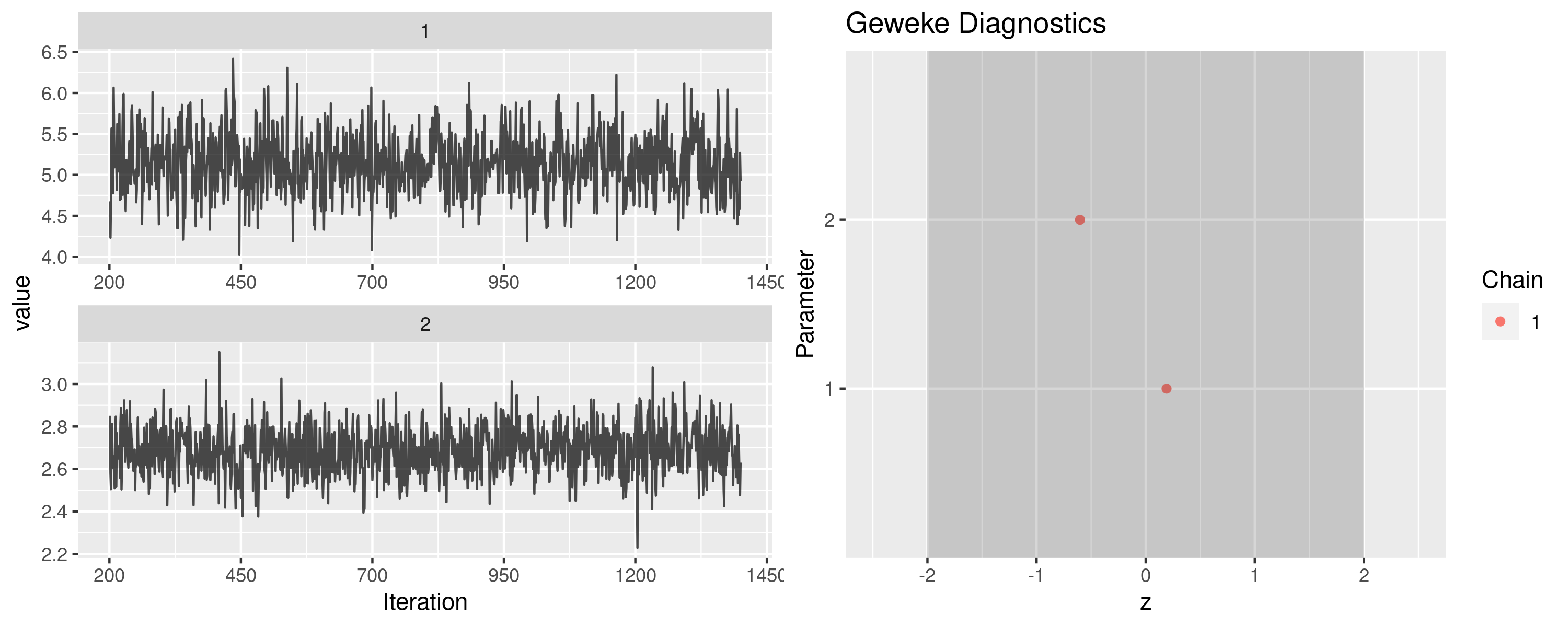}
	\vskip-1.5ex
	\caption{MCMC analysis of the lognormal hyperparameters $\mu_{0}$ and ${b}_{0}$.}	    
\end{subfigure}
	\caption{MCMC convergences diagnostics for some selected parameters: 
		Panel (a), (b) and (c) shows results for the top five
                $w_{i}\propto \frac{\pi_{1,c_{2,i}}}{n_{2,c_{2,i}}}$
                with largest posterior means,		  
		the lognormal hyperparameters $\mu_{0}$ and $b_{0}$ and 
		the Dirichlet mixture hyperparameters $\alpha_{1}$ and $\alpha_{2}$ in Step 5, respectively.
		In each panel, we show the corresponding traceplots
                across the thinned out MCMC samples on the left, and
                Geweke's diagnostics on the right.} 
	\label{sm fig:MCMC_lognorm_hyperparams}	
\end{figure}

\bibliographystylelatex{natbib}
\bibliographylatex{refs}
\end{document}